\documentclass[11pt]{article}
\usepackage[T1]{fontenc}
\usepackage{lmodern,amsmath,amssymb,amsthm,mathtools}
\usepackage[margin=1.05in,includefoot]{geometry}
\usepackage[colorlinks=true,linkcolor=blue,citecolor=blue,urlcolor=blue]{hyperref}
\hypersetup{
  pdftitle={Quantum Leakage from Shamir Sharing with Correlated Resources},
  pdfauthor={Isaac Barouch Essayag; Aryeh Lev Zabokritskiy (Yohananov)},
  pdfkeywords={Shamir secret sharing; quantum leakage; quantum side information; primitive quantum channels; matrix product states},
  pdflang={en-US}
}
\newtheorem{theorem}{Theorem}[section]
\newtheorem{lemma}[theorem]{Lemma}
\newtheorem{proposition}[theorem]{Proposition}
\newtheorem{corollary}[theorem]{Corollary}
\theoremstyle{definition}
\newtheorem{definition}[theorem]{Definition}
\theoremstyle{remark}
\newtheorem{remark}[theorem]{Remark}
\newcommand{\F}{\mathbb F}
\newcommand{\E}{\mathbb E}
\newcommand{\Tr}{\operatorname{Tr}}
\newcommand{\supp}{\operatorname{supp}}
\newcommand{\rank}{\operatorname{rank}}
\newcommand{\id}{\operatorname{id}}
\newcommand{\Adv}{\operatorname{Adv}}
\newcommand{\Dmax}{D_{\max}}
\newcommand{\norm}[1]{\left\lVert #1\right\rVert_1}
\title{Quantum Leakage from Shamir Sharing with Correlated Resources}
\author{Isaac Barouch Essayag$^{1,2}$ \and Aryeh Lev Zabokritskiy (Yohananov)$^{1,2}$}
\date{}
\begin{document}
\maketitle
\begin{center}
\small
$^1$MIGAL--Galilee Research Institute, Kiryat Shmona 1101602, Israel\\
$^2$Tel-Hai University of Kiryat Shmona in the Galilee, Israel\\
\texttt{isaac.es@migal.org.il}, \texttt{yuhanalev@telhai.ac.il}\\
Corresponding author: Aryeh Lev Zabokritskiy (Yohananov)\\
ORCID: \href{https://orcid.org/0000-0003-3151-6192}{0000-0003-3151-6192}
\end{center}
\begin{abstract}
We bound local quantum leakage from Shamir sharing uniformly over prime fields and permitted share-dependent channels. The adversary receives all outputs and a reference to a resource independent of the secret and sharing randomness.
For independent device--reference pairs at fixed output dimension, we determine the exact exponent of the worst-case trace distance for two chosen secrets at every fixed reconstruction rate above one half, including its dependence on entanglement within each pair. For unentangled one-qubit leakage, this improves the sufficient rate of about $0.7334$ proved by Batra et al.
Above one half, we also determine the exact exponent when an arbitrarily entangled support is tensored with single-device states elsewhere. Every fixed support fraction below the reconstruction rate gives exponential security, while a reconstruction set admits a known attack.
Exponential security extends to fixed finite-memory processes with primitive transfer channels, fixed stationary chains satisfying upper operator mixing, and canonical primitive matrix product states, allowing arbitrary finite adversary extensions of the device marginal.
For one-qubit outputs, fully separable resources and independent pairs whose joint output--reference states are positive under partial transpose for every share admit rates slightly below one half.
\end{abstract}
\noindent\textbf{Keywords:} Shamir secret sharing; quantum leakage; quantum side information; primitive quantum channels; matrix product states.

\section{Introduction}
Shamir sharing distributes a classical secret among $n$ devices so that any $t$ shares reconstruct it and fewer than $t$ reveal no information \cite{Shamir}. Leakage resilience asks what an adversary can learn when every device releases a small output computed from its share \cite{BDIR}.
We study this question over prime fields when each device applies a share-dependent quantum channel to its own register. The devices act without communication or feedback on a resource independent of the secret and sharing randomness. The adversary receives every output and a retained quantum reference.
Our security measure is the trace distance between the adversary's states for any two chosen secrets. How does this distance depend on the reconstruction rate $t/n$, the output dimension and the correlations in the resource?

Batra, Kitagawa, Nishimaki, and Yamakawa \cite{BKNY} proved exponential security above a reconstruction rate of about $0.7334$ for unentangled one-qubit leakage with a trivial reference. They also showed that a Greenberger--Horne--Zeilinger (GHZ) state on $t$ devices can yield constant distinguishing advantage.
They asked how far the unentangled rate can be lowered, how many devices can share an arbitrary entangled state, and whether the structure of the correlations improves on bounds obtained by revealing those devices' shares. Classical leakage supplies a benchmark for the first question. Sun and Wootters \cite{SW} proved security for any fixed number of classical bits per share at every fixed rate above one half. Quantum-proof multi-source extraction also treats quantum side information, under entropy assumptions and initial source independence \cite{CLW} or conditional independence given the side information \cite{AFPS}. Here the devices hold correlated shares of a fixed secret.

We first allow arbitrary entanglement between each device and its own part of the adversary's reference, with independence across these device--reference pairs. Fix the sharing parameters and two distinct secrets, and write $L$ for the supremum of their trace distance over the resources and channels in the class under consideration.
For fixed output dimension $d$ and $t/n\to R>1/2$, we prove
\[
 \lim_{n\to\infty}-\frac1n\log L
 =R\log\frac1{\kappa(K)},\qquad
 \kappa(K)\coloneqq \frac K\pi\sin\frac\pi K.
\]
Here $K=d^2$ for unrestricted pairs and $K=dr$ when each pair has Schmidt number at most $r\le d$. The Schmidt number is the smallest bound on Schmidt rank in a pure-state decomposition. For $d$ classical output symbols, $K=d$ even with arbitrary finite references.
The bounds hold uniformly over prime fields, evaluation points and distinct secret pairs. Local attacks attain the stated exponent for each such pair.
For a single classical bit with trivial reference, the exponent was already implied by Sun and Wootters \cite[Corollary~4.3]{SW} and the additive-sharing lower bound of Maji et al.\ \cite[Theorems~1 and~10]{Maji}.

The same independent-pair model permits a growing output budget. At $t/n\ge1/2+\eta$ for fixed $0<\eta<1/2$, each device may emit up to $\lfloor(\eta/200)\log_2 n\rfloor$ qubits with stretched-exponentially small distance. For unrestricted pairs and output dimension $d\to\infty$ in this range, $-\log L=(1+o(1))\pi^2t/(6d^4)$.
Nguyen's logarithmic-leakage bound concerns classical outputs \cite[Theorems~2 and~6]{Nguyen}. Here the outputs are quantum, the adversary retains the references, and the leading exponent is sharp.

For one-qubit outputs, we answer the first question of Batra et al.\ in part by proving exponential security at $t/n\ge1/2-5\cdot10^{-10}$ when the resource is fully separable across all devices and the adversary reference.
The same rate holds for independent families of joint output--reference states $\sigma_i(y)$ satisfying $\sigma_i(y)\le2\gamma_i$ for every share $y$, with a fixed density operator $\gamma_i$ at each device.
This includes independent qubit--reference outputs positive under partial transpose for every share and classical-bit outputs with arbitrary independent references.
For a classical bit with trivial reference, Sun and Wootters \cite{SW} crossed one half by an unspecified constant and Nguyen gave the explicit sufficient rate $1/2-10^{-14}$ \cite[Theorem~1 and Corollary~6]{Nguyen}.
Security at every positive rate for arbitrary one-qubit channels remains open, even with trivial references.

The second question concerns entanglement among the devices. Suppose an arbitrary state on a set $S$ of devices and the reference is tensored with individual states outside $S$, as in \cite[Definition~9]{BKNY}.
At $t/n\to R>1/2$ and $|S|/n\to\alpha<R$, the normalized exponent is $(R-\alpha)\log(1/\kappa(d))$. If only the full device marginal has this factorization and the adversary holds an arbitrary finite extension, the exponent is $(R-\alpha)\log(1/\kappa(d^2))$.
These results settle the fixed-fraction form of the question above one half in both models. The earlier sufficient support fraction was about $(R-0.7334)/(1-0.7334)$ for $R>0.7334$ \cite{BKNY}. Every fixed fraction below $R$ now gives exponential security, and a support of $t$ devices admits the known GHZ attack.
For a support of size $t-\ell$ with $1\le\ell<t$, our attacks improve the partial-GHZ lower bound of order $2^{-\ell}$ \cite[Section~6.3]{BKNY} to a constant times $\kappa(K)^\ell$, with $K=d$ or $d^2$ in the respective models.

A weaker assumption specifies only a marginal on some of the devices. If a marginal on at least a fixed fraction $b_0$ of them is bounded above in operator order by a fixed multiple of a product of single-device states, every fixed rate $R>1-b_0/2$ gives exponential security for fixed output dimension and arbitrary finite adversary extensions.
Thus when only the devices outside $S$ have a product marginal, security holds for $\alpha<2R-1$. The interval $2R-1\le\alpha<R$ remains open under this weaker hypothesis.
The marginal bounds give the same sufficient rate $R>1-b_0/2$ for a network of independent edge sources with a fraction $b_0$ of pairwise nonadjacent devices. They also give security above $3/4$ for graph states with a product marginal on half the vertices \cite{Hein} and, with exponentially high probability, for a Haar-random pure state of the $n$ device qubits \cite{HLW}. The random-state guarantee covers strategies chosen with knowledge of that state.

Correlations may also extend throughout the devices. At fixed output dimension, we prove exponential security at every fixed rate above one half for successive emissions from a fixed finite memory with primitive transfer channel, for fixed stationary chains satisfying upper operator mixing \cite{Hiai,HMO}, and for canonical matrix product states generated by fixed matrices with a primitive transfer channel \cite{Sanz}.
The guarantee allows every finite extension of the device marginal, uniformly over the initial memory state or matrix product boundary vectors as applicable. Its exponent and length cutoff may depend on the fixed process and the rate margin.
One fixed two-dimensional memory has correlations at every separation and satisfies the theorem. A family of two-dimensional primitive memories varying with $n$ admits constant distinguishing advantage at full threshold. These examples relate to the third question of Batra et al.\ by identifying which restrictions on a preparation process suffice.

For independent blocks of bounded size, we obtain the same above-half guarantee with arbitrary block references and a separate share-dependent channel at each device. The two-device blocks of Batra et al.\ \cite{BKNY} permit joint access to the shares within each block, with a trivial outside reference.
For resources prepared from a product device marginal by disjoint gates of fixed arity acting on complete device register groups, logarithmic depth is necessary and sufficient in order of growth for constant distinguishing advantage at full threshold with one-qubit outputs, when all causal influence passes through the gates.

Section~\ref{sec:model} defines the leakage model. Sections~\ref{sec:ea} and~\ref{sec:sharp} give the independent-pair bounds, exact exponents and growing output budgets. Section~\ref{bh:section} treats one-qubit leakage below one half. Sections~\ref{sec:correlation} and~\ref{sec:mixing} treat selected marginals, entangled supports, mixing chains and finite quantum memory. Section~\ref{sec:conclusion} states the remaining open questions.

\section{Model and notation}\label{sec:model}
We use Shamir's polynomial sharing scheme \cite{Shamir}. Throughout, $n,t$ are positive integers with $1\le t\le n$, $p>n$ is prime, and the evaluation points $x_1,\ldots,x_n$ are distinct nonzero elements of $\F_p$. Write $x\coloneqq(x_1,\ldots,x_n)$.
To share $s\in\F_p$, sample uniformly a polynomial $f$ of degree less than $t$ with $f(0)=s$, and give $Y_i\coloneqq f(x_i)$ to device $i$.
Any set $T$ of $t$ devices is a reconstruction set: Lagrange interpolation gives nonzero reconstruction coefficients $\lambda_i$ with $s=\sum_{i\in T}\lambda_iY_i$.

An initial state $\Omega_{A_1\cdots A_n E}$ is independent of the secret and sharing randomness.
For each possible share $y$, device $i$ has a fixed completely positive trace-preserving map
$\Phi_i^y:A_i\longrightarrow B_i$, where $\dim B_i=2$ unless a result specifies a larger output dimension.
An output dimension is fixed as $n$ grows except in Section~\ref{sec:growing-output}, where its dependence on $n$ is stated explicitly.
For classical-bit leakage, each $\Phi_i^y$ has a fixed classical output register: its joint output with any reference is block diagonal in the designated basis of $B_i$.
All spaces are finite dimensional, and ancillary dimensions may be arbitrarily large.
The attacker receives $B_1\cdots B_nE$ and may process them jointly.
There is no communication between devices or feedback from the attacker during leakage.
Writing $\rho_s$ for the state received by the attacker when the shared secret is $s$, our security measure is
\[
 \Adv(\rho,\sigma)\coloneqq \frac12\norm{\rho-\sigma}.
\]
Every security statement is for every pair of chosen secrets and every permitted family of maps.
A \emph{public seed} is a classical random variable, independent of the secret and the sharing randomness, whose value selects the initial state and the fixed local maps of one branch of the experiment.
Such a seed may be conditioned on, provided each branch, as a complete conditional experiment including the adversary's reference, satisfies the claimed hypotheses.
For a public-seed experiment, we average the conditional leakage bounds over the seed; the same bound also holds when the seed is withheld, by contractivity under discarding the classical label.
A uniform exponential estimate requires constants independent of the seed value.
Every bound below that is uniform over its stated class therefore extends to such experiments, provided each branch keeps the standard Shamir sampling law and a resource independent of the secret and sharing randomness.

We write $[n]\coloneqq \{1,\ldots,n\}$ and $\mathcal D(H)$ for the density operators on a finite-dimensional Hilbert space $H$.
The order $X\le Y$ means that $Y-X$ is positive semidefinite.
For an operator $X$, $\|X\|_1\coloneqq \Tr\sqrt{X^\dagger X}$ is the trace norm and $\|X\|_2\coloneqq (\Tr X^\dagger X)^{1/2}$ is the Hilbert--Schmidt norm. We write $\|X\|_\infty$ for the operator norm.
The notation $\Omega_{A_I}$ denotes the reduced state on the device registers indexed by $I$.
Every inverse of a possibly singular density operator is restricted to its support.
Unless the base is displayed, logarithms are natural, including in binary entropy functions.

For independent device--reference pairs, the initial resource has the form $\bigotimes_{i=1}^n\Omega_{A_iR_i}$, with all $R_i$ held by the adversary.
For a general resource $\Omega_{A_1\cdots A_nE}$, we will compare selected device marginals with products of single-device or block states.
The rest of the resource remains available as an adversary reference in that comparison.
The below-half section first treats independent qubit state maps and fully separable resources, then extends the argument to independent output--reference families dominated by twice a fixed density operator, including PPT qubit outputs.

The common numerical function in the Fourier bounds is
\[
 \kappa(K)\coloneqq \frac K\pi\sin\frac\pi K \quad(K\ge2),\qquad
 \kappa(4)=\frac{2\sqrt2}{\pi}<1.
\]
The argument $K$ denotes an operator-domination constant: every state of an indexed output family, jointly with any retained reference, is at most $K$ times one fixed density operator, as in Lemma~\ref{lem:fourier}.
In the qubit leakage model with an arbitrary reference, $K=4$; for an emitted system of dimension $d$, it is $d^2$.
The constant of \cite{BKNY} indexed by $b$ output qubits, which have no retained reference, equals $\kappa(2^b)$, the value $K=d$ for $d=2^b$ in this notation.

Several results reveal a subset of shares and analyze the leakage from the remaining devices.
We recall the conditioning statement of \cite[Lemma~5]{BKNY} in our notation.

\begin{lemma}[Exact conditioning]\label{lem:conditioning}
Let $S\subseteq[n]$, $r\coloneqq |S|<t$, and $I\coloneqq [n]\setminus S$.
For each secret, $Y_S$ is uniform on $\F_p^S$.
Conditioned on $Y_S=z$, the remaining shares are coordinatewise affine images of an $(n-r,t-r)$ Shamir instance.
\end{lemma}
\begin{proof}
Let $h_z$ interpolate $z$ at $S$ with degree less than $r$, and let
$L_S(X)\coloneqq \prod_{j\in S}(X-x_j)$.
For $S=\varnothing$, use $h_z\coloneqq 0,L_S\coloneqq 1$.
The parametrization in \cite[Lemma~5, Eqs.~(21)--(23)]{BKNY} is
\[
 f=h_z+L_Sg,\qquad \deg g<t-r,\qquad
 g(0)=\frac{s-h_z(0)}{L_S(0)}.
\]
Conditioned on $Y_S=z$, the polynomial $g$ is uniform with these constraints, and each branch has probability $p^{-r}$ independently of $s$.
For $i\in I$,
$f(x_i)=h_z(x_i)+L_S(x_i)g(x_i)$, with $L_S(x_i)\ne0$.
These secret-independent affine coordinate transformations depend on $z$ and may be absorbed into the leakage maps.
\end{proof}

The one-device Fourier estimate reduces leakage for independent device--reference pairs to a weighted sum on the dual code of Shamir sharing, which is maximum distance separable (MDS).
Combining this reduction with the weighted estimate of Sun and Wootters gives the bound above one half in Theorem~\ref{thm:half}.
Matching lower constructions determine the exact exponent and its dependence on the entanglement within each pair.
Keeping the output-dimension dependence explicit gives the growing-budget estimate.
The below-half qubit proof instead expands each output around its mean state and controls the resulting tensors through real inner products.
For correlated resources, revealing complementary shares and comparing the selected marginal with a product transports the independent-pair or block bound to the original resource.
The support theorem uses a direct Fourier argument that retains the entangled support in the Fourier sum.

\section{Independent device--reference pairs}\label{sec:ea}
The joint output--reference space may have arbitrarily large finite dimension.
The useful replacement is domination by one fixed density operator, uniformly over the share value.
The bound for independent device--reference pairs proved here is the quantitative input for the single-device marginal comparisons below.

\begin{theorem}[Security above one half]\label{thm:half}
Fix an output dimension $d\ge2$ and $\eta\in(0,1/2)$.
Suppose the initial resource is $\bigotimes_{i=1}^n\Omega_{A_iR_i}$ and each local channel emits a system of dimension $d$.
There are constants $A,b,n_0>0$, depending only on $d,\eta$, such that, for $n\ge n_0$ and $t/n\ge1/2+\eta$,
\begin{equation}\label{eq:half}
 \Adv(\rho_s,\rho_{s'})
 \le A n^b\kappa(d^2)^{t-1}
       \exp\!\left(\frac{2\pi^2 n}{3p^2}\right).
\end{equation}
In particular, the distance is at most $2^{-cn}$ for some $c>0$ depending only on $d,\eta$.
The constants are uniform over every prime $p>n$, all distinct nonzero evaluation points, all chosen secrets and local channels, and all finite input and reference dimensions.
\end{theorem}

Our analysis follows the Fourier approach to quantum Shamir leakage in \cite{BKNY}.
For independent device--reference pairs, we dominate each local output family by a single state with a constant that depends on the emitted dimension and is independent of the reference dimension.
This bounds the maximum nonzero Fourier coefficient and the squared Fourier mass.
The dual MDS code of the sharing scheme then turns leakage into a weighted code sum.
We use the weighted Fourier thresholding method of Sun and Wootters \cite{SW}, based on algebraic list recovery \cite{Jeronimo}; Section~\ref{sec:growing-output} uses Nguyen's quantitative version when the output dimension grows \cite{Nguyen}.
We prove the local bounds first and return to Theorem~\ref{thm:half} after identifying that code.

\subsection{A dominated-state Fourier estimate}\label{sec:fourier}
For an integer $m\ge1$ and a map $\sigma:\F_p^m\to\mathcal D(H)$, define the normalized Fourier coefficients and their norms by
\[
 \widehat{\sigma}(a)\coloneqq \E_z\omega^{-\langle a,z\rangle}\sigma_z,\qquad
 h(a)\coloneqq \norm{\widehat{\sigma}(a)},\qquad \omega\coloneqq e^{2\pi i/p}.
\]

The output-dimension Fourier estimates of Batra et al.\ \cite[Lemma~2]{BKNY} bound both the squared Fourier mass and the nonzero coefficients.
Here a common operator dominator replaces the output-dimension hypothesis, and the finite-field cap has a multiplicative quadrature correction.

\begin{lemma}\label{lem:fourier}
Suppose $\gamma,\sigma_z$ are density operators and
$\sigma_z\le K\gamma$ for all $z$, where $K\ge2$.
Put
\[
 \lambda_p\coloneqq 1+\frac{2\pi^2}{3p^2},\qquad
 c_{K,p}\coloneqq \min\{1,\kappa(K)\lambda_p\}.
\]
Then
\[
 h(0)=1,\quad \sum_a h(a)^2\le K,\quad
 \sum_{a\ne0}h(a)^2\le K-1,\quad h(a)\le c_{K,p}\ (a\ne0).
\]
Consequently, for every real $Q\ge2$,
\[
 \sum_{a\ne0}h(a)^Q\le (K-1)c_{K,p}^{Q-2}.
\]
All bounds are independent of $\dim H$ and of $m$.
\end{lemma}
The common dominator controls both parts of the estimate.
After weighting by its inverse on its support, Parseval bounds the total squared Fourier mass by $K$.
For a nonzero frequency, trace-norm duality reduces the cap to maximizing the real part of a character over a scalar weight of mass one and height at most $K$.
The maximizing weight occupies an arc, whose normalized cosine integral is $\kappa(K)$.
The finite-field quadrature error gives the factor $1+2\pi^2/(3p^2)$ in the lemma.
The complete argument, including singular dominators and this error estimate, is in Supplementary Section~\ref{supp:fourier-proof}.

The estimate applies to local leakage once the output and the adversary reference admit a common dominating state.
The following elementary pinching bound supplies that state.

\begin{lemma}[Fixed-reference domination]\label{lem:reference}
If $\sigma_{BR}$ is a state with $\dim B=d$ and
$\Tr_B\sigma\eqqcolon \tau_R$, then
\[
 \sigma_{BR}\le d I_B\otimes\tau_R
 =d^2\left(\frac{I_B}{d}\otimes\tau_R\right).
\]
\end{lemma}
\begin{proof}
Pinching in a basis of $B$ gives $\sigma\le d\,\mathcal P_B(\sigma)$.
Each diagonal block is positive and at most the sum of all diagonal blocks, namely $\tau_R$.
Thus $\mathcal P_B(\sigma)\le I_B\otimes\tau_R$.
\end{proof}
A local trace-preserving map leaves the reference marginal unchanged for every share value.
For one emitted qubit, Lemma~\ref{lem:reference} therefore supplies the common domination constant $K=4$, regardless of reference dimension.

\subsection{The Shamir Fourier sum}
The preceding bounds hold separately at each device.
Independence of the device--reference resources makes their Fourier coefficients tensor products, so the remaining estimate is a sum over the dual sharing code.
Assume the initial resource is $\bigotimes_i\Omega_{A_iR_i}$.
Let $\sigma_i(y_i)$ be the joint state of the emitted system $B_i$ and its pre-existing reference $R_i$.
Conditioned on $y$, the state is $\bigotimes_i\sigma_i(y_i)$.

Let
\[
 \mathcal C_0\coloneqq \{(f(x_1),\ldots,f(x_n)): f(0)=0,\ \deg f<t\}
\]
be the zero-secret sharing space, with orthogonality taken in the standard dot product on $\F_p^n$, and put $\mathcal C\coloneqq \mathcal C_0^\perp$.
Then $\mathcal C$ is a linear code of length $n$, dimension $k\coloneqq n-t+1$, and minimum nonzero Hamming weight $t$; since $t=n-k+1$, it is maximum distance separable (MDS).
For completeness, the parity-check columns are
$(x_i,x_i^2,\ldots,x_i^{t-1})^\mathsf T$.
Any $t-1$ columns are independent by a nonzero column scaling of a Vandermonde determinant.
Thus every projection of $\mathcal C$ onto $k$ coordinates is bijective; we call any such set of $k$ coordinates an information set.

The sharing vector for secret $s$ is uniform on $s\mathbf1+\mathcal C_0$, where $\mathbf1$ is the all-ones vector. Fourier inversion and coset orthogonality give
\[
 \rho_s=\sum_{a\in\mathcal C}\omega^{s\sum_i a_i}
              \bigotimes_i\widehat{\sigma_i}(a_i).
\]
The zero frequency cancels between two secrets.
The difference of two phases has modulus at most two, and trace norm is multiplicative on tensor products.
Hence
\begin{equation}\label{eq:fourier-security}
 \Adv(\rho_s,\rho_{s'})
 \le\sum_{a\in\mathcal C\setminus\{0\}}\prod_i h_i(a_i),
\end{equation}
where $h_i(a)\coloneqq\norm{\widehat{\sigma_i}(a)}$.

\subsection{The bound above one half}
The remaining sum in Equation~\eqref{eq:fourier-security} has nonnegative scalar weights, so it can be estimated by a result about generalized Reed--Solomon (GRS) codes.
Such a code has the form
\[
 \mathcal G\coloneqq \{(v_iP(x_i))_{i=1}^n:\deg P<\nu\},
 \qquad v_i\in\F_p^\times,
\]
where the evaluation points are distinct and $\nu$ is the dimension.
For $1\le\nu\le n$, every such code is MDS: a nonzero polynomial of degree below $\nu$ has fewer than $\nu$ roots, so any $\nu$ coordinates form an information set.
The zero-secret space $\mathcal C_0$ is of this form with $v_i=x_i$ and $\nu=t-1$.
Its dual is the code $\mathcal C$ appearing in the leakage sum.
Such duals are again of this form, by the following description of Vandermonde kernels.

\begin{lemma}[Vandermonde kernels]\label{lem:vandermonde-kernel}
Let $x_1,\ldots,x_N\in\F_p$ be distinct, put $T(X)\coloneqq\prod_{i=1}^N(X-x_i)$, and let $0\le M<N$.
The solutions $u\in\F_p^N$ of $\sum_iu_ix_i^l=0$ for $0\le l<M$ are exactly the vectors $(Q(x_i)/T'(x_i))_i$ with $\deg Q<N-M$.
\end{lemma}
\begin{proof}
Lagrange's top-coefficient identity gives $\sum_if(x_i)/T'(x_i)=0$ whenever $\deg f\le N-2$, so every displayed vector is a solution.
These vectors form a space of dimension $N-M$, which is the dimension of the kernel.
\end{proof}

Applying the lemma with $M=\nu$ to $(a_iv_i)_i$ shows that $\mathcal G^\perp$ is the generalized Reed--Solomon code of dimension $n-\nu$ with multipliers $1/(v_iT'(x_i))$.
In particular, $\mathcal C$ is such a code of dimension $k$ with multipliers $(x_i\prod_{j\ne i}(x_i-x_j))^{-1}$.

We use the following specialization of Sun and Wootters \cite[Lemma~1.15 and Section~4.2]{SW}.
The polynomial factor in their proof uses the field-independent list-recovery bound of Jeronimo \cite[Theorem~1.1 and Corollary~8.1]{Jeronimo}.
Here list recovery bounds the number of codewords agreeing with prescribed coordinate lists on sufficiently many positions.
The lists have bounded size, and the bound is polynomial in the code length uniformly in the prime field.

\begin{lemma}[Weighted dual sum, Sun--Wootters]\label{lem:weighted}
Fix $c_*\in(0,1)$, $\Gamma>0$, and $\epsilon\in(0,1/2)$.
There are $A,b,n_0>0$, depending only on these constants, with the following property.
Let $n\ge n_0$, let $p\ge n$ be prime, and let $\mathcal G$ be a generalized Reed--Solomon code of dimension $\nu$ satisfying
$1/2+\epsilon\le\nu/n<1$.
If $w_i:\F_p^\times\to[0,\infty)$ satisfy
\[
 \max_{z\ne0}w_i(z)\le c_*,
 \qquad \sum_{z\ne0}w_i(z)^2\le\Gamma
 \quad(i\in[n]),
\]
then
\[
 \sum_{a\in\mathcal G^\perp\setminus\{0\}}
       \prod_{i\in\supp(a)}w_i(a_i)
 \le A n^b c_*^\nu.
\]
The constants are independent of the field, evaluation points, coordinate multipliers, and weights.
\end{lemma}

In the cited proof, a fixed cutoff separates large weights from small ones.
The squared-mass bound makes the set of large weights at each coordinate a bounded list.
List recovery controls codewords meeting those lists often; two disjoint information sets and Cauchy--Schwarz control the remaining codewords.
This argument applies directly to the trace norms in Equation~\eqref{eq:fourier-security}.
The stated bound is also the case of singleton blocks in Lemma~\ref{mb:weighted}, proved in Supplementary Section~\ref{supp:blocks-weighted}: apply it to the generalized Reed--Solomon code $\mathcal G^\perp$ with $K=1+\Gamma$ and $w_i(0)\coloneqq1$, which gives the slightly stronger exponent $\nu+1$.

\begin{proof}[Proof of Theorem~\ref{thm:half}]
For each local output--reference state, Lemma~\ref{lem:reference} gives a common dominator with $K=d^2$.
Put
$w_i(z)\coloneqq h_i(z)/\lambda_p$ for $z\ne0$, with $\lambda_p$ from Lemma~\ref{lem:fourier}.
Lemma~\ref{lem:fourier} gives the bounds
\[
 w_i(z)\le\kappa(K)<1,
 \qquad \sum_{z\ne0}w_i(z)^2\le K-1.
\]
For $n\ge2/\eta$, the dimension $\nu=t-1$ of $\mathcal C_0$ satisfies
$\nu/n\ge1/2+\eta/2$.
Apply Lemma~\ref{lem:weighted} to $\mathcal G=\mathcal C_0$ with
$c_*=\kappa(K)$, $\Gamma=K-1$, and $\epsilon=\eta/2$.
Since $h_i(0)=1$, each term of the original sum is $\lambda_p^{|\supp(a)|}$ times its normalized term, with $\lambda_p^{|\supp(a)|}\le\lambda_p^n$.
Equation~\eqref{eq:fourier-security} consequently gives
\begin{equation}\label{eq:half-K}
 \Adv(\rho_s,\rho_{s'})
 \le A n^b\kappa(K)^{t-1}
       \exp\!\left(\frac{2\pi^2n}{3p^2}\right).
\end{equation}
Because $p>n$ and $\kappa(K)<1$ is fixed, the polynomial factor is absorbed by the exponential decay for sufficiently large $n$.
All constants used above depend only on $K,\eta$.
\end{proof}

More generally, Equation~\eqref{eq:half-K} holds for any product of local state families $\sigma_i(y)$ for which
$\sigma_i(y)\le K\gamma_i$ with density operators $\gamma_i$ independent of $y$ and a fixed $K\ge2$.
We shall use this form for restricted Schmidt number and classical outputs in Section~\ref{sec:sharp}.

The explicit finite estimates pair opposite Fourier frequencies; Supplementary Section~\ref{supp:fourier} develops this paired-moment argument and the elementary MDS estimates.
They give explicit finite bounds, sufficient reconstruction rates in the scalar case, and constant-prefactor bounds at full threshold.
The weighted estimate above supplies the wider asymptotic rate range.

\section{Exact exponents and growing leakage budgets}\label{sec:sharp}
We determine the optimal exponent of the distance for independent device--reference pairs at every fixed reconstruction rate above one half.
The full-threshold construction identifies how the exponential base depends on the entanglement within each pair; Supplementary Section~\ref{supp:rates-endpoint} gives its constant-prefactor bounds.
We then apply the lower construction to a set of $t$ shares and combine it with the weighted upper bound to obtain the exponent throughout the above-half range.
We first fix the output dimension $d\ge2$, as in Theorem~\ref{thm:half}; Section~\ref{sec:growing-output} then allows it to grow with $n$.

The Schmidt number of a mixed bipartite state is the least integer $r$ for which it is a convex combination of pure states of Schmidt rank at most $r$ \cite{TH}.
Local channels cannot increase Schmidt number: applying each Kraus operator to a vector cannot increase its Schmidt rank.

The input and reference dimensions remain arbitrary and finite.
The following bound for states of Schmidt number at most $r$ refines Lemma~\ref{lem:reference}; it also follows from the reduction map in \cite[arXiv version, Definition~2 and Lemma~2]{TH}.

\begin{lemma}\label{lem:schmidt}
If $\sigma_{BR}$ has Schmidt number at most $r$ and $\Tr_B\sigma\eqqcolon \tau_R$, then
\[
 \sigma_{BR}\le r I_B\otimes\tau_R
 =dr\left(\frac{I_B}{d}\otimes\tau_R\right),\qquad d\coloneqq \dim B.
\]
\end{lemma}
\begin{proof}
For a pure state $|\psi\rangle$ of Schmidt rank $r_0\le r$, let $\tau_\psi$ be its reference marginal.
On the support of $\tau_\psi$, the vector $(I\otimes\tau_\psi^{-1/2})|\psi\rangle$ has squared norm $r_0$.
Thus $|\psi\rangle\langle\psi|\le r I_B\otimes\tau_\psi$.
Sum this inequality over a pure-state decomposition of $\sigma$ with Schmidt ranks at most $r$.
\end{proof}

The matching lower bound uses the output dimension to transmit orthogonal labels.
Choose a maximally entangled device--reference pair of Schmidt rank $1\le r\le d$. This pair admits $K=dr$ orthogonal dense-coding labels, each with the same reference marginal, while a classical $d$-symbol output supplies $K=d$ labels.
For a field value $z$, choose the label whose phase is nearest to $2\pi z/p$ after a common angular offset; the following lemma bounds the trace norm of the resulting first Fourier coefficient.

\begin{lemma}[Arc labels]\label{lem:arc-labels}
Let $|\beta_0\rangle,\ldots,|\beta_{K-1}\rangle$ be orthonormal, with $K\ge2$, and let $p$ be prime.
For an offset $\theta$, let $j_\theta(z)$ maximize $\cos(2\pi z/p-\theta-2\pi j/K)$ over $0\le j<K$, breaking ties deterministically, and put $\sigma_\theta(z)\coloneqq |\beta_{j_\theta(z)}\rangle\langle\beta_{j_\theta(z)}|$.
Then some offset $\theta$ satisfies
\[
 h_\theta(1)\coloneqq \norm{\E_z\omega^{-z}\sigma_\theta(z)}\ge\kappa(K),
\]
for every prime.
\end{lemma}
\begin{proof}
Let $U_\theta$ have eigenvalue $e^{i(\theta+2\pi j/K)}$ on $|\beta_j\rangle$.
Trace-norm duality gives
\[
 h_\theta(1)\ge
 \frac1p\sum_z\max_j\cos(2\pi z/p-\theta-2\pi j/K).
\]
Averaging the right side over $\theta$ gives exactly
\[
 \frac K{2\pi}\int_{-\pi/K}^{\pi/K}\cos v\,dv=\kappa(K),
\]
so some offset attains at least this value.
\end{proof}

On a reconstruction set of $t$ shares, scale each share by its reconstruction coefficient and use these local labels.
With a uniform secret, the scaled shares on that set are independent and uniform, so the first Fourier coefficient of the secret-indexed output is the tensor product of the local coefficients.
Trace-norm duality and averaging over pairs of secrets give a distinguishing distance of at least $\kappa(K)^t/2$ for some pair.
A common affine substitution in the shares transports the construction to every prescribed distinct pair.
Supplementary Section~\ref{supp:rates-endpoint} gives the complete label construction, its orthogonality proof, and the constant-prefactor endpoint estimates.
For classical outputs, the output--reference state is separable even when the initial pair is entangled; Lemma~\ref{lem:schmidt} then gives the upper domination constant $K=d$.

The matching construction also determines the exponent below full threshold.
Any selected $t$ shares already form a full-threshold instance on $t$ devices.
Using the construction there and emitting fixed states elsewhere gives a lower bound with power $t$; the upper bound in Equation~\eqref{eq:half-K} has the same power up to a polynomial factor.

\begin{corollary}[Exact exponent above one half]\label{cor:rate}
Fix $d\ge2$ and $\eta\in(0,1/2)$.
For independent device--reference pairs, assign $K\coloneqq d^2$ to $d$-dimensional outputs with unrestricted initial pairs, $K\coloneqq dr$ to $d$-dimensional outputs when every initial pair has Schmidt number at most $r$ with $1\le r\le d$, and $K\coloneqq d$ to a classical output alphabet of size $d$ with unrestricted initial pairs.
For fixed distinct secrets $s,s'$, let $L$ be the supremum of
$\Adv(\rho_s,\rho_{s'})$ over resources and channels in the indicated class, at parameters $n,t,p,x$.
There are constants $A,b,n_0>0$, depending only on $d,\eta$, such that for $n\ge n_0$ and $t/n\ge1/2+\eta$,
\begin{equation}\label{eq:matched-rate}
 \frac12\kappa(K)^t
 \le L
 \le A n^b\kappa(K)^{t-1}
       \exp\!\left(\frac{2\pi^2n}{3p^2}\right).
\end{equation}
Consequently,
\[
 -\frac1n\log L
 =\frac tn\log\frac1{\kappa(K)}
   +O_{d,\eta}\!\left(\frac{\log n}{n}\right),
\]
uniformly over every prime $p>n$, all evaluation points and all distinct secret pairs.
In particular, if $t/n\to R>1/2$, the limiting exponent is
$R\log(1/\kappa(K))$.
\end{corollary}
\begin{proof}
For the upper bound, the respective dominators are supplied by Lemmas~\ref{lem:reference} and~\ref{lem:schmidt}, and by the classical-output argument in Corollary~\ref{cor:classical}.
Apply Equation~\eqref{eq:half-K} in each case.
There are only finitely many admissible $r$ for a fixed $d$, so the constants can be chosen uniformly over $r$.

For the lower bound, choose any set $T\subseteq[n]$ of size $t$.
The original uniform polynomial of degree less than $t$ with fixed constant term induces exactly a $(t,t)$ Shamir sharing on $T$.
Theorem~\ref{thm:sharp}, or Corollary~\ref{cor:classical} for classical outputs (both in Supplementary Section~\ref{supp:rates-endpoint}), provides the required resource and channels on these $t$ devices for the prescribed secrets.
Give every other device a trivial independent resource and a constant-output channel.
Tensoring with their fixed output states leaves the trace distance unchanged and preserves all resource restrictions.
The selected construction has distance at least $\tfrac12\kappa(K)^t$.
Taking logarithms of the two bounds proves the exponent statement.
\end{proof}
The constant-prefactor bounds of Theorem~\ref{thm:sharp} and Corollary~\ref{cor:classical} remain sharper at $t=n$.
The exponent statement combines their lower constructions with the weighted upper estimate of Sun and Wootters; its quantum and reference dependence enters through the domination constant $K$.
For classical one-bit leakage with trivial references, the same exponent follows from Sun and Wootters \cite[Corollary~4.3]{SW} and the additive-sharing lower bound of Maji et al.\ \cite[Theorems~1 and~10]{Maji}, applied to a reconstruction set.
Corollary~\ref{cor:rate} permits arbitrary finite independent references and gives the exponent for quantum outputs as a function of Schmidt number.

\subsection{Growing output dimension}\label{sec:growing-output}
The constants in Corollary~\ref{cor:rate} were chosen with the output dimension fixed.
We now allow the number of emitted qubits to grow with the number of shares.
This requires an explicit bound on the lists used in the weighted argument.
Nguyen's quantitative list-recovery estimate supplies that bound, with explicit dependence on the list size.
The logarithmic-leakage theorem of Nguyen already treats classical outputs \cite[Theorems~2 and~6]{Nguyen}; here we apply the estimate to quantum outputs with retained references and determine their leading distinguishing exponent.

\begin{theorem}[Logarithmic quantum leakage]\label{lq:theorem}
Fix $0<\eta<1/2$.
There is $n_0(\eta)$ such that, for $n\ge n_0(\eta)$ and $t/n\ge1/2+\eta$, independent device--reference pairs remain secure when each device emits at most
\[
 q\le\left\lfloor\frac{\eta}{200}\log_2 n\right\rfloor
\]
qubits.
More precisely, writing $d\coloneqq 2^q$, every pair of secrets satisfies
\begin{equation}\label{lq:security}
 \Adv(\rho_s,\rho_{s'})
 \le\exp\!\left(-\frac{n}{8d^4}\right)
 \le\exp\!\left(-\frac18 n^{1-\eta/50}\right).
\end{equation}
The bound is uniform over primes $p>n$, evaluation points, all allowed local channels, and arbitrary finite input and reference dimensions, which may depend on $n$.
The initial resource is independent of the secret and sharing randomness, and the channels have the locality and no-feedback conditions of Section~\ref{sec:model}.
\end{theorem}

\begin{proof}[Proof outline of Theorem~\ref{lq:theorem}]
We use the quantitative weighted estimate in Lemma~\ref{lq:weighted}, whose full derivation from Nguyen's list-recovery bound is in Supplementary Section~\ref{supp:growing}.
For a local domination constant $K\ge2$ it gives Equation~\eqref{lq:finite}, whose dependence on $K$ is explicit.
The first term there counts codewords with many large weights; two disjoint information sets control the remaining term.
For $2\le K\le n^{\eta/100}$, the list parameter $\ell$ of that lemma satisfies $\ell^{36}=O_\eta(n^{9/10})$, so its length condition holds uniformly, and $\ell^{12}\log n=O_\eta(n^{3/10}\log n)$.
Since $-\log\kappa(K)\ge K^{-2}$, $n/K^2\ge n^{99/100}$ and $t\ge n/2$, both terms of that bound are at most $\exp(-n/(4K^2))$, and a further uniform increase of the length cutoff gives $\Adv\le\exp(-n/(8K^2))$.
For $q\ge1$, fixed-reference domination supplies $K=d^2=4^q\le n^{\eta/100}$, which gives the two inequalities in \eqref{lq:security}.
For $q=0$, the emitted systems are trivial and the reference marginal is unchanged, so the distance is zero.
Supplementary Section~\ref{supp:growing} gives the complete finite inequalities, including integer cutoffs and the elementary sine bound.
\end{proof}

The same estimate also determines the leading decay when the dimensions grow.
The lower constructions in Theorem~\ref{thm:sharp} are finite-dimensional identities, so they remain available for each value of the growing dimension.
At $t=n$, the finite upper bound in Theorem~\ref{thm:sharp} already gives growing-dimension consequences.
Lemma~\ref{lq:weighted} extends this comparison to every fixed rate above one half.

\begin{corollary}[Matching exponents for growing outputs]\label{lq:rate}
Fix $0<\eta<1/2$ and consider any of the three classes in Corollary~\ref{cor:rate}, with $d$ and, when applicable, $r$ allowed to depend on $n$.
Use $K\coloneqq d^2$, $dr$, or $d$ for the respective classes, and let $L$ be the supremum of the distance for a prescribed pair of distinct secrets.
For $n\ge n_0(\eta)$, $t/n\ge1/2+\eta$, and $2\le K\le n^{\eta/100}$,
\begin{align}
 -\log L
 &=t\log\frac1{\kappa(K)}
   +O_\eta\!\left(K^{12+24/\eta}\log n+1\right)\label{lq:additive-rate}\\
 &=\left(1+O_\eta(n^{-69/100}\log n)\right)
   t\log\frac1{\kappa(K)}.\label{lq:relative-rate}
\end{align}
All error bounds are uniform over the permitted fields, evaluation points, dimensions, references and distinct secrets.
If $K\to\infty$ in this range, then
\[
 -\log L=(1+o(1))\frac{\pi^2t}{6K^2}.
\]
In particular, the denominator is $d^4$ for unrestricted device--reference pairs with $d\to\infty$.
\end{corollary}
\begin{proof}[Proof outline]
The construction above gives $L\ge\kappa(K)^t/2$ at every finite allowed dimension.
Combining it with \eqref{lq:finite} and using $2^{-n}/\kappa(K)^t\le(\pi/4)^n$ bounds the additive logarithmic error; the leading term is at least $n/(2K^2)$, which gives the relative error, and $-\log\kappa(K)=\pi^2/(6K^2)+O(K^{-4})$ gives the growing-$K$ asymptotic.
Supplementary Section~\ref{supp:growing-rate-proof} gives the complete two-sided finite comparison and uniformity calculation.
\end{proof}

\begin{remark}[The prime-field restriction]\label{rem:prime}
Classical one-bit leakage can reveal information about a secret in linear masking over fields of characteristic two, including Shamir sharing; see \cite[Section~4.1, comparison with linear masking schemes]{BFGV}. The following trace formulation gives a full-threshold example.
Over $\F_{2^h}$ with $2^h>n$, full-threshold reconstruction has the form $s=\sum_i\lambda_iY_i$.
Each device can release the single bit $\operatorname{Tr}_{\F_{2^h}/\F_2}(a\lambda_iY_i)$, using no correlated resource.
The XOR of these bits is $\operatorname{Tr}_{\F_{2^h}/\F_2}(as)$.
For any prescribed distinct pair $s,s'$, nondegeneracy of the trace pairing supplies $a$ for which $\operatorname{Tr}_{\F_{2^h}/\F_2}(a(s-s'))=1$.
The two secrets are then perfectly distinguishable.
\end{remark}

\section{Independent qubit leakage below one half}\label{bh:section}
For independent one-qubit outputs, the real geometry of the Bloch ball gives an exponential security bound at a fixed rate below one half.
In the Shamir experiment, the resulting sufficient rate is $t/n\ge1/2-5\cdot10^{-10}$ for all sufficiently large $n$.
The bound also allows a resource separable across all devices and the full adversary reference, as specified in Corollary~\ref{bh:separable}.
We formulate the estimate first for affine generalized Reed--Solomon codes, so that the comparison state is independent of the chosen secret.
The argument combines the support decomposition and truncated Fourier estimate of Sun and Wootters \cite[Sections~3.2 and~4.2--4.3]{SW} with Nguyen's real-function and bias estimates \cite[Section~6]{Nguyen}.
On large supports, centered qubit tensors are expressed through real Bloch vectors, and their squared Euclidean norms become averages of real Gram kernels on a doubled input alphabet.

\begin{theorem}\label{bh:main}
Put $c_{\rm q}\coloneqq 1/2-10^{-9}$.
There are absolute constants $a>0$ and $n_0$ such that the following holds for every prime $p>n\ge n_0$.
Let $\mathcal G\subseteq\F_p^n$ be a generalized Reed--Solomon code of dimension $\nu$ with $\nu/n\ge c_{\rm q}$, and let
$\rho_i:\F_p\to\mathcal D(\mathbb C^2)$ be arbitrary state maps.
For $z\in\F_p^n$, define
\[
 \bar\rho_i\coloneqq \E_{x\in\F_p}\rho_i(x),\qquad
 R_z\coloneqq \E_{X\in z+\mathcal G}\bigotimes_{i=1}^n\rho_i(X_i),
 \qquad R_* \coloneqq \bigotimes_{i=1}^n\bar\rho_i.
\]
Then
\begin{equation}\label{bh:coset-bound}
 \Adv(R_z,R_*)\le e^{-an}.
\end{equation}
The constants are uniform over the field, evaluation points, nonzero coordinate multipliers, state maps, and translate $z$.
\end{theorem}

Each $\rho_i(x)$ is the entire local output in this theorem.
Thus an independent device can use a private ancilla and discard it before emitting its qubit.
The proof separates the centered tensor expansion by support size.
Moderate supports are controlled by a truncated Fourier sum, whereas the larger supports use the real Gram representation below.
The supporting estimates and their full proofs are in Supplementary Section~\ref{supp:bh-details}.

\subsection{Centered supports and the truncated Fourier estimate}
Write the qubit states in Bloch form,
\begin{equation}\label{bh:bloch}
 \rho_i(x)=\frac{I+v_i(x)\cdot\sigma}{2},\qquad
 v_i(x)\in\mathbb R^3,\quad |v_i(x)|\le1,
 \qquad b_i\coloneqq \E_xv_i(x),\quad u_i(x)\coloneqq v_i(x)-b_i.
\end{equation}
Here $\sigma\coloneqq (\sigma_1,\sigma_2,\sigma_3)$ is the vector of Pauli matrices, and $|\cdot|$ denotes the Euclidean norm, also on $\mathbb C^3$ when applied to Fourier coefficients.
For $A\subseteq[n]$, let
\[
 T_A(z)\coloneqq \E_{X\in z+\mathcal G}
       \bigotimes_{i\in A}\bigl(\rho_i(X_i)-\bar\rho_i\bigr).
\]
Expanding each state about its uniform-field mean gives
\begin{equation}\label{bh:support-expansion}
 R_z-R_*=
 \sum_{\varnothing\ne A\subseteq[n]}
 T_A(z)\otimes\bigotimes_{i\notin A}\bar\rho_i,
\end{equation}
with tensor factors restored to their coordinate order.
Every projection of the affine MDS code $z+\mathcal G$ onto at most $\nu$ coordinates is uniform.
Consequently $T_A(z)=0$ for $0<|A|\le\nu$.

Use the normalized Fourier transform from Section~\ref{sec:fourier}, entrywise for vector-valued functions.
For $a\ne0$, set $w_i(a)\coloneqq |\widehat v_i(a)|$.
Pauli orthogonality, $\Tr(\sigma_j^\dagger\sigma_l)=2\delta_{jl}$, gives
\[
 \|\widehat\rho_i(a)\|_1
 \le\sqrt2\|\widehat\rho_i(a)\|_2
 =w_i(a).
\]
This Hilbert--Schmidt identity also holds when the Fourier coefficient is non-Hermitian.
Parseval and the reality of the Bloch vectors give
\begin{equation}\label{bh:fourier-cap}
 \sum_{a\ne0}w_i(a)^2
 =\E_x|v_i(x)|^2-|b_i|^2\le1,
 \qquad w_i(a)=w_i(-a)\le2^{-1/2}.
\end{equation}
For the second inequality, $a$ and $-a$ are distinct because $p$ is odd, so their two equal squared magnitudes are bounded by the total nonzero mass.

Centering removes every zero Fourier frequency inside $A$.
Fourier inversion and orthogonality on the coset therefore give the exact formula
\begin{equation}\label{bh:exact-support}
 T_A(z)=
 \sum_{\substack{y\in\mathcal G^\perp\\\supp(y)=A}}
 \omega^{\langle y,z\rangle}
 \bigotimes_{i\in A}\widehat\rho_i(y_i).
\end{equation}
In particular, summing the trace norms over supports of size at most $m$ counts each nonzero dual word of weight at most $m$ once:
\begin{equation}\label{bh:low-reduction}
 \sum_{0<|A|\le m}\|T_A(z)\|_1
 \le
 \sum_{\substack{0\ne y\in\mathcal G^\perp\\|\supp(y)|\le m}}
       \prod_{i\in\supp(y)}w_i(y_i).
\end{equation}
The factors $\bar\rho_i$ outside a support have trace norm one.
For later use, the precise consequence of Lemma~\ref{bh:truncated} is as follows.
For fixed $c_0,\eta>0$, if $\nu\ge c_0n$ and the integer cutoff $\nu\le m\le n$ satisfies $2\nu-m\ge\eta n$, then for all sufficiently large $n$,
\begin{equation}\label{bh:main-low-contract}
 \sum_{0<|A|\le m}\|T_A(z)\|_1
 \le (A_0n^{b_0}+1)2^{-\nu/2}.
\end{equation}
Here $A_0,b_0$ and the length threshold depend only on $c_0,\eta$.
This follows from \eqref{bh:low-reduction} and \eqref{bh:fourier-cap} with the parameters $\gamma=2^{-1/2}$, $\Gamma=1$ of Lemma~\ref{bh:truncated}.

\subsection{Gram kernels and the bias estimate}
We now return to the support operators in \eqref{bh:support-expansion}.
Define the real Euclidean tensor
\[
 M_A(z)\coloneqq \E_{X\in z+\mathcal G}\bigotimes_{i\in A}u_i(X_i)
       \in(\mathbb R^3)^{\otimes |A|}.
\]
Write $\|\cdot\|_{\rm E}$ for the Euclidean norm on this tensor space, in its orthonormal product basis.
If $m_\alpha$ are the coefficients of $M_A(z)$, then
\[
 T_A(z)=2^{-|A|}\sum_{\alpha\in\{1,2,3\}^A}
              m_\alpha\bigotimes_{i\in A}\sigma_{\alpha_i},
 \qquad
 \|T_A(z)\|_2^2=2^{-|A|}\|M_A(z)\|_{\rm E}^2.
\]
The trace/Hilbert--Schmidt inequality on a $2^{|A|}$-dimensional space consequently gives
\begin{equation}\label{bh:dimension-cancellation}
 \|T_A(z)\|_1
 \le2^{|A|/2}\|T_A(z)\|_2
 =\|M_A(z)\|_{\rm E}.
\end{equation}
Thus Pauli orthogonality exactly cancels the dimension factor in the trace-norm bound.

Let $X,Y$ be independent uniform words in $z+\mathcal G$, and define
\[
 K_i(x,y)\coloneqq \langle u_i(x),u_i(y)\rangle,
 \qquad r_i\coloneqq |b_i|,\qquad B_A\coloneqq \sum_{i\in A}r_i^2.
\]
Taking the inner product of two copies of the tensor average gives the exact Gram identity
\begin{equation}\label{bh:gram}
 \|M_A(z)\|_{\rm E}^2
 =\E_{X,Y\in z+\mathcal G}\prod_{i\in A}K_i(X_i,Y_i).
\end{equation}
Each $K_i$ is real, has mean zero under uniform independent field inputs, and satisfies
$|K_i(x,y)|\le(1+r_i)^2$.
The scalar input is the following consequence of Lemmas~\ref{bh:real-gap} and~\ref{bh:product-mds}.
Put $\gamma_*\coloneqq (3/4)^{1/8}$.
If $L_i:\F_p^2\to\mathbb R$ have mean zero and $|L_i|\le1$ for $i\in A$, and $J\subseteq A$ satisfies $2|J|\le3\nu-n$, then
\begin{equation}\label{bh:main-scalar-contract}
 \left|\E_{X,Y\in z+\mathcal G}\prod_{i\in A}L_i(X_i,Y_i)\right|
 \le\gamma_*^{|J|}.
\end{equation}
Thus the centered Gram kernels meet exactly the hypotheses needed for a uniform contraction over paired codewords.

Assume $n/3<\nu\le n/2$ and put $j_0\coloneqq \lfloor(3\nu-n)/2\rfloor$.
For $|A|\ge\nu+1$, choose $J\subseteq A$ of size $j_0$.
Apply Lemma~\ref{bh:product-mds} to the paired codeword average in \eqref{bh:gram}, using the normalized kernels on $A$ and the constant function one outside $A$.
Taking square roots in \eqref{bh:gram} and using \eqref{bh:dimension-cancellation} gives
\begin{equation}\label{bh:small-bias}
 \|T_A(z)\|_1
 \le\gamma_*^{j_0/2}\prod_{i\in A}(1+r_i)
 \le\exp\!\left(\frac{j_0}{2}\log\gamma_*+\sqrt{nB_A}\right).
\end{equation}
Here $\log(1+r_i)\le r_i$ and Cauchy--Schwarz give
$\sum_{i\in A}r_i\le\sqrt{|A|B_A}\le\sqrt{nB_A}$.
This bound is strongest when the mean Bloch vectors are small.

For the complementary bias range, two disjoint information sets of size $\nu$, together with
$\E_x|u_i(x)|^2\le1-r_i^2$, give the following second bound; its full variance calculation is in Supplementary Section~\ref{supp:bh-bias}:

\begin{equation}\label{bh:large-bias}
 \|T_A(z)\|_1
 \le\exp\!\left((n-2\nu)(\log2+\tfrac12)-\tfrac12B_A\right).
\end{equation}
The inequalities include $r_i=1$: in that case the variance is zero and the centered map at that coordinate vanishes identically.
\begin{proof}[Proof of Theorem~\ref{bh:main}]
First suppose $c_{\rm q}\le\nu/n\le1/2$ and set $m\coloneqq \lfloor(1-10^{-8})n\rfloor$.
Then $2\nu-m\ge(10^{-8}-2\cdot10^{-9})n=8\cdot10^{-9}n$, so \eqref{bh:main-low-contract}, with $c_0=c_{\rm q}$ and $\eta=8\cdot10^{-9}$, bounds the smaller supports by $(A_0n^{b_0}+1)2^{-\nu/2}$, with absolute $A_0,b_0$.
For the larger supports, split according to $B_A\le2^{-18}n$ or $B_A>2^{-18}n$ and use \eqref{bh:small-bias} or \eqref{bh:large-bias}, respectively.
The entropy count of supports with complement smaller than $10^{-8}n$ and the numerical estimates in Lemma~\ref{bh:large-support-sum} give
\[
 \sum_{|A|>m}\|T_A(z)\|_1
 \le(n+1)e^{-10^{-6}n+\lambda_*/2},
 \qquad \lambda_*\coloneqq -\log\gamma_*.
\]
Equation~\eqref{bh:support-expansion} now yields
\[
 \Adv(R_z,R_*)\le\frac12\left[(A_0n^{b_0}+1)2^{-\nu/2}
                  +(n+1)e^{-10^{-6}n+\lambda_*/2}\right]
 \le e^{-an}
\]
for an absolute $a>0$ and all sufficiently large $n$.
All parameters of the truncated bound and of the large-support estimate were fixed independently of the code and state maps.
For larger $\nu$, choose the GRS subcode of dimension $\nu_0\coloneqq \lceil c_{\rm q}n\rceil$ obtained by restricting polynomial degree, with the same evaluation points and multipliers.
For $n\ge10^9$, it satisfies $c_{\rm q}\le\nu_0/n\le1/2$ and $\nu_0\le\nu$.
Every coset of $\mathcal G$ is a uniform mixture of cosets of this subcode, all with the same comparison state $R_*$, which uses uniform field inputs.
Convexity of trace distance proves the assertion, with constants uniform over all data in the theorem.
\end{proof}

\subsection{Chosen secrets and fully separable resources}
For Shamir sharing, the code dimension in Theorem~\ref{bh:main} is $t-1$.
The following rate leaves room for that integer offset.

\begin{corollary}\label{bh:shamir}
Suppose each device emits a qubit according to an independent state map $\rho_i(Y_i)$.
For every sufficiently large $n$, if
\[
 \frac tn\ge c_{\rm S}\coloneqq \frac12-5\cdot10^{-10},
\]
then, for all chosen secrets $s,s'$,
\begin{equation}\label{bh:chosen-bound}
 \Adv(\rho_s,\rho_{s'})\le2e^{-an},
\end{equation}
with the absolute exponent $a$ from Theorem~\ref{bh:main}.
The length threshold and bound are uniform over every prime $p>n$, evaluation points, and local state maps.
In particular, after increasing the length threshold, the distance is at most $e^{-an/2}$.
\end{corollary}
\begin{proof}
A fixed-secret polynomial has the form $s+XQ(X)$ with $\deg Q<t-1$, so the share vector is uniform on $s\mathbf1+\mathcal C_0$, where $\mathcal C_0$ is GRS of dimension $t-1$ and multipliers $x_i$.
For $n\ge2\cdot10^9$, the stated rate implies $(t-1)/n\ge c_{\rm q}$.
Theorem~\ref{bh:main} places both secret outputs within $e^{-an}$ of the same $R_*$, giving \eqref{bh:chosen-bound} by the triangle inequality.
The last assertion follows for $n\ge2\log2/a$.
\end{proof}

The next consequence permits a finite adversary reference when the whole resource is separable across the devices and that reference.
It also includes shared classical randomness, by giving the adversary the index of the separable decomposition.

\begin{corollary}\label{bh:separable}
In the local leakage model of Section~\ref{sec:model}, suppose the secret- and sharing-independent initial resource has a decomposition
\begin{equation}\label{bh:full-separability}
 \Omega_{A_1\cdots A_nE}
 =\sum_\ell q_\ell
       \left(\bigotimes_{i=1}^n\tau_{A_i}^{\ell}\right)
       \otimes\eta_E^{\ell},
\end{equation}
where $(q_\ell)$ is a probability distribution and every factor is a normalized density operator.
Each local channel emits one qubit.
Under the rate and length assumptions of Corollary~\ref{bh:shamir},
$\Adv(\rho_s,\rho_{s'})\le2e^{-an}$, uniformly in all finite input and reference dimensions and in the decomposition.
In particular, this holds for a product of device--reference pairs each of Schmidt number one.
\end{corollary}
\begin{proof}
Give the adversary an additional orthogonal classical label $\ell$.
Conditional on this label, define the qubit state maps
$\rho_i^\ell(y)\coloneqq \Phi_i^y(\tau_{A_i}^\ell)$, and let $R_s^\ell$ be their joint output for secret $s$.
The conditional reference $\eta_E^\ell$ is independent of the secret and sharing randomness.
Writing $B_\ell\coloneqq \bigotimes_i\E_y\rho_i^\ell(y)$, the reduction in the proof of Corollary~\ref{bh:shamir} and Theorem~\ref{bh:main} give
$\Adv(R_s^\ell,B_\ell)\le e^{-an}$, uniformly in $\ell$.
The labeled output and its secret-independent comparison state are
\[
 L_s\coloneqq \sum_\ell q_\ell|\ell\rangle\langle\ell|
                   \otimes R_s^\ell\otimes\eta_E^\ell,
 \qquad
 L_* \coloneqq \sum_\ell q_\ell|\ell\rangle\langle\ell|
                   \otimes B_\ell\otimes\eta_E^\ell.
\]
Tensoring with a density operator preserves trace norm, and orthogonal classical blocks add their trace norms.
Hence
\[
 \Adv(L_s,L_*)
 =\sum_\ell q_\ell\Adv(R_s^\ell,B_\ell)\le e^{-an},
 \qquad \Adv(L_s,L_{s'})\le2e^{-an}.
\]
Discarding the label recovers the original experiment, so data processing proves the desired bound.

For Schmidt-number-one pairs, expand each finite separable decomposition across $A_i:R_i$ and multiply the decompositions.
The tuple of local indices gives \eqref{bh:full-separability} with $E=R_1\cdots R_n$.
Finite-dimensional separability admits a finite convex decomposition, and the norm estimate is uniform in its length.
\end{proof}

The partition in \eqref{bh:full-separability} is $A_1|\cdots|A_n|E$, including the full adversary reference.
This identifies exactly the resource assumption used to reduce the experiment to independent qubit state maps.
The emitted registers of the sequential GHZ generator in Supplementary Section~\ref{supp:memory-examples} have a classical-mixture marginal. An adversary retaining the final memory obtains the known constant-advantage attack.

\subsection{A common dominator and PPT outputs}\label{ppt:section}
The same below-half estimate applies to larger local output spaces when
every state is bounded above by twice one fixed density operator.
Whitening such a state family replaces the Bloch ball by a unit ball in a real
Hilbert space, with the same real Gram bounds. The qubit reduction criterion supplies
the required dominator for joint PPT output--reference states. The resulting estimate is uniform in the finite dimension of the retained quantum reference.

\begin{theorem}[Domination by twice a density operator]\label{ppt:domination-two}
In Theorem~\ref{bh:main}, replace each qubit state map by
$\sigma_i:\F_p\to\mathcal D(H_i)$, where $H_i$ is any finite-dimensional
space. Suppose there is a density operator $\gamma_i$ on $H_i$ such that
\begin{equation}\label{ppt:domination}
 \sigma_i(y)\le2\gamma_i\qquad(y\in\F_p).
\end{equation}
Define
\[
 \bar\sigma_i\coloneqq\E_y\sigma_i(y),\qquad
 R_z\coloneqq\E_{X\in z+\mathcal G}\bigotimes_i\sigma_i(X_i),
 \qquad R_*\coloneqq\bigotimes_i\bar\sigma_i.
\]
The same absolute constants $a,n_0$ can be chosen so that
$\Adv(R_z,R_*)\le e^{-an}$ whenever
$p>n\ge n_0$ and $\nu/n\ge c_{\rm q}=1/2-10^{-9}$.
The estimate is uniform over the finite dimensions, dominators, state maps,
GRS evaluation points and multipliers, and translate $z$.
\end{theorem}
The choice $\gamma_i=I_2/2$ recovers Theorem~\ref{bh:main} exactly.
Supplementary Section~\ref{ppt:proof} proves the extension by checking the
Fourier, Gram-kernel and variance estimates used in that theorem.

For qubit outputs, the reduction criterion of Horodecki and Horodecki
\cite[arXiv version, Section~III]{HHReduction} identifies a useful physical instance of
\eqref{ppt:domination}.
A bipartite state is positive under partial transpose, or PPT, across
$B:R$ when its partial transpose on $B$ is positive semidefinite.

\begin{corollary}[PPT outputs and classical bits]\label{ppt:physical}
Suppose the full initial resource is a product
$\bigotimes_i\Omega_{A_iR_i}$ of finite device--reference states,
independent of the secret and sharing randomness.
Each device applies its local CPTP map $\Phi_i^y:A_i\to B_i$ and emits
at most one qubit, and the adversary retains every $R_i$.
Put
\[
 \sigma_i(y)\coloneqq(\Phi_i^y\otimes\id_{R_i})(\Omega_{A_iR_i}).
\]
If $\sigma_i(y)$ is PPT across $B_i:R_i$ for every $i,y$, then, for
all sufficiently large $n$ and $t/n\ge c_{\rm S}=1/2-5\cdot10^{-10}$,
\[
 \Adv(\rho_s,\rho_{s'})\le2e^{-an}
\]
for every pair of secrets.
The constants are uniform over every prime $p>n$, all evaluation points,
permitted channels, and finite input and reference dimensions.
In particular, the conclusion holds for every local one-qubit channel
when each initial pair is PPT, and for every classical-bit output channel
with arbitrary independent initial pairs.
\end{corollary}
\begin{proof}
Trace preservation fixes the reference marginal
$\tau_i\coloneqq\Tr_{B_i}\sigma_i(y)$ independently of $y$.
For a qubit $B$ and
$J\coloneqq\left(\begin{smallmatrix}0&-1\\1&0\end{smallmatrix}\right)$,
direct block multiplication gives
\begin{equation}\label{ppt:reduction}
 I_B\otimes\tau-\sigma
 =(J\otimes I_R)\sigma^{T_B}(J^\dagger\otimes I_R).
\end{equation}
Thus an output PPT state obeys
$\sigma\le I_B\otimes\tau=2[(I_B/2)\otimes\tau]$.
A one-dimensional output may be embedded into a qubit.
Theorem~\ref{ppt:domination-two} now applies to the full output--reference
families. The Shamir zero-secret code has dimension $t-1$, and
$n\ge2\cdot10^9$ makes the stated rate imply $(t-1)/n\ge c_{\rm q}$.
Both secret outputs are within $e^{-an}$ of the same product of local
means, giving the displayed bound by the triangle inequality.

For an initially PPT pair, a Kraus representation gives
\[
 \sigma_i(y)^{T_{B_i}}
 =\sum_\alpha(\overline{K_{i,\alpha}^y}\otimes I)
       \Omega_{A_iR_i}^{T_{A_i}}
       ((K_{i,\alpha}^y)^{\mathsf T}\otimes I)\ge0.
\]
For a classical bit the output has the form
\[
 |0\rangle\langle0|\otimes M_i(y)
 +|1\rangle\langle1|\otimes[\tau_i-M_i(y)],
 \qquad 0\le M_i(y)\le\tau_i,
\]
so each block is at most $\tau_i$, regardless of initial entanglement.
\end{proof}
The PPT premise concerns each emitted qubit together with its retained
reference. Independence is imposed on the full device--reference pairs.
Corollary~\ref{bh:separable} continues to cover the fully separable
multipartite resource model, including its shared classical mixtures.

\section{Selected marginals and entanglement support}\label{sec:correlation}
A correlated resource can still have a large marginal bounded by a scalar multiple of a product of single-device or block states.
Revealing the complementary shares leaves the exact smaller Shamir instance of Lemma~\ref{lem:conditioning} and preserves uniformity in the field size.
A filter on the adversary's reference, using the postselection mechanism of Christandl, K\"onig, and Renner \cite{CKR}, transfers the bounds for independent device--reference pairs or independent block--reference factors at the cost of that scalar multiple.
When an arbitrarily entangled support factors from the other devices, a direct argument later in the section gives a sharper bound by keeping that support in the Fourier sum.

The next theorem converts a marginal operator inequality into a security bound.
The comparison is made on a purification, which permits arbitrary correlations with the unselected devices and the adversary.

\begin{theorem}[Marginal domination transfer]\label{thm:domination}
Let $\Omega_{A_1\cdots A_nE}$ be any share-independent initial state.
Choose $I\subseteq[n]$, put $S\coloneqq [n]\setminus I$, $m\coloneqq |I|$, and $r\coloneqq n-m$, and assume $r<t$.
Suppose there exist local density matrices $\tau_i$ and $\Lambda\ge1$ such that
\[
 \Omega_{A_I}\le\Lambda\bigotimes_{i\in I}\tau_i.
\]
If $\varepsilon_{\rm ref}(m,t-r,p)$ is any bound for independent device--reference pairs applicable to $m$ shares, threshold $t-r$, and field $\F_p$, uniform over evaluation points, all distinct secret pairs, permitted local channels, and finite reference dimensions, then
\[
 \Adv(\rho_s,\rho_{s'})
 \le \min\{1,\Lambda\varepsilon_{\rm ref}(m,t-r,p)\}.
\]
\end{theorem}
\begin{proof}
Set $\sigma\coloneqq \bigotimes_i\tau_i$ and $\eta_I\coloneqq \Omega_{A_I}$.
In a fixed product basis use $\operatorname{vec}(X)\coloneqq \sum_{a,b}X_{ab}|a\rangle|b\rangle$, so that
$(I\otimes F)\operatorname{vec}(X)=\operatorname{vec}(XF^{\mathsf T})$.
Take the canonical purifications
$|\psi_\sigma\rangle\coloneqq \operatorname{vec}(\sqrt{\sigma})$ and
$|\psi_{\eta_I}\rangle\coloneqq \operatorname{vec}(\sqrt{\eta_I})$.
The reference of the first purification factors into independent $R_i$.
All inverses below are on $\supp(\sigma)$.
Define a reference operator $F$ by
\[
 F^{\mathsf T}\coloneqq \frac{\sigma^{-1/2}\sqrt{\eta_I}}{\sqrt{\Lambda}}.
\]
Then
\[
 F^\dagger F=
 \left(\frac{\sigma^{-1/2}\eta_I\sigma^{-1/2}}{\Lambda}\right)^{\mathsf T}
 \le I,
 \qquad
 (I\otimes F)|\psi_\sigma\rangle
 =\Lambda^{-1/2}|\psi_{\eta_I}\rangle.
\]
Support inclusion justifies the last equality.
Any extension of $\eta_I$, including $\Omega_{A_I A_SE}$, is obtained from a purification by a channel on the purifying system.
Thus a single completely positive trace-nonincreasing reference map transforms the product purification into $\Omega/\Lambda$.

This map commutes with the leakage channels at $I$.
Reveal $Y_S$ and simulate the $S$ channels afterwards.
For every fixed revealed vector, the trace-nonincreasing map is independent of the residual secret and its success probability is exactly $1/\Lambda$.
A positive trace-nonincreasing map contracts trace norm on Hermitian differences.
Hence the original conditional distance is at most $\Lambda$ times the product-reference distance.
Lemma~\ref{lem:conditioning} and equal-weight averaging complete the proof.
\end{proof}
The same transfer holds when all leakage maps have classical-bit outputs and $\varepsilon_{\rm ref}$ is replaced by the bound for that class.
Indeed, the filter and extension map act only on the reference, and each selected device still applies its original classical-output channel to the product purification.
In general, $\varepsilon_{\rm ref}$ must hold for the original local channels, with the affine input changes of Lemma~\ref{lem:conditioning}, acting on the product purification; output-state restrictions must therefore be verified on that product purification.

The proof also applies when the comparison state is a product of block states.

\begin{corollary}[Block form of the transfer]\label{thm:domination-blocks}
In Theorem~\ref{thm:domination}, partition $I$ into nonempty blocks $I_1,\ldots,I_M$ and replace the comparison by
\[
 \Omega_{A_I}\le\Lambda\bigotimes_{j=1}^M\alpha_j,
\]
where each $\alpha_j$ is a density operator on $A_{I_j}$.
If $\varepsilon_{\rm block}(m,t-r,p)$ is any bound for independent blocks with their original local channels, applicable to $m$ shares in these blocks, threshold $t-r$, and field $\F_p$, uniform over evaluation points, all distinct secret pairs, permitted local channels, and finite reference dimensions, then
\[
 \Adv(\rho_s,\rho_{s'})
 \le \min\{1,\Lambda\varepsilon_{\rm block}(m,t-r,p)\}.
\]
Each device receives only its own share.
\end{corollary}
\begin{proof}
Repeat the proof of Theorem~\ref{thm:domination} with $\sigma\coloneqq \bigotimes_j\alpha_j$.
Its canonical purification is a product of independent purifications of the blocks, so the reference factors into independent block references.
The same reference filter changes only the reference, so the maps inside each retained block remain local, and the original conditional distance is at most $\Lambda$ times the independent-block distance.
\end{proof}

Supplementary Section~\ref{supp:marginal-consequences} states the resulting bounds for exact and approximate product marginals, including residual-threshold conditions, rates and smoothing error; Supplementary Section~\ref{sec:resources} applies them.

\subsection{The size of an entangled support}\label{sup:section}
How large can an arbitrarily entangled set of devices be while local leakage remains exponentially small?
This is the support question of Batra, Kitagawa, Nishimaki, and Yamakawa \cite[Section~7, question~2]{BKNY}.
When the support factors from the outside devices in the full device marginal, we retain it as one factor in the Fourier sum.
Throughout this subsection, $r\coloneqq |S|$ denotes the support size.

\begin{theorem}[An entangled support below the reconstruction threshold]\label{sup:upper}
Fix an integer $d\ge2$, $\eta\in(0,1/2)$, and $\xi\in(0,1]$.
There are constants $A,b,n_0>0$, depending only on $d,\eta,\xi$, with the following property.
Consider an $(n,t)$ Shamir instance with $n\ge n_0$, prime $p>n$, and arbitrary distinct nonzero evaluation points.
Let $S\subseteq[n]$, put $r\coloneqq |S|$, $O\coloneqq [n]\setminus S$, and $N\coloneqq n-r$, and suppose
\begin{equation}\label{sup:margins}
 \frac tn\ge\frac12+\eta,
 \qquad t-r\ge\xi n.
\end{equation}
Each device emits a system of dimension at most $d$ by a fixed local share-dependent CPTP map, with no communication or adversary feedback during leakage.
The initial state is independent of the secret and sharing randomness.
If its device marginal satisfies
\begin{equation}\label{sup:coarse-product}
 \Omega_{A_1\cdots A_n}
 =\Omega_{A_S}\otimes\bigotimes_{i\in O}\tau_i,
\end{equation}
then every finite-dimensional external extension is allowed, and every pair of secrets obeys
\begin{equation}\label{sup:bound}
 \Adv(\rho_s,\rho_{s'})
 \le A n^b\kappa(K)^{t-r}
       \exp\!\left(\frac{2\pi^2N}{3p^2}\right),
 \qquad K\coloneqq d^2.
\end{equation}
If the full initial resource instead has the form
\begin{equation}\label{sup:no-outside-reference}
 \Omega_{A_1\cdots A_nE}
 =\Omega_{A_SE}\otimes\bigotimes_{i\in O}\tau_i,
\end{equation}
the same bound holds with $K=d$.
All constants are uniform over support sets, evaluation points, resources, finite input and reference dimensions, and permitted channels.
\end{theorem}

For $d=2$, Equation~\eqref{sup:no-outside-reference} is the support model of \cite[Definition~9]{BKNY}: outside devices have no retained quantum references.
Equation~\eqref{sup:coarse-product} permits such references and their joint processing.
It requires the support to factor from the outside devices, whereas Corollary~\ref{cor:product} requires only a product marginal on the outside devices and permits arbitrary correlations with the support.

\begin{proof}[Proof outline]
For \eqref{sup:coarse-product}, a product of purifications of the support state and outside states purifies the full device marginal.
Every allowed extension is obtained by a common channel on that reference, so trace-distance contraction reduces the analysis to independent support and outside factors.
In the model \eqref{sup:no-outside-reference}, the outside factors are the original single-device states.

Let $H$ be the multivariate Fourier norm of the support output and its reference, and let $h_i$ be the outside Fourier norms.
Set $w_i(0)\coloneqq 1$ and $w_i(a)\coloneqq h_i(a)/\lambda_p$ for $a\ne0$.
With $c\coloneqq \kappa(K)$ and $k\coloneqq n-t+1$, the Fourier and domination bounds give
\[
 0\le H\le1,\quad \sum_v H(v)^2\le d^{2r},\quad
 w_i(a)\le c\ (a\ne0),\quad \sum_a w_i(a)^2\le d^2,
\]
and
\[
 \Adv(\rho_s,\rho_{s'})
 \le\lambda_p^N\sum_{a\in\mathcal C\setminus\{0\}}
           H(a_S)\prod_{i\notin S}w_i(a_i).
\]
The support is retained as one factor throughout this sum.
The two dimension margins are
\[
 n-2k=2t-n-2,\qquad N-k=t-r-1.
\]
Both have a fixed positive linear margin by \eqref{sup:margins}.
Choose a fixed cutoff for the outside weights.
Words with few outside weights below that cutoff are counted by fixed-list recovery on the restriction (puncturing) of $\mathcal C$ to the outside coordinates; puncturing is injective because $N\ge k$.
There are polynomially many such words, and each has at least $t-r$ nonzero outside coordinates, contributing at most $c^{t-r}$.

For the remaining words, mark a fixed positive fraction of light outside coordinates.
Two disjoint injective coordinate projections can be chosen as unions of whole factors: use the support alone when $r\ge k$, or add $k-r$ outside coordinates when $r<k$, and take a disjoint outside information set for the second projection.
Cauchy--Schwarz then costs at most $d^n$ in total squared mass.
A sufficiently small fixed cutoff $\theta$ absorbs this cost and the choices of marked positions: for their number $q$, $2^N\theta^q d^n\le c^n\le c^{t-r}$.
Combining both classes and restoring $\lambda_p^N$ proves \eqref{sup:bound}.
Supplementary Section~\ref{supp:support-upper-proof} gives the complete puncturing, floor and two-projection calculations.
\end{proof}

The proof uses the two exact dimension margins $2t-n-2$ and $t-r-1$.
By comparison, applying the above-half independent-leakage bound after revealing the support requires $N\ge2k$ even before reserving a fixed margin.
For fixed limiting fractions this revelation argument gives $\alpha<2R-1$, whereas Theorem~\ref{sup:upper} permits $\alpha<R$ for $R>1/2$.
The support estimate in \cite[Theorem~17, Corollary~2 and Equation~(25)]{BKNY} uses its independent-qubit threshold $\tau_*\approx0.73339$ and gives $\alpha<(R-\tau_*)/(1-\tau_*)$ when $R>\tau_*$.

Remark~\ref{sup:extensions} gives the extension to joint access within the support and to classical mixtures satisfying the full conditional-resource premise.

Batra et al.\ already describe a partial-GHZ attack: for $1\le\ell<t$, a GHZ state on $t-\ell$ devices and randomized bits from the other $\ell$ reconstruction devices give advantage at least $2^{-\ell}\cos^2(\pi/(2p))$ \cite[Section~6.3, final paragraph]{BKNY}.
The full-support endpoint is their Theorem~18.
We combine this phase construction with the arc-label channels of Lemma~\ref{lem:arc-labels} to match the exponent in Equation~\eqref{sup:bound}.
The support devices emit quantum states in this lower construction.

\begin{proposition}[A matching support-dependent lower bound]\label{sup:lower}
Let $d\ge2$ be an integer, let $p>n\ge2$ be prime, and fix $1\le r\le t$, a support $S$ of size $r$, and distinct secrets $s,s'$.
There is a local strategy satisfying Equation~\eqref{sup:coarse-product}, with output dimension at most $d$, such that
\begin{equation}\label{sup:lower-bound}
 \Adv(\rho_s,\rho_{s'})
 \ge\cos\!\left(\frac\pi{2p}\right)\kappa(d^2)^{t-r}.
\end{equation}
Under the more restrictive resource premise in Equation~\eqref{sup:no-outside-reference}, the same assertion holds with $\kappa(d)$.
The latter outside outputs can be classical labels in $d$ orthogonal states.
\end{proposition}
\begin{proof}[Proof outline]
Choose a reconstruction set $T\supseteq S$ of size $t$.
Use a GHZ state on the nonempty support and the arc-label channels on the $t-r$ outside reconstruction devices.
The support phases combine to the secret phase minus the sum of the scaled outside shares.
Those outside shares are independent and uniform for each fixed secret, since their number is less than $t$.
The averaged output is a two-by-two block matrix with secret-independent diagonal blocks; its off-diagonal block is the secret phase times a tensor product of outside first Fourier coefficients, divided by two.
If their norms are $h_j(1)$, the trace distance for two secrets is therefore
\[
 \frac{|\omega^{us}-\omega^{us'}|}{2}
       \prod_{j\in T\setminus S}h_j(1).
\]
Choose $u$ so that $u(s'-s)=(p-1)/2$.
The phase factor is $\cos(\pi/(2p))$, and Lemma~\ref{lem:arc-labels} gives $h_j(1)\ge\kappa(K)$ with $K=d^2$ or $d$ in the respective resource classes.
This proves the displayed lower bound, including the empty outside product when $r=t$.
The complete local construction and off-diagonal trace-norm calculation are in Supplementary Section~\ref{supp:support-lower-proof}.
\end{proof}

\begin{corollary}[The exponent for a linear support deficit]\label{sup:rate}
Fix an integer $d\ge2$, $\eta\in(0,1/2)$ and $\xi\in(0,1]$.
For fixed distinct secrets and parameters $n,t,p,x$, let $L_K(n,t,r,p,x;s,s')$ be the supremum of the leakage distance over local strategies with support size at most $r$, in the model of Equation~\eqref{sup:coarse-product} for $K=d^2$ or Equation~\eqref{sup:no-outside-reference} for $K=d$.
Under Equation~\eqref{sup:margins}, for all sufficiently large $n$,
\begin{equation}\label{sup:two-sided}
 \beta_{r,p}\kappa(K)^{t-r}
 \le L_K
 \le A n^b\kappa(K)^{t-r}
          \exp\!\left(\frac{2\pi^2(n-r)}{3p^2}\right),
 \qquad
 \beta_{r,p}\coloneqq \begin{cases}
  \tfrac12,&r=0,\\
  \cos(\pi/(2p)),&r\ge1.
 \end{cases}
\end{equation}
In particular,
\begin{equation}\label{sup:exponent}
 -\frac1n\log L_K
 =\frac{t-r}{n}\log\frac1{\kappa(K)}
   +O_{d,\eta,\xi}\!\left(\frac{\log n}{n}\right),
\end{equation}
uniformly over the field, evaluation points and chosen distinct secrets.
If $t/n\to R>1/2$ and $r/n\to\alpha<R$, the limiting exponent is
$(R-\alpha)\log(1/\kappa(K))$.
\end{corollary}
\begin{proof}
The upper bound is Theorem~\ref{sup:upper}; supports smaller than $r$ can be enlarged as in Remark~\ref{sup:extensions}.
For $r\ge1$, Proposition~\ref{sup:lower} gives the lower bound using any support of size $r$.
For $r=0$, use the selected-$t$ construction in Corollary~\ref{cor:rate}, with unrestricted pairs when $K=d^2$ and orthogonal output labels when $K=d$.
The lower prefactor is bounded away from zero for $p\ge3$, and the field-dependent term in the upper logarithm is $O(1/n)$ because $p>n$.
Taking logarithms proves Equation~\eqref{sup:exponent}.
\end{proof}

For the qubit support model of \cite[Definition~9]{BKNY}, the exponent is therefore $(R-\alpha)\log(\pi/2)$.
Every fixed support fraction $\alpha<R$ is exponentially secure at every fixed rate $R>1/2$; allowing a support of size $r=t$ includes the GHZ attack described in \cite[Theorem~18]{BKNY}.
This answers the strict fixed-fraction form of \cite[Section~7, question~2]{BKNY} in the above-half regime, with supremal support fraction $R$.
The partial-GHZ bound of Batra et al.\ already gives constant leakage for a fixed deficit $t-r$ and rules out a positive exponential decay rate when $t-r=o(n)$.
For $t-r=O(\log n)$ it also rules out negligible leakage.
For larger sublinear deficits, the lower bound may itself be negligible.

Under the no-outside-reference premise \eqref{sup:no-outside-reference}, revealing the support also permits the below-half qubit estimate to be applied to the remaining shares.
Corollary~\ref{sup:below-half} gives the resulting finite condition and small positive support fractions at and slightly below rate one half.

\section{Mixing chains and finite quantum memory}\label{sec:mixing}
We compare separated blocks of emitted registers with a product of block states, using upper factorization of quantum-chain states and completely positive comparison with a reset channel $X\mapsto\Tr(X)\sigma$, for a fixed state $\sigma$ \cite{Hiai,HMO}, and convergence of primitive channels to that reset \cite{Sanz}.
A one-coordinate slice bounds every nonzero block Fourier coefficient independently of block length.
Separated intervals can therefore be long enough to retain nearly all devices while their comparison cost is absorbed by the leakage exponent, giving exponential security at every fixed reconstruction rate above one half.
For the positive results in this section, all local dimensions and the underlying chain or memory process are fixed as $n$ grows.

\subsection{Fixed mixing resources}
A stationary quantum chain assigns to each finite interval of consecutive integer sites a state on the tensor product of copies of one fixed one-site space; these states are consistent under partial trace and unchanged by translating the interval.
The condition below compares separated intervals in operator order.
We will iterate it to control the marginal on separated intervals; Supplementary Section~\ref{app:sparse-memory} gives the singleton estimates.

\begin{definition}\label{def:upper-mixing}
A stationary quantum-chain state satisfies the \emph{upper operator mixing condition} if, for every $\beta>1$, there is an integer $g\ge0$ such that its marginals on any two finite intervals $J_1,J_2$ separated by at least $g$ omitted sites satisfy the full operator inequality
\[
 \Omega_{J_1J_2}\le\beta\,\Omega_{J_1}\otimes\Omega_{J_2}.
\]
The inequality applies to every positive joint observable.
Compare the two-sided operator bounds for algebraic states in \cite[Proposition~2.3]{Hiai}; here we require only the upper comparison for finite intervals.
\end{definition}

A sequential preparation gives a concrete way to verify the product comparison.
The relevant property is loss of memory under the channel obtained by tracing out each emission.
A channel $T$ is primitive if some power maps every nonzero positive operator to a positive definite operator.

For a fixed mixing process, retaining consecutive blocks permits a larger density of selected devices than retaining isolated sites.
The local leakage maps still act separately inside each retained block.
This locality keeps the relevant Fourier cap independent of the block length and gives the full range of fixed reconstruction rates above one half.

\begin{theorem}[Fixed mixing resources above one half]\label{mb:memory}
Fix an integer $d\ge2$ and $\eta\in(0,1/2)$.
Consider any one of the following fixed families of device marginals.
\begin{enumerate}
\item Fix finite-dimensional spaces $R,A$ and a CPTP emission channel
$\Psi:R\to A\otimes R$ whose transfer channel $T\coloneqq \Tr_A\circ\Psi$ is primitive.
The device marginal consists of $n$ successive emissions, starting from any state of $R$.
\item Fix a stationary quantum chain with a finite-dimensional one-site space and satisfying the upper operator mixing condition of Definition~\ref{def:upper-mixing}.
The device marginal is its restriction to $n$ consecutive sites.
\item Fix positive integers $D,v$ and matrices $A_1,\ldots,A_v\in\mathbb C^{D\times D}$ such that
\[
 \sum_{j=1}^v A_j^\dagger A_j=I_D,
 \qquad T(X)\coloneqq \sum_{j=1}^v A_jXA_j^\dagger
\]
is primitive.
The device state is the normalized pure matrix product state (MPS) with amplitudes
\[
 \langle u_R|A_{j_n}\cdots A_{j_1}|u_L\rangle,
 \qquad j_i\in\{1,\ldots,v\},
\]
where $u_L,u_R\in\mathbb C^D$ are arbitrary unit vectors.
\end{enumerate}
There are constants $c>0$ and $n_0$, depending only on the fixed process, chain, or matrices and on $d,\eta$, with the following property.
For $n\ge n_0$ and $t/n\ge1/2+\eta$, every share-independent initial resource extending the stated device marginal, and every permitted family of local CPTP leakage maps with output dimension at most $d$, satisfy
\[
 \Adv(\rho_s,\rho_{s'})\le2^{-cn}.
\]
The constants are uniform over all primes $p>n$, distinct nonzero evaluation points, chosen secrets, local maps, and finite adversary references.
They are also uniform over the initial memory state in the first case and both boundary vectors in the third case.
In the third case the normalization is well defined for all $n\ge n_0$.
\end{theorem}

We first bound leakage when independence holds between bounded groups of devices.
We then combine that bound with a product comparison for separated blocks.
Lemma~\ref{mb:comparison} below obtains the required comparison from reset-channel and boundary-normalization arguments, while the block length controls the fraction of shares retained.

\subsection{Independent entangled blocks}
The block model of Batra et al.\ permits entanglement and joint access to shares within each independent block, without outside quantum references \cite[Definition~10, Lemmas~2--3 and Theorems~15--16]{BKNY}.
Here the reference attached to each block may have arbitrary finite dimension, and each device applies its own share-dependent channel.
The latter restriction gives a Fourier cap determined by one output system even when the block contains many devices.
Recall the finite-field factor $\lambda_p=1+2\pi^2/(3p^2)$ from Lemma~\ref{lem:fourier}.

\begin{lemma}[A local slice of a block Fourier coefficient]\label{mb:slice}
Let $m\ge1$ and let $\Omega_{A_1\cdots A_mR}$ be a state independent of $z\in\F_p^m$.
For each $i$, let $\Phi_i^{z_i}:A_i\to B_i$ be CPTP with $\dim B_i\le d$, where $d\ge2$.
Set
\[
 \sigma_z\coloneqq \left(\bigotimes_{i=1}^m\Phi_i^{z_i}\otimes\id_R\right)(\Omega),
 \qquad
 h(a)\coloneqq \left\|\E_z\omega^{-\langle a,z\rangle}\sigma_z\right\|_1.
\]
Then
\begin{equation}\label{mb:slice-bound}
 h(0)=1,\qquad
 \sum_{a\in\F_p^m}h(a)^2\le d^{2m},\qquad
 h(a)\le\min\{1,\kappa(d^2)\lambda_p\}\quad(a\ne0).
\end{equation}
The bounds are independent of the input and reference dimensions.
\end{lemma}
\begin{proof}
The reference marginal is the same for all $z$.
Treating $B_1\cdots B_m$ as one output of dimension at most $d^m$, Lemma~\ref{lem:reference} gives a common dominator with constant $d^{2m}$.
Lemma~\ref{lem:fourier} gives the squared-mass bound, and the zero Fourier coefficient is a density operator.

Fix $a\ne0$ and choose $i$ with $a_i\ne0$.
Hold $z_{-i}$ fixed and regard $B_{-i}R$ as the reference for $B_i$.
Its marginal is independent of $z_i$, because $\Phi_i^{z_i}$ is trace preserving and acts only on $A_i$.
Lemmas~\ref{lem:reference} and~\ref{lem:fourier}, now with the single-output constant $d^2$, yield
\[
 \left\|\E_{z_i}\omega^{-a_i z_i}\sigma_{z_i,z_{-i}}\right\|_1
 \le\min\{1,\kappa(d^2)\lambda_p\}.
\]
Averaging these operators over $z_{-i}$ with the remaining Fourier phases and applying the triangle inequality gives the same bound for $h(a)$.
The reference state can vary with the held-fixed shares; the upper bound is uniform over those states.
\end{proof}

The scalar input for these blocks is Lemma~\ref{mb:weighted}, proved in Supplementary Section~\ref{supp:blocks-weighted}.
For a fixed maximum block length $\ell$, take weights with $w_j(0)=1$ and $0\le w_j\le1$, a fixed bound $c_*\in(0,1)$ on every nonzero block weight, and a fixed squared-mass bound $K$ on each block. For fixed $\epsilon\in(0,1/2)$ and all sufficiently large $n$, the lemma bounds the nonzero weighted sum on a GRS code of dimension $k\le(1/2-\epsilon)n$ and minimum distance $\Delta\coloneqq n-k+1$ by
\[
 A n^b c_*^{\lceil\Delta/\ell\rceil}.
\]
Its constants depend only on the fixed block parameters and $\epsilon$, uniformly over the field, partition and weights.
The proof projects large block weights to bounded coordinate lists and uses two disjoint collections of whole blocks for the light weights.
The block length is fixed before the number of shares grows.

\begin{proposition}[Independent blocks above one half]\label{mb:independent}
Fix integers $d\ge2$, $\ell\ge1$, and $\eta\in(0,1/2)$.
Suppose the devices are partitioned into nonempty blocks $I_1,\ldots,I_M$ of size at most $\ell$, and the initial resource has the form
\[
 \bigotimes_{j=1}^M\Omega_{A_{I_j}R_j},
\]
with all $R_j$ held by the adversary.
Each device applies its own share-dependent CPTP map and emits a system of dimension at most $d$.
There are constants $A,b,n_0>0$, depending only on $d,\ell,\eta$, such that, for $n\ge n_0$ and $t/n\ge1/2+\eta$,
\begin{equation}\label{mb:independent-bound}
 \Adv(\rho_s,\rho_{s'})
 \le A n^b\kappa(d^2)^{\lceil t/\ell\rceil}
       \exp\!\left(\frac{2\pi^2M}{3p^2}\right).
\end{equation}
The bound is uniform over all primes $p>n$, evaluation points, partitions, chosen secrets, resources, local maps, and finite input and reference dimensions.
\end{proposition}
\begin{proof}
For each block, include its reference in the output state family and let $h_j$ be its multivariate Fourier norm.
Independence between the block--reference factors gives the block version of the expansion leading to Equation~\eqref{eq:fourier-security}:
\[
 \Adv(\rho_s,\rho_{s'})
 \le\sum_{a\in\mathcal C\setminus\{0\}}\prod_{j=1}^M h_j(a_{I_j}),
 \qquad \mathcal C=\mathcal C_0^\perp.
\]
This dual GRS code has dimension $k\coloneqq n-t+1$ and minimum distance $t$.
Set $w_j(0)\coloneqq 1$ and $w_j(v)\coloneqq h_j(v)/\lambda_p$ for $v\ne0$.
Lemma~\ref{mb:slice} gives
\[
 w_j(v)\le\kappa(d^2)\quad(v\ne0),\qquad
 \sum_vw_j(v)^2\le d^{2|I_j|}\le d^{2\ell}.
\]
For $n\ge2/\eta$, we have $k\le(1/2-\eta/2)n$.
Apply Lemma~\ref{mb:weighted} with $\mathcal G=\mathcal C$, $\Delta=t$, $c_*=\kappa(d^2)$, $K=d^{2\ell}$, and $\epsilon=\eta/2$.
Undoing the normalization costs at most one factor $\lambda_p$ per nonzero block, hence at most
$\lambda_p^M\le\exp(2\pi^2M/(3p^2))$.
This proves Equation~\eqref{mb:independent-bound}.
\end{proof}

\subsection{Separated blocks and the memory bound}
Retain blocks of length $\ell$ separated by $g$ omitted devices.
If their marginal obeys a comparison with block states, Corollary~\ref{thm:domination-blocks} applies to independent purifications of these blocks and leaves the individual share-dependent device maps local.
Revealing the omitted shares leaves a residual instance on the retained devices, to which Proposition~\ref{mb:independent} applies when its reconstruction rate exceeds one half by a fixed margin.

\begin{proposition}[A rate from block marginal domination]\label{mb:criterion}
Fix integers $d\ge2$, $\ell\ge1$, $g\ge0$, and constants $C_0\ge1$, $\beta\ge1$.
For $n\ge\ell+g$, put
\[
 M\coloneqq \left\lfloor\frac n{\ell+g}\right\rfloor,\qquad
 I_j\coloneqq \{1+(j-1)(\ell+g),\ldots,\ell+(j-1)(\ell+g)\},
 \qquad I\coloneqq \bigcup_{j=1}^M I_j.
\]
Suppose the share-independent device marginal satisfies
\begin{equation}\label{mb:block-domination}
 \Omega_{A_I}\le C_0\beta^{M-1}\bigotimes_{j=1}^M\alpha_j,
\end{equation}
where each $\alpha_j$ is a density operator on $A_{I_j}$.
Any finite external extension of the full device marginal is permitted.
If a fixed $\delta\ge0$ satisfies
\begin{equation}\label{mb:strict-criterion}
 \delta<\frac{\ell}{\ell+g}
   \min\left\{\frac12,
      1-\frac{\log\beta}{\log(1/\kappa(d^2))}\right\},
\end{equation}
then every permitted local leakage strategy with output dimension at most $d$ has exponentially small distance whenever $t/n\ge1-\delta$.
More precisely, for every
\begin{equation}\label{mb:exponent}
 0<c<\frac1{\log2}
 \left[\frac{-\log\bigl(\beta\kappa(d^2)\bigr)}{\ell+g}
       -\frac{\delta}{\ell}\log\frac1{\kappa(d^2)}\right],
\end{equation}
the bound $\Adv(\rho_s,\rho_{s'})\le2^{-cn}$ holds for all sufficiently large $n$.
The length cutoff depends only on the displayed fixed parameters and $c$, and is uniform over the field, evaluation points, states, local maps, and finite input and reference dimensions.
\end{proposition}
In the strict criterion, the first term ensures a fixed residual above-half margin; the second makes the final exponent positive after paying the comparison factor.
Supplementary Section~\ref{supp:blocks-transfer} gives the proof.
For a fixed mixing process, we choose a gap first and then lengthen the retained blocks; the local slice cap remains unchanged.

The block comparisons for the three resource classes of Theorem~\ref{mb:memory} come from one lemma, which the sparse estimates of Supplementary Section~\ref{app:sparse-memory} also use.
Write $\mathcal S_1\le_{\rm CP}\mathcal S_2$ when $\mathcal S_2-\mathcal S_1$ is completely positive.

\begin{lemma}[Block comparisons for fixed processes]\label{mb:comparison}
Let $\ell\ge1$, $g\ge0$ and $M\ge1$ be integers, and let $I_1,\ldots,I_M\subseteq[n]$ be intervals of $\ell$ consecutive sites, in order, with $I_1$ starting at site one and exactly $g$ omitted sites between consecutive intervals.
Put $I\coloneqq \bigcup_jI_j$ and fix $\beta>1$.
\begin{enumerate}
\item For a stationary chain satisfying Definition~\ref{def:upper-mixing} with gap $g$ for this $\beta$, let $\alpha$ be its marginal on $\ell$ consecutive sites.
Then $\Omega_{A_I}\le\beta^{M-1}\alpha^{\otimes M}$.
\item Let $T$ be a primitive channel on a finite-dimensional memory, with full-rank stationary state $\sigma$ and reset $\mathcal R_\sigma(X)\coloneqq \Tr(X)\sigma$.
Then there is a fixed integer $g_0\ge0$ such that, for $g=g_0$,
\begin{equation}\label{mb:reset}
 T^g\le_{\rm CP}\beta\mathcal R_\sigma.
\end{equation}
\item In the first case of Theorem~\ref{mb:memory}, suppose \eqref{mb:reset} holds for its transfer channel and the spacing $g$, and let $\alpha$ be the joint state of $\ell$ consecutive emissions starting from $\sigma$.
For every initial memory state $\tau$, with $\Lambda_0\coloneqq \|\sigma^{-1/2}\tau\sigma^{-1/2}\|_\infty\le1/\lambda_{\min}(\sigma)$, where $\|\cdot\|_\infty$ is the operator norm,
\begin{equation}\label{mb:memory-domination}
 \Omega_{A_I}
 \le\Lambda_0\beta^{M-1}\alpha^{\otimes M}
 \le\frac{1}{\lambda_{\min}(\sigma)}
       \beta^{M-1}\alpha^{\otimes M}.
\end{equation}
\item In the third case of Theorem~\ref{mb:memory}, suppose \eqref{mb:reset} holds for its transfer channel and the spacing $g$, and let $\alpha$ be as in the preceding item for the emission isometry $V\coloneqq \sum_j|j\rangle\otimes A_j$.
For all sufficiently large $n$, uniformly over the two unit boundary vectors,
\[
 \Omega_{A_I}\le\frac{2}{\lambda_{\min}(\sigma)^2}\beta^{M-1}\alpha^{\otimes M}.
\]
\end{enumerate}
\end{lemma}
\begin{proof}
For the first item, apply upper operator mixing to the interval containing the earlier retained blocks and the next retained block, which are separated by $g$ omitted sites.
Trace out the omitted sites inside the left interval and iterate the resulting operator inequality $M-1$ times.
Stationarity gives the stated bound.

For the second item, primitivity gives convergence of $T^g$ to $\mathcal R_\sigma$ \cite[arXiv version, Proposition~3 and Eq.~(2)]{Sanz}.
For linear maps on the memory, use the unnormalized Choi matrix
\[
 J(\mathcal S)\coloneqq \sum_{a,b}|a\rangle\langle b|
                         \otimes\mathcal S(|a\rangle\langle b|)
\]
in a fixed orthonormal basis.
Since $J(\mathcal R_\sigma)=I\otimes\sigma>0$, choose a fixed $g_0$ with
\[
 \|J(T^{g_0}-\mathcal R_\sigma)\|_\infty
 \le(\beta-1)\lambda_{\min}(\sigma).
\]
Then \eqref{mb:reset} holds for $g=g_0$.

For the third item, every initial memory state satisfies $\tau\le\Lambda_0\sigma$, and $\tau\le I\le\sigma/\lambda_{\min}(\sigma)$ gives $\Lambda_0\le1/\lambda_{\min}(\sigma)$.
Between successive retained blocks, the $g$ omitted emissions induce the transfer map $T^g$.
Replace each of the $M-1$ gap maps using Equation~\eqref{mb:reset}.
Complete positivity keeps the order valid while the memory is entangled with earlier retained outputs, and each reset removes those correlations.
Partial trace over the final memory and any trailing emissions preserves this domination bound.
The initial cost is uniform over $\tau$.

For the fourth item, the map $V$ is an isometry by the canonical normalization.
Start its memory in $|u_L\rangle$, so that the preceding item applies with $\tau=|u_L\rangle\langle u_L|$, and postselect the final memory onto $|u_R\rangle$.
The success probability is
\[
 q_n\coloneqq \langle u_R|T^n(|u_L\rangle\langle u_L|)|u_R\rangle
 \ge\tfrac12\lambda_{\min}(\sigma)
\]
for all sufficiently large $n$, uniformly over the two unit boundary vectors, by uniform finite-dimensional convergence to the full-rank reset.
The unnormalized marginal for this measurement outcome is bounded above by the unconditional marginal, since the complementary outcome is positive.
Normalizing therefore multiplies Equation~\eqref{mb:memory-domination} by at most $2/\lambda_{\min}(\sigma)$.
\end{proof}

\begin{proof}[Proof of Theorem~\ref{mb:memory}]
Fix the process, chain, or matrices, and fix $d,\eta$.
Choose
\begin{equation}\label{mb:beta-choice}
 1<\beta<\kappa(d^2)^{-1/2}.
\end{equation}
For a stationary upper-mixing chain, Definition~\ref{def:upper-mixing} supplies a fixed gap $g$ for this $\beta$.
For a primitive emission process, let $\sigma>0$ be the stationary memory state.
For this process, and for the primitive transfer channel in the MPS case, the second item of Lemma~\ref{mb:comparison} supplies a fixed gap $g$ with \eqref{mb:reset}.

Having fixed the gap, put $\delta\coloneqq 1/2-\eta$ and choose a fixed integer $\ell\ge1$ such that
\begin{equation}\label{mb:length-choice}
 \delta<\frac{\ell}{2(\ell+g)}.
\end{equation}
For example, any $\ell>g(1-2\eta)/(2\eta)$ suffices.
Choose the retained blocks as in Proposition~\ref{mb:criterion}.
By Equation~\eqref{mb:beta-choice}, the minimum in Equation~\eqref{mb:strict-criterion} equals $1/2$, so Equation~\eqref{mb:length-choice} is exactly the required strict criterion.
It remains to obtain the block marginal comparison with a constant initial cost.

Lemma~\ref{mb:comparison} gives the block marginal comparison \eqref{mb:block-domination}, with every $\alpha_j$ equal to the $\alpha$ defined there.
Its first item gives $C_0=1$ for the chain; its third item gives $C_0=1/\lambda_{\min}(\sigma)$ for the emission process, uniformly over the initial memory state; and its fourth item gives $C_0=2/\lambda_{\min}(\sigma)^2$ for the MPS, uniformly over both boundary vectors for all sufficiently large $n$.

Apply Proposition~\ref{mb:criterion} in each case.
Only after $\beta,g,\ell$ have been fixed do we choose $n_0$ for the list-recovery bound, residual rate margin, boundary normalization where needed, and fixed-prefactor absorption.
The resulting constants have exactly the dependence stated in the theorem.
The reference filter in that proposition permits every share-independent extension of the device marginal, including preparation environments and the final memory.
\end{proof}

The sparse-site estimates in Lemma~\ref{lem:mix-finite} and Proposition~\ref{prop:sparse-rate} retain explicit finite-field costs and sufficient exponents for a prescribed gap.
The block argument instead establishes every fixed reconstruction rate above one half for a fixed process; its constants also depend on the chosen block length.

\paragraph{An example and the fixed-process premise.}
The two-dimensional memory constructed in Supplementary Section~\ref{supp:memory-examples} has transfer channel
$T=(1-\mu)\id+\mu\mathcal R_{I/2}$ and emits a four-dimensional device register at each step.
Its correlations are nonzero at every separation.
For $\mu=1/2$ and a gap $g=6$, its reset comparison has $\beta=67/64$ and $\beta\sqrt{\kappa(4)}<1$ for qubit leakage.
Sampling isolated sites gives exponential security at every fixed rate above $13/14$; blocks of length $54$ give exponential security at every fixed rate $t/n\ge11/20+\eta$, $\eta>0$.
Longer fixed blocks reach every fixed rate above half as in Theorem~\ref{mb:memory}.
These are sufficient bounds, with onset depending on the fixed parameters.
The same supplementary section gives the complete numerical examples and the contrasting nonprimitive GHZ and varying primitive memories.
The two-dimensional primitive-memory family that varies with $n$ has constant leakage at full threshold.

\paragraph{Further structured resources and preparation.}
Supplementary Sections~\ref{sec:resources} and~\ref{sec:depth} apply the marginal transfer to source networks with a large independent vertex set, to graph states through the cut rank, and to Haar-random resources, where one high-probability event controls all strategies; the graph-state and Haar examples give security above rate $3/4$.
For preparation from a product device marginal by disjoint fixed-arity gates on complete device groups, with all causal influence carried by those gates, constant leakage at full threshold requires depth $\Omega(\log n)$ and the GHZ construction attains depth $O(\log n)$.
The complete models, finite bounds, graph constructions and proofs are given there.

\section{Conclusion}\label{sec:conclusion}
The security of Shamir sharing against local quantum leakage depends quantitatively on the correlations available to the devices and the adversary.
For independent device--reference pairs, every fixed reconstruction rate above one half permits exponentially small chosen-secret distance at fixed output dimension.
Its exact exponent records both the emitted dimension and the Schmidt number of each pair, and matching attacks apply to every prescribed pair of distinct secrets.
For unrestricted independent device--reference pairs, the growing-output bounds permit logarithmically many emitted qubits per device and determine the leading exponent when the output dimension tends to infinity within the proved range.

Above one half, the support theorem determines the sharp fixed-fraction frontier when an arbitrary entangled support is tensored with single-device states elsewhere; its normalized exponent measures the gap between the reconstruction rate and the support fraction.
The selected-marginal criterion also gives security at every fixed rate $R>1-b_0/2$ when at least a fixed fraction $b_0$ of the devices has a marginal bounded above by a fixed multiple of a product, for every fixed output dimension and arbitrary finite adversary extensions.
The source-network and graph-state applications meet such marginal conditions, while Haar-random pure device states satisfy an approximate version with exponentially small error.
At each fixed output dimension, fixed stationary chains satisfying upper operator mixing, fixed finite-memory processes with primitive transfer channels, and canonical primitive matrix product states have exponential security at every fixed rate above one half, again allowing arbitrary finite adversary extensions.
Their decay exponents and length cutoffs may depend on the fixed process, output dimension and rate margin.
The varying two-dimensional primitive-memory examples admit attacks with chosen-secret distance bounded below by a positive constant at full threshold and distinguish this hypothesis from primitivity at each length.
In the preparation model with fixed-arity gates acting on complete device registers, logarithmic depth is necessary and sufficient in order of growth for such constant distance at full threshold with one-qubit outputs.

For one-qubit outputs, the sufficient rate extends to $1/2-5\cdot10^{-10}$ when the entire resource, including the reference, is fully separable.
The same rate holds for independent families of joint output--reference states satisfying $\sigma_i(y)\le2\gamma_i$ for every share, with each $\gamma_i$ a fixed density operator.
This includes joint qubit--reference outputs that are PPT for every share and classical-bit outputs with arbitrary initial pairs.

Several questions remain open.
The first is whether arbitrary share-dependent qubit channels are secure at every positive reconstruction rate, both for independent devices with trivial references and for arbitrary independent entangled references.
The sharp entangled-support tradeoff is open at or below one half, as is the range $2R-1\le\alpha<R$ above one half when only the outside devices have a product marginal.
In the two factored support models of Section~\ref{sup:section}, fix the output dimension and let $t/n\to R$ with $1/2<R<1$. If $\ell=t-|S|$ grows faster than $\log n$ and satisfies $\ell=o(n)$, does the worst-case distinguishing distance decay faster than every inverse polynomial in $n$?
For a specified fixed memory, sharper exponents and explicit numerical length cutoffs would make the fixed-process theorem more quantitative.
A criterion that finds, from a concise resource description, a large marginal bounded above in operator order by a small multiple of a product would extend the results to further explicit families.
Determining the optimal coefficient in the logarithmic preparation-depth requirement is another question.

\section*{Statements and Declarations}
\paragraph{Funding.} The authors received no funding for this work.
\paragraph{Competing interests.} The authors declare no competing interests.
\paragraph{Author contributions.} The authors contributed equally to this work. Aryeh Lev Zabokritskiy (Yohananov) supervised the research, and Isaac Barouch Essayag carried it out.
\paragraph{Data availability.} This theoretical work did not generate or analyse datasets.
\begin{samepage}
\paragraph{Use of AI-assisted tools.} The authors used OpenAI ChatGPT and Codex and Anthropic Claude to assist with developing and checking arguments and with drafting and editing. The authors take responsibility for the mathematical claims, references, and final text.
\end{samepage}

\clearpage
\appendix
\renewcommand{\thesection}{S\arabic{section}}
\renewcommand{\theHsection}{supp.\arabic{section}}
\section*{Supplementary Information}
\addcontentsline{toc}{section}{Supplementary Information}
\label{supp:start}
The following sections give the detailed estimates, constructions and additional resource applications supporting the main text.
All notation and the leakage model are those of the main article; theorem and equation references remain internal to this document.
Sections~\ref{supp:fourier}--\ref{supp:bh-details} give the independent-output estimates, and Sections~\ref{supp:support}--\ref{sec:depth} treat correlated resources and preparation depth.
Section~\ref{ppt:proof} proves the domination-two and PPT extension from Section~\ref{bh:section}.

\section{Fourier moments and explicit finite estimates}\label{supp:fourier}
The Fourier estimate separates reference dimension from emitted dimension. We give its full proof and then the finite-moment consequences.

\subsection{The dominated Fourier estimate}\label{supp:fourier-proof}
The following is the complete proof of Lemma~\ref{lem:fourier}.
\begin{proof}
Positivity and domination imply $\supp(\sigma_z)\subseteq\supp(\gamma)$.
Work on this support and use inverse powers only there.
For any supported operator $X$, Schatten H\"older gives
\[
 \norm{X}^2
 \le \|X\gamma^{-1/2}\|_2^2\|\gamma^{1/2}\|_2^2
 =\Tr(X\gamma^{-1}X^\dagger).
\]
Parseval applied entrywise to $\sigma_z\gamma^{-1/2}$ yields
\[
 \sum_a h(a)^2\le
 \E_z\Tr(\sigma_z\gamma^{-1}\sigma_z).
\]
Writing $A_z\coloneqq \gamma^{-1/2}\sigma_z\gamma^{-1/2}$, we have
$0\le A_z\le KI$, $\Tr(\gamma A_z)=1$, and
\[
 \Tr(\sigma_z\gamma^{-1}\sigma_z)=\Tr(\gamma A_z^2)
 \le K\Tr(\gamma A_z)=K.
\]
The zero coefficient is a density operator, so $h(0)=1$.

Fix $a\ne0$ and choose a unitary $U$ with
$h(a)=\operatorname{Re}\Tr(U\widehat{\sigma}(a))$.
Such a unitary exists by the polar decomposition, including for singular coefficients.
Write $U=\sum_j e^{i\phi_j}P_j$, and let
\[
 q_j\coloneqq \Tr(P_j\gamma),\qquad w_j(z)\coloneqq \Tr(P_j\sigma_z).
\]
Both sets of weights sum to one, and $0\le w_j(z)\le Kq_j$.
With $b_0\coloneqq \pi/K$ and $\ell\coloneqq \cos b_0$,
\[
 \cos\theta\,w\le \ell w+Kq(\cos\theta-\ell)_+
 \quad(0\le w\le Kq).
\]
The nonzero functional $\langle a,z\rangle$ is uniform on $\F_p$.
We need a bound uniform over the spectral phases $\phi_j$.
For the periodic function $v(\theta)\coloneqq (\cos\theta-\ell)_+$, the distributional second derivative consists of the density $-\cos\theta$ on $(-b_0,b_0)$ and two atoms of mass $\sin b_0$ at the endpoints.
Its total variation is $4\sin b_0$, since $b_0\le\pi/2$.
Consequently, its Fourier coefficients satisfy
\[
 |\widehat v(m)|\le\frac{4\sin b_0}{2\pi m^2}\quad(m\ne0),\qquad
 \widehat v(m)\coloneqq \frac1{2\pi}\int_0^{2\pi}v(\theta)e^{-im\theta}\,d\theta.
\]
The Fourier series converges absolutely.
Averaging over $p$ equally spaced points keeps only frequencies divisible by $p$, so, for every phase $\phi$,
\[
 \left|\frac1p\sum_{z=0}^{p-1}v(\phi+2\pi z/p)-\widehat v(0)\right|
 \le\sum_{m\ne0}|\widehat v(mp)|
 \le\frac{2\pi\sin b_0}{3p^2}.
\]
Moreover,
\[
 \frac1{2\pi}\int_0^{2\pi}v(\theta)\,d\theta
 =\frac{\sin b_0-b_0\cos b_0}{\pi}.
\]
Averaging the scalar inequality and summing over $j$ gives
\[
 h(a)\le \ell+
 K\left(\frac{\sin b_0-b_0\cos b_0}{\pi}+\frac{2\pi\sin b_0}{3p^2}\right)
 =\kappa(K)\left(1+\frac{2\pi^2}{3p^2}\right).
\]
The triangle inequality also gives $h(a)\le1$.
The higher moments follow by multiplying $h(a)^2$ by the pointwise bound to power $Q-2$.
\end{proof}

\subsection{Opposite frequencies and finite MDS sums}\label{supp:paired-mds}
Opposite frequencies have equal norms because the states are Hermitian, a symmetry also used in \cite[Section~3]{BKNY}.
Keeping this symmetry gives a sharper moment estimate than the last bound of Lemma~\ref{lem:fourier}.

\begin{lemma}[Moments of paired frequencies]\label{lem:paired}
Under the hypotheses of Lemma~\ref{lem:fourier}, assume $p$ is odd and put
\[
 c\coloneqq c_{K,p},\qquad B_K\coloneqq \frac{K-1}{2},\qquad
 j\coloneqq \left\lfloor\frac{B_K}{c^2}\right\rfloor,\qquad r\coloneqq B_K-jc^2.
\]
For every real $Q\ge2$,
\[
 \sum_{a\ne0}h(a)^Q\le M_{K,p}(Q)\coloneqq 2\bigl(jc^Q+r^{Q/2}\bigr)
 \le (K-1)c_{K,p}^{Q-2}.
\]
For $K=2,4$ and fixed $Q\ge2$, uniformly over primes $p>n$ as $n\to\infty$,
\[
 M_{K,p}(Q)\longrightarrow M_K(Q)\coloneqq 
 2\kappa(K)^Q+2\left(\frac{K-1}{2}-\kappa(K)^2\right)^{Q/2}.
\]
For these two values of $K$, $M_K(Q)<(K-1)\kappa(K)^{Q-2}$ whenever $Q>2$.
\end{lemma}
\begin{proof}
Hermiticity gives $\widehat\sigma(-a)=\widehat\sigma(a)^\dagger$, so $h(-a)=h(a)$.
Because $p$ is odd, the nonzero frequencies partition into pairs $\{a,-a\}$.
Choose one representative per pair and write $u_a\coloneqq h(a)^2$.
Then $0\le u_a\le c^2$ and $\sum u_a\le B_K$.
For $q\coloneqq Q/2\ge1$, moving mass from a smaller positive coordinate to a larger unsaturated coordinate cannot decrease $\sum u_a^q$, by convexity.
For an upper bound we may append zero coordinates and increase the total mass to $B_K$.
Repeatedly concentrating mass leaves $j$ coordinates equal to $c^2$ and one residual coordinate $r$.
Doubling their $q$-moment gives the first bound; $r^{Q/2}\le r c^{Q-2}$ gives the second.
This maximizes over a relaxed set of lists, giving an upper bound for the physically realizable lists.

For $K=2,4$, one has $\kappa(K)^2<B_K<2\kappa(K)^2$.
Since $0\le c_{K,p}-\kappa(K)\le2\pi^2\kappa(K)/(3n^2)$ for $p>n$, eventually $j=1$, uniformly over those primes.
The displayed limit follows, and its residual lies strictly between zero and $\kappa(K)^2$, which gives the strict inequality for $Q>2$.
For finite primes the floor and residual in $M_{K,p}$ remain necessary; in particular $j$ can be zero when $K=2$.
\end{proof}

The scalar information-set estimate needed for the finite sum is the following.
We use the MDS Brascamp--Lieb inequality of Horinaga and Yamakawa \cite[Theorem~3.1]{HY}, also obtained as the singleton case of \cite[Lemma~4]{BKNY}.
Its proof needs only injective coordinate projections, so we state it in that form.
\begin{lemma}[MDS Brascamp--Lieb inequality]\label{lem:mds}
Let $\mathcal V\subseteq\Sigma_1\times\cdots\times\Sigma_n$ be a set of words over finite alphabets whose projection onto every $k$-element coordinate set is injective, where $1\le k\le n$; for example, $\mathcal V=\mathcal C$ with alphabets $\F_p$.
For $q\coloneqq n/k$ and nonnegative $g_i:\Sigma_i\to\mathbb R$,
\[
 \sum_{a\in\mathcal V}\prod_i g_i(a_i)
 \le \prod_i\left(\sum_zg_i(z)^q\right)^{1/q}.
\]
\end{lemma}
\begin{proof}
We recall the proof of this specialization by generalized H\"older applied to all $k$-element coordinate subsets.
Put $N\coloneqq \binom nk$ and, for every $k$-element subset $J\subseteq[n]$, set
$F_J(a)\coloneqq \prod_{i\in J}g_i(a_i)^q$.
Each coordinate belongs to $Nk/n$ such subsets, so generalized H\"older gives
\[
 \sum_{a\in\mathcal V}\prod_i g_i(a_i)
 =\sum_{a\in\mathcal V}\prod_{|J|=k}F_J(a)^{1/N}
 \le\prod_{|J|=k}\left(\sum_{a\in\mathcal V}F_J(a)\right)^{1/N}.
\]
The projection of $\mathcal V$ onto $J$ is injective, hence its factor is
$\sum_{a\in\mathcal V}F_J(a)\le\prod_{i\in J}\sum_zg_i(z)^q$, with equality for $\mathcal V=\mathcal C$, whose projection is bijective.
Collecting the $Nk/n$ occurrences of each coordinate gives the claimed exponent $k/n=1/q$.
This also covers $k=n$, when there is only one subset; for $\mathcal V=\mathcal C$ equality then holds.
\end{proof}
\subsection{Explicit finite bounds}
Theorem~\ref{thm:half} gives the strongest asymptotic rate range used here.
The following elementary estimates have explicit finite parameters and give a constant prefactor at full threshold.
Each nonzero dual codeword has at least $t$ nonzero coordinates.
A weight parameter exploits this restriction while the MDS inequality bounds the sum.

\begin{theorem}[Finite bound for independent device--reference pairs]\label{thm:finite}
Suppose the initial resource is $\bigotimes_i\Omega_{A_iR_i}$ and each device emits one qubit.
Put
\[
 c_p\coloneqq c_{4,p}=\min\{1,\kappa(4)+2\pi^2\kappa(4)/(3p^2)\},\quad
 q\coloneqq \frac n{n-t+1}.
\]
If $q\ge2$, then for every $\lambda\ge1$,
\begin{equation}\label{eq:finite}
 \begin{aligned}
 \Adv(\rho_s,\rho_{s'})
 &\le \lambda^{-t}
 \left[\left(1+M_{4,p}(q)\lambda^q\right)^{n/q}-1\right]\\
 &\le \lambda^{-t}
 \left[\left(1+3c_p^{q-2}\lambda^q\right)^{n/q}-1\right].
 \end{aligned}
\end{equation}
\end{theorem}
\begin{proof}
In Lemma~\ref{lem:mds}, take $g_i(0)\coloneqq 1$ and
$g_i(a)\coloneqq \lambda h_i(a)$ for $a\ne0$.
Every nonzero dual codeword has weight at least $t$.
Its tilted product is therefore at least $\lambda^t$ times its original product.
Subtract the zero-codeword contribution, which is exactly one, and apply Lemma~\ref{lem:paired} with $K=4$; here $q\ge2$ forces $n\ge2$ and hence $p$ is odd.
Its comparison with the simpler moment bound gives the second inequality.
\end{proof}

Optimizing a fixed moment of the local Fourier coefficients gives a sufficient rate independent of both the field and the reference dimensions.
Define the rate by the following equation; its decimal value is used only for interpretation.

\begin{theorem}[An explicit scalar sufficient rate]\label{thm:ea}
Let $Q_{\rm ref}>2$ be the unique solution of
\begin{equation}\label{eq:root}
 M_4(Q_{\rm ref})=(Q_{\rm ref}-1)Q_{\rm ref}^{-Q_{\rm ref}/(Q_{\rm ref}-1)},
 \qquad \tau_{\rm ref}\coloneqq 1-\frac1{Q_{\rm ref}}.
\end{equation}
In the model with independent device--reference pairs and one emitted qubit per device, for every fixed $\eta>0$ there are $c>0,n_0$ such that
$\Adv(\rho_s,\rho_{s'})\le2^{-cn}$ whenever
$n\ge n_0$ and $t/n\ge\tau_{\rm ref}+\eta$.
The constants are uniform over every prime $p>n$, all evaluation points, all local maps, and all finite reference dimensions.
\end{theorem}
\begin{proof}
Write $B(\alpha)\coloneqq \alpha(1-\alpha)^{(1-\alpha)/\alpha}$.
Its logarithmic derivative is
$-\log(1-\alpha)/\alpha^2>0$.
As $Q$ increases from two to infinity, both exponential terms in $M_4(Q)$ strictly decrease, and their sum decreases from three to zero, while
$B(1-1/Q)=(Q-1)Q^{-Q/(Q-1)}$ increases from $1/4$ to one.
This proves existence and uniqueness.

If the asserted rate range is nonempty, set
\[
 \alpha\coloneqq \tau_{\rm ref}+\eta/2<1,\quad Q\coloneqq (1-\alpha)^{-1},\quad
 A\coloneqq M_4(Q)<B(\alpha),\quad
 x\coloneqq \frac{\alpha}{(1-\alpha)A}>1.
\]
For $n\ge2/\eta$, $q=n/(n-t+1)\ge Q$.
As in the finite bound, take $g_i(0)\coloneqq 1$ and $g_i(a)\coloneqq \lambda h_i(a)$ for $a\ne0$.
Monotonicity of the unnormalized sequence norms gives
$\|g_i\|_q\le\|g_i\|_Q$ in Lemma~\ref{lem:mds}.
Taking $\lambda\coloneqq x^{1/Q}$ yields
\[
 \frac1n\log\Adv(\rho_s,\rho_{s'})
 \le \frac1Q\left[\log(1+A_p(Q)x)-\alpha\log x\right],
 \quad A_p(Q)\coloneqq M_{4,p}(Q).
\]
By the uniform convergence in Lemma~\ref{lem:paired}, the right side tends uniformly over $p>n$ to
\[
 \frac{\alpha}{Q}\log\frac{A}{B(\alpha)}<0.
\]
This supplies $c,n_0$ independent of the strategy.
The fixed-$Q$ proof also covers $t=n$.
\end{proof}

Numerically, Equation~\eqref{eq:root} gives $\tau_{\rm ref}\approx0.9197212757$.
Using the simpler moment bound instead gives $0.9331803707$, the scalar rate of the estimate for independent blocks of two devices in \cite{BKNY}.
The improvement uses the paired moment estimate, while the common dominator supplies uniformity over arbitrary finite adversary references.

\section{Endpoint constructions and growing-output estimates}\label{supp:rates}

\subsection{Constant-prefactor endpoint bounds}\label{supp:rates-endpoint}
We use the output and Schmidt-number conventions of Section~\ref{sec:sharp}.
For $K\ge2$, put
\[
 C_K\coloneqq (K-1)\kappa(K)^{-2}e^{\pi^2/18}.
\]
Let $L_d(n,p,x)$ be the supremum of the chosen-secret distance over leakage strategies using independent device--reference pairs and over pairs of distinct secrets.

Let $L_{d,r}(n,p,x)$ denote the same supremum as $L_d(n,p,x)$, restricted to initial pairs of Schmidt number at most $r$.
\begin{theorem}\label{thm:sharp}
Fix integers $d\ge2$ and $1\le r\le d$.
For $t=n\ge2$ and every prime $p>n$,
\[
 \frac12\kappa(dr)^n
 \le L_{d,r}(n,p,x)\le\min\{1,C_{dr}\kappa(dr)^n\}.
\]
For unrestricted initial pairs, the corresponding bounds are
\[
 \frac12\kappa(d^2)^n
 \le L_d(n,p,x)\le\min\{1,C_{d^2}\kappa(d^2)^n\}.
\]
In both cases, for every prescribed pair of distinct secrets there is an allowed strategy attaining the lower bound.
The optimal exponential bases are therefore $\kappa(dr)$ and $\kappa(d^2)$, respectively, uniformly in the field and evaluation points.
\end{theorem}
\begin{proof}
A local channel preserves the reference marginal and cannot increase Schmidt number.
Lemma~\ref{lem:schmidt} supplies a share-independent dominator with $K=dr$ for the restricted class.
Lemma~\ref{lem:reference} supplies $K=d^2$ for unrestricted pairs, regardless of their input dimensions.
For either value of $K$, the simpler moment bound used in Theorem~\ref{thm:finite}, with $3$ replaced by $K-1$, specializes at $t=n$ to
\[
 \Adv\le(K-1)c_{K,p}^{n-2}
 \le(K-1)\kappa(K)^{n-2}
       \exp\!\left(\frac{2\pi^2(n-2)}{3p^2}\right)
 \le C_K\kappa(K)^n.
\]
The last step uses $p\ge n+1$ and $(n-2)/(n+1)^2\le1/12$, since
$(n+1)^2-12(n-2)=(n-5)^2$.

\medskip\noindent\emph{A matching construction.}
Give each device an independent pair $r^{-1/2}\sum_{j=0}^{r-1}|j\rangle_B|j\rangle_R$, with device dimension $d$ and reference dimension $r$.
The dense-coding construction of Yan and Wang \cite[arXiv version, Section~II]{YW} produces the $K=dr$ states
\[
 |\beta_{a,b}\rangle\coloneqq \frac1{\sqrt r}\sum_{j=0}^{r-1}
 e^{2\pi i bj/r}|j+a\bmod d\rangle_B|j\rangle_R,
 \qquad 0\le a<d,\quad0\le b<r.
\]
They are orthonormal: the reference inner product forces equal $j$, the device inner product then forces equal $a$, and the geometric sum forces equal $b$.
Each state has Schmidt rank $r$ and reference marginal $I_r/r$.
The required device unitary is a diagonal phase on the first $r$ basis vectors, extended to the whole space, followed by a cyclic shift in dimension $d$.
Enumerate these states as $|\beta_j\rangle$, $0\le j<K$.
The device can select a label from its share and send its half.

Lemma~\ref{lem:arc-labels} gives an offset $\theta$ and a state map $\sigma_\theta$ with $h_\theta(1)\ge\kappa(K)$.

Let $\lambda_i\ne0$ be the reconstruction coefficients, so
$s=\sum_i\lambda_iY_i$, and let device $i$ use $\sigma_\theta(\lambda_iY_i)$.
For a uniform secret at full threshold, the variables $Z_i\coloneqq \lambda_iY_i$ are independent and uniform.
Writing $\widehat\sigma_\theta(1)=\E_z\omega^{-z}\sigma_\theta(z)$, we therefore have exactly
\[
 \E_s\omega^{-s}\rho_s=\widehat\sigma_\theta(1)^{\otimes n}.
\]
With $\bar\rho\coloneqq \E_s\rho_s$,
\[
 \norm{\widehat\sigma_\theta(1)^{\otimes n}}
 \le\E_s\norm{\rho_s-\bar\rho}
 \le\E_{s,s'}\norm{\rho_s-\rho_{s'}}
 \le2\max_{s,s'}\Adv(\rho_s,\rho_{s'}).
\]
The left side is $h_\theta(1)^n\ge\kappa(K)^n$.

Finally, if this strategy attains the bound for $u\ne u'$, then for any prescribed $s\ne s'$ set
$a\coloneqq (u'-u)/(s'-s)$ and $b\coloneqq u-as$.
The local substitution $y_i\mapsto ay_i+b$ transports uniform Shamir sharing of $s,s'$ to that of $u,u'$, respectively, since $f\mapsto af+b$ is a bijection of the polynomial spaces.
Composing the maps with this substitution proves the final assertion.
\end{proof}
The rank-$d$ construction also supplies the unrestricted lower bound.
The upper and lower bounds have the same exponential base; their prefactors may differ.
For qubit outputs, $C_4\approx6.4041221842$.
The sharper field-dependent bound $3c_{4,p}^{n-2}$ is below $0.567$ at $(n,p)=(20,23)$ and below $0.022$ at $(50,53)$.

\begin{corollary}[Classical outputs with arbitrary references]\label{cor:classical}
For independent device--reference pairs and a classical output alphabet of size $d\ge2$, the optimal full-threshold exponential base is $\kappa(d)$, with arbitrary Schmidt number in the initial pairs.
More precisely, the supremum of the leakage distance lies between $\tfrac12\kappa(d)^n$ and $\min\{1,C_d\kappa(d)^n\}$.
The lower bound holds for every prescribed pair of distinct secrets.

For classical-bit outputs, let $Q_{\rm cl}>2$ be the unique solution of
\[
 M_2(Q_{\rm cl})
 =(Q_{\rm cl}-1)Q_{\rm cl}^{-Q_{\rm cl}/(Q_{\rm cl}-1)},
 \qquad \tau_{\rm cl}\coloneqq 1-\frac1{Q_{\rm cl}}.
\]
For every fixed $\eta>0$, there are $c>0,n_0$ such that
$\Adv(\rho_s,\rho_{s'})\le2^{-cn}$ whenever
$n\ge n_0$ and $t/n\ge\tau_{\rm cl}+\eta$.
The constants are uniform over all primes $p>n$, evaluation points, local maps, and arbitrary finite input and reference dimensions.
\end{corollary}
\begin{proof}
Each local output--reference state is classical--quantum and hence separable, so Lemma~\ref{lem:schmidt} supplies $K=d$ with $r=1$ at the output.
The upper proof above applies to these states, even when the initial pair is entangled.
For the lower bound, use the $d$ orthogonal classical labels in the rank-one construction.

For the rate statement, the same share-independent dominator has $K=2$.
Lemma~\ref{lem:fourier} therefore gives $\sum_{a\ne0}h(a)^2\le1$ and the cap $c_{2,p}$.
Lemma~\ref{lem:paired} sharpens the local moment bound to $M_{2,p}(Q)$.
Apply the fixed-$Q$ proof of Theorem~\ref{thm:ea} with $M_{4,p}(Q)$ replaced by $M_{2,p}(Q)$ and $M_4(Q)$ by $M_2(Q)$.
Now $M_2(Q)$ strictly decreases from one to zero, so the same monotonicity argument gives the unique root $Q_{\rm cl}>2$ and the negative limiting exponent proves the uniform bound.
\end{proof}
In particular, a classical bit has base $2/\pi$, while an entanglement-assisted qubit has base $2\sqrt2/\pi$.
The paired estimate gives $\tau_{\rm cl}\approx0.7172542188$.
The simpler moment bound gives $0.7333927776$, agreeing with the one-qubit leakage bound without prior entanglement in \cite[Section~5.4, Theorem~16 and Corollary~1]{BKNY}; Corollary~\ref{cor:classical} asserts the paired rate $\tau_{\rm cl}$ only for classical-bit outputs.
The present classical-output statements permit arbitrary finite references and arbitrary entanglement within each initial pair.
These scalar estimates provide explicit sufficient rates; the weighted bound gives the wider above-half range in Corollary~\ref{cor:rate}.

\subsection{Quantitative list recovery and the qubit budget}\label{supp:growing}
We supply the full finite-parameter argument for Theorem~\ref{lq:theorem}.
To prove the theorem, we make the cap and squared-mass dependence in Lemma~\ref{lem:weighted} explicit.
The coding input is list recovery for GRS codes, in a fixed-list form and a growing-list form; later sections use the same lemma.

\begin{lemma}[GRS list recovery]\label{gv:quantitative-list}
\begin{enumerate}
\item For fixed $\xi>0$ and a fixed list-size bound $\ell$, there are $A,b,n_1$ such that, for $n\ge n_1$, a prime-field GRS code of length $n$ and dimension $k\le(1-\xi)n$ has at most $A n^b$ codewords agreeing with lists of size at most $\ell$ on at least $k+\lceil\xi n\rceil$ coordinates.
\item Fix $0<\eta<1$.
There is $C_\eta\ge1$, depending only on $\eta$, such that for every
integer $L\ge1$, every $n\ge C_\eta L^{9/\eta}$, every prime $p>n$,
and every length-$n$ GRS code of dimension
$1\le k\le(1-\eta)n$, the number of codewords agreeing with prescribed
coordinate lists of size at most $L$ in at least
$k+\lceil\eta n\rceil$ coordinates is at most
\[
 n^{C_\eta L^{3/\eta}}.
\]
\end{enumerate}
The constants are independent of the prime, dimension, evaluation points,
nonzero coordinate multipliers, and list contents.
\end{lemma}
The first part is the list-cardinality consequence of Jeronimo \cite[Theorem~1.1 and Corollary~8.1]{Jeronimo}.
The second part is Nguyen's quantitative form \cite[Theorem~4 and Appendix~A]{Nguyen}, derived from Jeronimo's interpolation and counting argument \cite[Proposition~3.7]{Jeronimo}, with error radius $\lfloor n-k-\eta n\rfloor$.
For GRS coordinate multipliers, divide the lists by the corresponding nonzero multipliers to obtain the statement for ordinary Reed--Solomon codes (all multipliers one); cardinalities and agreement counts are unchanged.
The single constant $C_\eta$ applies to all $L$ subject to the displayed length condition.

Taking $\eta=1/4$ in the second part gives an absolute constant $C\ge1$ such that, for every integer $\ell\ge1$, every $n\ge C\ell^{36}$, and every prime $p>n$, a length-$n$, dimension-$k$ GRS code with $1\le k\le3n/4$ has at most
\begin{equation}\label{lq:list}
 n^{C\ell^{12}}
\end{equation}
codewords agreeing with prescribed coordinate lists of size at most $\ell$ on at least $k+\lceil n/4\rceil$ coordinates.
The constant $C$ is independent of $\ell$ and $p$.
This explicit dependence is needed when the Fourier mass increases with $n$.

\begin{lemma}[A quantitative weighted sum]\label{lq:weighted}
Let $0<\eta<1/2$, $n\ge\max\{4,2/\eta\}$, and $(1/2+\eta)n\le t\le n$.
Let $\mathcal G$ be a length-$n$ GRS code over $\F_p$ of dimension $k\coloneqq n-t+1$, where $p>n$.
Suppose $K\ge2$ and $w_i:\F_p\to[0,1]$ satisfy
\[
 w_i(0)=1,\qquad
 \sum_{a\in\F_p}w_i(a)^2\le K,\qquad
 w_i(a)\le\kappa(K)\quad(a\ne0).
\]
Set
\[
 \theta\coloneqq (4\sqrt K)^{-2/\eta},\qquad
 \ell\coloneqq \left\lceil1+K\theta^{-2}\right\rceil.
\]
If $n\ge C\ell^{36}$, with $C$ as in Equation~\eqref{lq:list}, then
\begin{equation}\label{lq:weighted-bound}
 \sum_{a\in\mathcal G\setminus\{0\}}\prod_{i=1}^n w_i(a_i)
 \le n^{C\ell^{12}}\kappa(K)^t+2^{-n}.
\end{equation}
\end{lemma}
\begin{proof}
We track the parameters in the weighted Fourier thresholding method of Sun and Wootters \cite[Lemma~1.15 and Section~4.2]{SW}.
Put $m\coloneqq \lfloor\eta n\rfloor$ and
\[
 S_i\coloneqq \{0\}\cup\{a\ne0:w_i(a)>\theta\}.
\]
The squared-mass bound gives $|S_i|\le\ell$.
Call a coordinate light for a word $a$ when $a_i\notin S_i$.
Such a coordinate has weight at most $\theta$.

A word with fewer than $m$ light coordinates agrees with these lists at least $n-m+1$ times.
Since
\[
 n-m+1-k-\lceil n/4\rceil
 =t-m-\lceil n/4\rceil
 \ge n/2-\lceil n/4\rceil\ge0,
\]
and $k\le n/2$, Equation~\eqref{lq:list} bounds the number of these words by $n^{C\ell^{12}}$.
Every nonzero word has at least $t$ nonzero coordinates, so its product of weights is at most $\kappa(K)^t$.

For each remaining word, mark $m$ light coordinates.
There are at most $2^n$ choices of the marked set $M$.
Two disjoint sets $A,B$ of size $k$ can be chosen outside $M$, because
\[
 n-m-2k=2t-n-m-2\ge\eta n-2\ge0.
\]
Extract the factor $\theta^m$ and discard the other unassigned factors, all of which are at most one.
Cauchy--Schwarz gives
\begin{align*}
 \sum_{\substack{a\in\mathcal G\\M\text{ is light for }a}}
       \prod_iw_i(a_i)
 &\le\theta^m
 \left(\sum_{a\in\mathcal G}\prod_{i\in A}w_i(a_i)^2\right)^{1/2}
 \left(\sum_{a\in\mathcal G}\prod_{i\in B}w_i(a_i)^2\right)^{1/2}\\
 &\le\theta^m K^k.
\end{align*}
The last step uses the bijection from an MDS code to its coordinates on any information set of size $k$; each sum is a product of $k$ coordinate sums.
Since $m\ge\eta n/2$ and $k\le n/2$, summing over the marked sets costs at most
\[
 2^n\theta^m K^k
 \le2^n(4\sqrt K)^{-n}K^{n/2}=2^{-n}.
\]
Every word in this class is counted at least once.
Adding the two contributions proves the lemma.
\end{proof}

\begin{proof}[Proof of Theorem~\ref{lq:theorem}]
We first give a bound for any product of local state families satisfying $\sigma_i(y)\le K\gamma_i$, where $K\ge2$ and each $\gamma_i$ is a fixed density operator.
Lemma~\ref{lem:fourier} gives the local Fourier norms $h_i(a)$ with
\[
 h_i(0)=1,\qquad \sum_a h_i(a)^2\le K,
 \qquad h_i(a)\le\min\{1,\kappa(K)\lambda_p\}\quad(a\ne0).
\]
These are finite inequalities for every $K$.
Define $w_i(0)\coloneqq 1$ and $w_i(a)\coloneqq h_i(a)/\lambda_p$ for $a\ne0$.
The normalized weights satisfy Lemma~\ref{lq:weighted}, including
$\sum_a w_i(a)^2\le1+(K-1)/\lambda_p^2\le K$.
The Shamir dual has dimension $n-t+1$, so the Fourier reduction and Equation~\eqref{lq:weighted-bound} give
\begin{equation}\label{lq:finite}
 \Adv(\rho_s,\rho_{s'})
 \le\lambda_p^n
 \left[n^{C\ell^{12}}\kappa(K)^t+2^{-n}\right],
 \qquad
 \lambda_p^n\le\exp\!\left(\frac{2\pi^2}{3n}\right),
\end{equation}
whenever the lemma's length conditions hold.

We verify those conditions uniformly for $2\le K\le n^{\eta/100}$.
Put $B_\eta\coloneqq 4^{4/\eta}+2$.
The definition of $\ell$ implies
\[
 \ell\le B_\eta K^{1+2/\eta},\qquad
 \ell^{36}\le B_\eta^{36}n^{9/10},\qquad
 \ell^{12}\log n\le B_\eta^{12}n^{3/10}\log n.
\]
Thus $n\ge C\ell^{36}$ eventually holds uniformly throughout this range.
For $K\ge2$, the alternating sine estimate at $x\coloneqq \pi/K\le\pi/2$ yields
\[
 \kappa(K)\le1-\frac{x^2}{6}+\frac{x^4}{120}
 \le1-\frac{x^2}{8}
 \le1-K^{-2}\le e^{-K^{-2}}.
\]
Here the middle inequalities use $x^2<5$ and $\pi^2>8$.
Consequently, with $E\coloneqq C\ell^{12}\log n$, the first term in brackets in Equation~\eqref{lq:finite} is at most
\[
 \exp\!\left(E-\frac{n}{2K^2}\right).
\]
The bound $n/K^2\ge n^{1-\eta/50}\ge n^{99/100}$ makes $E\le n/(4K^2)$ for all sufficiently large $n$, uniformly in $K$.
Both bracketed terms are then at most $\exp(-n/(4K^2))$.
Increasing the same uniform length cutoff absorbs $2\lambda_p^n$ and gives
\[
 \Adv(\rho_s,\rho_{s'})\le\exp\!\left(-\frac{n}{8K^2}\right).
\]
For $q\ge1$, Lemma~\ref{lem:reference} supplies $K=d^2=4^q\le n^{\eta/100}$, proving Equation~\eqref{lq:security}.
For $q=0$ the emitted systems are trivial and the reference marginal is unchanged, so the distance is zero.
\end{proof}

\subsection{Uniform comparison with the lower bound}\label{supp:growing-rate-proof}
The following is the complete proof of Corollary~\ref{lq:rate}.
\begin{proof}
The three dominators used in Corollary~\ref{cor:rate} are valid at every finite dimension.
Thus Equation~\eqref{lq:finite} applies with the corresponding $K$.
The uniform onset in the preceding proof ensures its finite length conditions for every permitted $K$.
The selected-$t$ construction in Corollary~\ref{cor:rate} gives $L\ge\tfrac12\kappa(K)^t$ for every permitted prime.
With $E\coloneqq C\ell^{12}\log n$, we therefore have
\[
 \frac12\kappa(K)^t\le L
 \le\lambda_p^n\left[e^E\kappa(K)^t+2^{-n}\right].
\]
Since $\kappa(K)\ge\kappa(2)=2/\pi$ and $t\le n$,
$2^{-n}/\kappa(K)^t\le(\pi/4)^n$.
Taking logarithms shows that the additive error is bounded by
$E+n\log\lambda_p+O(1)$, which proves Equation~\eqref{lq:additive-rate}.
Also $t\log(1/\kappa(K))\ge n/(2K^2)$.
The ratio of the error to this quantity is
\[
 O_\eta\!\left(K^{14+24/\eta}\frac{\log n}{n}\right)
 =O_\eta\!\left(n^{-76/100+14\eta/100}\log n\right)
 =O_\eta(n^{-69/100}\log n),
\]
proving Equation~\eqref{lq:relative-rate}.
Finally, $-\log\kappa(K)=\pi^2/(6K^2)+O(K^{-4})$ gives the last assertion.
\end{proof}
\section{Supporting estimates for qubit leakage below one half}\label{supp:bh-details}
This section supplies the truncated Fourier estimate, the real scalar contraction, and the bias calculation used in Section~\ref{bh:section}.
The centered support operators $T_A(z)$, Bloch vectors, and Gram kernels retain the definitions given there.

\subsection{The truncated Fourier estimate}
The Fourier weights in \eqref{bh:fourier-cap} and the exact-support formula \eqref{bh:exact-support} reduce the moderate supports to the following weighted code sum.

We need the truncated version of the weighted estimate, in which the code dimension may be below $n/2$ but the support cutoff is below $2\nu$.
The following form follows from Sun and Wootters \cite[Lemma~1.15 and Section~4.2]{SW}; we give the weighted Fourier thresholding argument to specify its uniformity in the field.

\begin{lemma}[Truncated weighted dual sum]\label{bh:truncated}
Fix $0<\gamma<1$, $\Gamma>0$, $c_0>0$, and $\eta>0$.
There are constants $A,b,n_1>0$, depending only on these four parameters, such that the following holds.
Let $p>n\ge n_1$ be prime, let $\mathcal G$ be a GRS code of dimension $\nu\ge c_0n$, and let the integer cutoff $\nu\le m\le n$ satisfy $2\nu-m\ge\eta n$.
If $w_i:\F_p^\times\to[0,\gamma]$ satisfy $\sum_{a\ne0}w_i(a)^2\le\Gamma$, then
\begin{equation}\label{bh:truncated-bound}
 \sum_{\substack{0\ne y\in\mathcal G^\perp\\|\supp(y)|\le m}}
       \prod_{i\in\supp(y)}w_i(y_i)
 \le (A n^b+1)\gamma^\nu.
\end{equation}
The constants are independent of the field, code parameters, and weights.
\end{lemma}
\begin{proof}
The coding input is the first part of Lemma~\ref{gv:quantitative-list}.

The case $m=\nu$ is empty by the dual minimum distance $\nu+1$.
Otherwise put $s\coloneqq \lfloor\nu-m/2\rfloor$.
After increasing $n_1$, we have $s\ge\eta n/4>0$ and $s\le2\nu-m$.
Choose a fixed $\tau\in(0,\gamma)$ so small that
\[
 3(1+\Gamma)\tau^{\eta/4}\le\gamma,
\]
and form the coordinate lists
$H_i\coloneqq \{0\}\cup\{a\ne0:w_i(a)>\tau\}$.
They have size at most $1+\Gamma/\tau^2$, independently of $p$.
Call a nonzero coordinate light when its weight is at most $\tau$.

A word with fewer than $s$ light coordinates agrees with these lists on at least $n-s+1$ coordinates.
The dual GRS code has dimension $k\coloneqq n-\nu$, and
$\nu-s\ge m/2\ge c_0n/2$.
Thus the quoted list bound applies with fixed slack $\xi=c_0/4$, for sufficiently large $n$.
There are at most $A n^b$ such words.
Each nonzero dual word has weight at least $\nu+1$, so their total contribution is at most $A n^b\gamma^{\nu+1}$.

For every remaining word of weight at most $m$, select a container $S\subseteq[n]$ of size $m$ containing its support and a set $L\subseteq S$ of $s$ light nonzero coordinates.
There are at most $3^n$ possible pairs $(S,L)$.
Restrict to $S$ the subcode of $\mathcal G^\perp$ that vanishes outside $S$. This shortening has dimension $m-\nu$ and is MDS.
Because $s\le2\nu-m$, the set $S\setminus L$ contains disjoint sets $J_1,J_2$, each of size $m-\nu$; each is an information set of this shortened code.
Extend the weights by $w_i(0)\coloneqq 1$.
For a fixed $(S,L)$, bound the $L$ factors by $\tau^s$, discard the other factors outside $J_1\cup J_2$, which are at most one, and use Cauchy--Schwarz.
Bijectivity of the two information projections gives
\begin{align*}
 &\sum_{\substack{y\in\mathcal G^\perp,\ \supp(y)\subseteq S\\
                   y_i\ne0,\ w_i(y_i)\le\tau\ (i\in L)}}
        \prod_{i\in S}w_i(y_i)\\
 &\qquad\le\tau^s
   \prod_{j=1}^2\left(
      \sum_{\substack{y\in\mathcal G^\perp\\\supp(y)\subseteq S}}
      \prod_{i\in J_j}w_i(y_i)^2\right)^{1/2}
 \le\tau^s(1+\Gamma)^{m-\nu}.
\end{align*}
The sum over all containers and selected light sets is at most
$3^n\tau^s(1+\Gamma)^{m-\nu}\le\gamma^n\le\gamma^\nu$.
Combining the two classes proves \eqref{bh:truncated-bound}.
\end{proof}

In our application, \eqref{bh:fourier-cap} supplies $\gamma=2^{-1/2}$ and $\Gamma=1$.
Thus any fixed positive gap $(2\nu-m)/n$ gives an exponentially small contribution from the supports in \eqref{bh:low-reduction}.
The next step keeps the full centered tensor on each of the larger supports.

\subsection{Real functions on the product group}
The squared norm of a real tensor moment will introduce two independent codewords.
Its coordinate functions consequently live on $\F_p^2$.
We establish the scalar estimates on $G\coloneqq \F_p^r$, where $r\ge1$, with normalized averages throughout.
For $f:G\to\mathbb C$, write
\[
 \|f\|_{L^2(G)}^2\coloneqq \E_x|f(x)|^2,\qquad
 \|f\|_{U^j(G)}^{2^j}
 \coloneqq \E_{x,h_1,\ldots,h_j}
   \prod_{\epsilon\in\{0,1\}^j}
   \mathfrak C^{|\epsilon|}
   f(x+\epsilon_1h_1+\cdots+\epsilon_jh_j),
\]
where $\mathfrak C$ is complex conjugation; for $j\ge2$, $\|f\|_{U^j(G)}$ are the Gowers uniformity norms.
In particular, $\|f\|_{U^2}^4=\sum_{\xi\in G}|\widehat f(\xi)|^4$ for the normalized character transform on $G$.

The real Fourier-pairing argument in Nguyen \cite[Lemma~16]{Nguyen} gives the following form for bounded mean-zero functions.
We record the proof because its product-group domain is needed below.

\begin{lemma}\label{bh:real-gap}
If $p$ is odd, $f:\F_p^r\to\mathbb R$, $\E f=0$, and $|f|\le1$, then
\[
 \|f\|_{U^3(\F_p^r)}\le\gamma_*,
 \qquad \gamma_*=(3/4)^{1/8}<1.
\]
\end{lemma}
\begin{proof}
For a real function $q$ on $G$, set $M\coloneqq \E q$ and $S\coloneqq \E q^2$.
Every nonzero frequency occurs in a distinct pair $\{\xi,-\xi\}$ of equal magnitude.
Parseval therefore gives
\[
 \|q\|_{U^2}^4
 \le M^4+\frac12(S-M^2)^2
 \le\frac12 S^2+\frac12 M^4,
\]
where the last inequality uses $M^2\le S$.
Apply this to $q_h(x)\coloneqq f(x+h)f(x)$ and put $M_h\coloneqq \E_x q_h(x)$.
We have $S_h\le1$, $|M_h|\le1$, and
\[
 \E_h M_h^2=\|f\|_{U^2}^4
 \le\tfrac12\|f\|_{L^2}^4\le\tfrac12.
\]
The Gowers recursion now yields
\[
 \|f\|_{U^3}^8
 =\E_h\|q_h\|_{U^2}^4
 \le\tfrac12+\tfrac12\E_h M_h^4
 \le\tfrac12+\tfrac12\E_h M_h^2
 \le\tfrac34.
\]
\end{proof}

We also need a product-group form of the MDS von Neumann inequality of Sun and Wootters \cite[Lemmas~2.8--2.9]{SW}, which they state over $\F_p$; we give the proof.
Let $F_i:G^\nu\to G$ apply the $i$th row of an $n\times\nu$ scalar MDS generator matrix componentwise.
For $J\subseteq[n]$, write $F_J\coloneqq (F_j)_{j\in J}$.
The MDS condition says that every at-most-$\nu$ collection of scalar rows is independent.
It therefore makes the corresponding map onto $G^J$ surjective, component by component.

\begin{lemma}\label{bh:product-mds}
Let $L_i:G\to\mathbb C$ and $c_i\in G$.
If $2|J|\le3\nu-n$, then
\begin{equation}\label{bh:product-mds-bound}
 \left|\E_{x\in G^\nu}\prod_{i=1}^n L_i(F_i(x)+c_i)\right|
 \le\prod_{j\in J}\|L_j\|_{U^3(G)}
       \prod_{i\notin J}\|L_i\|_\infty.
\end{equation}
\end{lemma}
\begin{proof}
Translate the functions to absorb $c_i$; both norms are translation invariant.
Divide the functions outside $J$ by their supremum norms, the case of a zero norm being immediate.
Put $V\coloneqq G^\nu$ and $H(u)\coloneqq \prod_{j\in J}L_j(u_j)$ on $G^J$.
The condition $2|J|\le3\nu-n$ is equivalent to
$n-|J|\le3(\nu-|J|)$.
Partition $[n]\setminus J$ into three sets $I_1,I_2,I_3$ of size at most $\nu-|J|$, and set
\[
 K_a\coloneqq \ker F_{I_a},\qquad
 Q_a(x)\coloneqq \prod_{i\in I_a}L_i(F_i(x))\quad(1\le a\le3).
\]
Each $Q_a$ is $K_a$-invariant and bounded in magnitude by one.
The map $F_J:K_a\to G^J$ is onto: for prescribed values on $J$ and zero values on $I_a$, the independent scalar rows indexed by $J\cup I_a$ admit a solution in each component of $G$.

For a subgroup $K\le V$ and a $K$-invariant function $Q$ with $|Q|\le1$, averaging on cosets of $K$ and applying Cauchy--Schwarz give
\begin{equation}\label{bh:invariant-cs}
 |\E_x Q(x)P(x)|^2
 \le\E_x|\E_{h\in K}P(x+h)|^2
 =\E_{x\in V,\,h\in K}P(x)\overline{P(x+h)}.
\end{equation}
Define $\Delta_hP(x)\coloneqq P(x)\overline{P(x+h)}$.
Starting from
$\Lambda\coloneqq \E_xH(F_Jx)Q_1(x)Q_2(x)Q_3(x)$,
apply \eqref{bh:invariant-cs} first with $K_1$.
For each subsequent application, first use Cauchy--Schwarz in the increments already introduced, then use \eqref{bh:invariant-cs} with $K_a$ in the $x$ variable.
Indeed, every iterated derivative of $Q_a$ remains $K_a$-invariant and bounded by one.
Inductively, for $a=1,2,3$, this gives
\[
 |\Lambda|^{2^a}
 \le\E_{\substack{h_1\in K_1,\ldots,h_a\in K_a\\x\in V}}
   \Delta_{h_a}\cdots\Delta_{h_1}(H\circ F_J)(x)
   \prod_{b=a+1}^3\Delta_{h_a}\cdots\Delta_{h_1}Q_b(x).
\]
At each induction step, the right side after the new application is a nonnegative average of squared coset means, as in \eqref{bh:invariant-cs}.
At $a=3$, only the eight-vertex cube of $H\circ F_J$ remains.
The four variables $F_Jx,F_Jh_1,F_Jh_2,F_Jh_3$ are independent and uniform on $G^J$, by the surjections proved above.
Normalized averages thus identify the last bound with $\|H\|_{U^3(G^J)}^8$.
Its cube average factors over $j\in J$, giving
$\|H\|_{U^3(G^J)}^8=\prod_{j\in J}\|L_j\|_{U^3(G)}^8$.
Taking eighth roots and restoring the supremum norms proves the claim.
For $J=\varnothing$, the same conclusion is the elementary bound by the product of supremum norms.
\end{proof}

In particular, the proof applies to two independently chosen codewords by taking $r\coloneqq 2$.
The underlying scalar code retains dimension $\nu$; the two copies are encoded by the componentwise maps into $\F_p^2$.

\subsection{The bias estimate and the fixed constants}\label{supp:bh-bias}
Lemma~\ref{bh:real-gap} applies to $K_i/(1+r_i)^2$ on $\F_p^2$.
Thus the scalar estimates above apply to the Gram kernels of \eqref{bh:gram} and give \eqref{bh:small-bias}.

For the complementary bias range, the two-information-set argument of Nguyen \cite[Lemma~17]{Nguyen} has a direct qubit variance form.
Partition $[n]$ into $I_1,I_2,E_0$ with
$|I_1|=|I_2|=\nu$ and $|E_0|=n-2\nu$.
The eigenvalues of $(u_i(x)\cdot\sigma)/2$ are $\pm|u_i(x)|/2$, so
$\|\rho_i(x)-\bar\rho_i\|_1=|u_i(x)|\le2$.
The triangle inequality and Cauchy--Schwarz yield
\begin{align*}
 \|T_A(z)\|_1
 &\le\E_X\prod_{i\in A}|u_i(X_i)|\\
 &\le2^{n-2\nu}\prod_{j=1}^2
   \left(\E_X\prod_{i\in A\cap I_j}|u_i(X_i)|^2\right)^{1/2}
 \le2^{n-2\nu}\prod_{i\in A\setminus E_0}\sqrt{1-r_i^2}.
\end{align*}
The last step uses independence within each MDS marginal $X_{I_j}$ and the variance identity
$\E_x|u_i(x)|^2=\E_x|v_i(x)|^2-r_i^2\le1-r_i^2$.
Since $\sqrt{1-r_i^2}\le e^{-r_i^2/2}$ and
$\sum_{i\in A\cap E_0}r_i^2\le n-2\nu$, we obtain Equation~\eqref{bh:large-bias}.

We combine these two estimates with the support cutoff, following Nguyen \cite[Proposition~4 and Section~6.3]{Nguyen} and the low/high support division of Sun and Wootters \cite[Section~4.3]{SW}.
For $c\coloneqq \nu/n$ and $\beta>0$, split according to whether $B_A\le\beta n$ or $B_A>\beta n$.
Equations~\eqref{bh:small-bias} and~\eqref{bh:large-bias} then give
\begin{equation}\label{bh:combined-bias}
 \begin{gathered}
 \|T_A(z)\|_1\le\exp\!\left(nF(c,\beta)+\lambda_*/2\right),
 \\[1mm]
 F(c,\beta)\coloneqq \max\left\{
  \frac{3c-1}{4}\log\gamma_*+\sqrt\beta,
  (1-2c)(\log2+\tfrac12)-\frac\beta2
 \right\},
 \end{gathered}
\end{equation}
where $\lambda_*\coloneqq -\log\gamma_*=\log(4/3)/8$.
The term $\lambda_*/2$ bounds the floor error in $j_0$ uniformly.

\begin{lemma}[The large supports]\label{bh:large-support-sum}
Let $c_{\rm q}\le c=\nu/n\le1/2$, and fix
\[
 \beta\coloneqq 2^{-18},\qquad \epsilon\coloneqq 10^{-8},\qquad
 m\coloneqq \lfloor(1-\epsilon)n\rfloor.
\]
If Equation~\eqref{bh:combined-bias} holds for every support $A$ with $|A|>m$, then
\[
 \sum_{|A|>m}\|T_A(z)\|_1\le(n+1)\exp(-10^{-6}n+\lambda_*/2).
\]
\end{lemma}
\begin{proof}
Put $\delta\coloneqq 10^{-9}$, so that $c_{\rm q}=1/2-\delta$.
Let $H(\epsilon)\coloneqq -\epsilon\log\epsilon-(1-\epsilon)\log(1-\epsilon)$.
An integer support size larger than $m$ has complement size less than $\epsilon n$.
Hence there are at most $(n+1)e^{nH(\epsilon)}$ such supports, and \eqref{bh:combined-bias} bounds their total trace norm by
\begin{equation}\label{bh:high-sum}
 (n+1)\exp\!\left(n[H(\epsilon)+F(c,\beta)]+\lambda_*/2\right).
\end{equation}

We verify a fixed negative exponent using elementary inequalities.
The integral bounds $1/4<\log(4/3)<1/3$ give $1/32<\lambda_*<1/24$.
Both entries of $F(c,\beta)$ decrease with $c$, so it suffices to use $c=1/2-\delta$.
Their respective upper bounds are
\[
 -\frac1{256}+\frac\delta{32}+\frac1{512}
 <-\frac1{1024},\qquad
 3\delta-2^{-19}.
\]
Also $H(\epsilon)\le\epsilon\log(e/\epsilon)<25\cdot10^{-8}$, using $\log10<3$.
Since $2^{-19}>1.9\cdot10^{-6}$, it follows that
\[
 H(\epsilon)+F(c,\beta)<-10^{-6}
 \quad(c_{\rm q}\le c\le1/2).
\]
Thus \eqref{bh:high-sum} is at most
$(n+1)\exp(-10^{-6}n+\lambda_*/2)$.
\end{proof}

\section{Marginal consequences and support estimates}\label{supp:support}

\subsection{Product and approximate product marginals}\label{supp:marginal-consequences}
The following consequences use the exact-conditioning and reference-filter arguments of Section~\ref{sec:correlation}.
\begin{corollary}[Product marginal]\label{cor:product}
Under the residual-threshold condition $n-|I|<t$, the same security bound as for independent device--reference pairs holds whenever
$\Omega_{A_I}=\bigotimes_{i\in I}\tau_i$, after revealing the complement shares.
The full initial state may be arbitrary and multipartite.
If $|I|\ge b_0n$ for fixed $b_0\in(0,1]$, this gives exponential security whenever
$t/n\ge1-b_0/2+\eta$ for fixed $\eta>0$.
The same rate holds for every fixed output dimension, including classical outputs with arbitrary finite references.
\end{corollary}
For the rate statement, put $m\coloneqq |I|$ and $r\coloneqq n-m$.
Then $t-r-m/2\ge\eta n\ge\eta m$.
Theorem~\ref{thm:half} applies to the residual instance, uniformly over its field and evaluation points.
The elementary rates $1-b_0(1-\tau_{\rm ref})+\eta$ and $1-b_0(1-\tau_{\rm cl})+\eta$ also follow directly from Theorem~\ref{thm:ea} and Corollary~\ref{cor:classical}, respectively.

Approximate product marginals can be treated with an additive error.
For normalized states define the root fidelity
$F(\rho,\sigma)\coloneqq \|\sqrt\rho\sqrt\sigma\|_1$ and the purified distance
$P(\rho,\sigma)\coloneqq \sqrt{1-F(\rho,\sigma)^2}$.
If a normalized $\bar\eta_{A_I}$ satisfies
$P(\Omega_{A_I},\bar\eta_{A_I})\le\delta$ and
$\bar\eta_{A_I}\le\Lambda\bigotimes_i\tau_i$, the extension property of purified distance gives
\begin{equation}\label{eq:smooth}
 \Adv(\rho_s,\rho_{s'})
 \le \min\{1,2\delta+\Lambda\varepsilon_{\rm ref}(m,t-r,p)\}.
\end{equation}
Indeed, lift $\bar\eta_{A_I}$ to a $\delta$-close extension of the full resource, apply Theorem~\ref{thm:domination}, and use two triangle-inequality terms after the secret-dependent channels.
The marginal comparison uses the selected marginal of the full initial state.

\subsection{The direct support upper bound}\label{supp:support-upper-proof}
We give the complete proof of Theorem~\ref{sup:upper}, with its two resource premises kept separate.
\begin{proof}
We adapt the weighted Fourier thresholding argument of Sun and Wootters \cite[Lemma~1.15 and Section~4.2]{SW}, applying its cutoff only to the outside coordinates.
The support is one block whose squared Fourier mass may grow exponentially with $r$.

First consider Equation~\eqref{sup:coarse-product}.
A product of purifications of $\Omega_{A_S}$ and the states $\tau_i$, $i\in O$, purifies the entire device marginal.
Every external extension is obtained from this purification by a CPTP map on its reference.
That map is independent of the shares, commutes with the device channels, and contracts trace distance.
It therefore suffices to prove the bound for a product of one support--reference resource and the outside device--reference resources.
Under Equation~\eqref{sup:no-outside-reference}, the outside state families have trivial reference registers.

Let $\sigma_S(z_S)$ contain the support outputs and their reference, and set
\[
 H(v)\coloneqq \norm{\E_{z_S}\omega^{-\langle v,z_S\rangle}\sigma_S(z_S)},
 \qquad v\in\F_p^S.
\]
Write $h_i(a)$ for the outside Fourier norms.
The support output has dimension at most $d^r$ and its reference marginal is fixed.
Lemmas~\ref{lem:reference} and~\ref{lem:fourier}, applied with domination constant $d^{2r}$, give
\begin{equation}\label{sup:block-mass}
 H(0)=1,\qquad 0\le H(v)\le1,
 \qquad\sum_v H(v)^2\le d^{2r}.
\end{equation}
For $r=0$, use $H=1$ and squared mass one.
For outside devices the domination constant is $K=d^2$ in the first model and $K=d$ in the second, since a state on a $d$-dimensional output satisfies $\sigma_i(y)\le I_d$.
With $\lambda_p$ from Lemma~\ref{lem:fourier} and
\[
 c\coloneqq \kappa(K),
 \qquad w_i(0)\coloneqq 1,\quad w_i(a)\coloneqq \frac{h_i(a)}{\lambda_p}\quad(a\ne0),
\]
the same lemmas give
\begin{equation}\label{sup:outside-weights}
 0\le w_i\le1,\qquad w_i(a)\le c\quad(a\ne0),
 \qquad\sum_a w_i(a)^2\le d^2.
\end{equation}
Only the outside coefficients have been normalized.

Let $\mathcal C=\mathcal C_0^\perp$ be the dual of the zero-secret sharing space.
By Lemma~\ref{lem:vandermonde-kernel}, it is a generalized Reed--Solomon code of dimension $k\coloneqq n-t+1$ and minimum nonzero weight $t$, whose elements are $(v_iP(x_i))_i$ with
\[
 v_i\coloneqq \left(x_i\prod_{j\ne i}(x_i-x_j)\right)^{-1},
 \qquad \deg P<k.
\]
Fourier inversion across the independent blocks, as in Equation~\eqref{eq:fourier-security}, gives
\begin{equation}\label{sup:fourier-sum}
 \Adv(\rho_s,\rho_{s'})
 \le\lambda_p^N
 \sum_{a\in\mathcal C\setminus\{0\}}
       H(a_S)\prod_{i\in O}w_i(a_i).
\end{equation}
Indeed, the zero frequency cancels and the phase difference has modulus at most two, canceling the factor $1/2$ in trace distance.

Put
\[
 \zeta\coloneqq \frac{\min\{\eta,\xi\}}4,
 \qquad\gamma\coloneqq \frac\xi4,
 \qquad q\coloneqq \lfloor\zeta n\rfloor.
\]
Increase $n_0$, depending only on the fixed parameters, so that
\begin{equation}\label{sup:integer-budgets}
 q\ge\frac{\zeta n}2,\qquad
 q\le\min\{n-2k,N-k\},\qquad
 N-q\ge k+\lceil\gamma N\rceil.
\end{equation}
These choices are uniform: the exact dimension margins are
\[
 n-2k=2t-n-2\ge2\eta n-2,\qquad
 N-k=t-r-1\ge\xi n-1,\qquad N\ge\xi n,
\]
and
\[
 N-q-k-\lceil\gamma N\rceil\ge\frac{\xi n}2-2.
\]
Choose a fixed $0<\theta<c$ so small that
\begin{equation}\label{sup:cutoff}
 (2d)\theta^{\zeta/2}\le c.
\end{equation}
At each outside coordinate use the list
\[
 L_i\coloneqq \{0\}\cup\{a\ne0:w_i(a)>\theta\}.
\]
Its size is at most the fixed integer $\lceil1+d^2/\theta^2\rceil$.
Call an outside coordinate light for a codeword if its symbol is not in this list.

For words with fewer than $q$ light coordinates, there are at least $N-q+1$ agreements with the outside lists.
Since $N\ge k$, puncturing $\mathcal C$ to $O$ is injective and gives a length-$N$, dimension-$k$ generalized Reed--Solomon code.
We use the first part of Lemma~\ref{gv:quantitative-list} at length $N$ with slack $\gamma$: the number of words agreeing on at least $k+\lceil\gamma N\rceil$ coordinates is polynomial in $N$.
Both the dimension and agreement conditions follow from Equation~\eqref{sup:integer-budgets}; $N\ge\xi n$ supplies the uniform length cutoff.
Hence at most $A_0n^{b_0}$ original words occur in this part of the sum.
Every nonzero word of $\mathcal C$ has at least $t-r$ nonzero outside coordinates, so each contributes at most $c^{t-r}$.
The contribution of these words is at most $A_0n^{b_0}c^{t-r}$.

For the other words, mark a set $M\subseteq O$ of $q$ light coordinates and sum over all such sets.
There are at most $2^N$ choices.
For each fixed $M$, extract the factor $\theta^q$ and choose two disjoint collections of whole blocks with injective coordinate projections on $\mathcal C$.
If $r\ge k$, use the support block as the first collection and any $k$ unmarked outside coordinates as the second; the latter exist because $q\le N-k$.
If $r<k$, use the support block together with $k-r$ unmarked outside coordinates for the first collection and another $k$ outside coordinates for the second.
This is possible because $q\le n-2k$.
When $r=0$, this simply gives two disjoint information sets.

Each collection contains at least $k$ scalar coordinates, so the MDS property makes its projection injective.
Let $F_1(a)$ and $F_2(a)$ be the products of the weights in the two collections, including $H(a_S)$ in the collection containing the support.
Discarding the remaining factors, which are at most one, and applying Cauchy--Schwarz gives
\[
 \sum_{\substack{a\in\mathcal C\\M\text{ light for }a}}
       H(a_S)\prod_{i\in O}w_i(a_i)
 \le\theta^q
       \left(\sum_{a\in\mathcal C}F_1(a)^2\right)^{1/2}
       \left(\sum_{a\in\mathcal C}F_2(a)^2\right)^{1/2}
 \le\theta^q d^n.
\]
For the last step, injectivity bounds each squared sum by the unrestricted sum over its block coordinates.
Those sums factor into the squared masses in Equations~\eqref{sup:block-mass} and~\eqref{sup:outside-weights}.
The collections are disjoint, so their combined square-root cost is at most
$(d^{2r}(d^2)^N)^{1/2}=d^n$.
The full contribution of this part is therefore at most
\[
 2^N\theta^qd^n
 \le\bigl[(2d)\theta^{\zeta/2}\bigr]^n
 \le c^n\le c^{t-r}.
\]
Combine the two parts, absorb the additional constant into $A$, and use
$\lambda_p^N\le\exp(2\pi^2N/(3p^2))$ in Equation~\eqref{sup:fourier-sum}.
\end{proof}

\begin{remark}[Joint support access and classical mixtures]\label{sup:extensions}
Theorem~\ref{sup:upper} also permits a single share-dependent CPTP map on all of $A_S$, with access to $Y_S$, provided its total output dimension is at most $d^r$.
Communication within $S$ is thus allowed in this enlarged model; every emitted register counts toward that dimension.
Outside devices still act locally, and communication involving them or adversary feedback remains excluded.
The proof uses only the block output dimension and its fixed reference marginal.

The bound also holds for finite share-independent mixtures whose \emph{whole conditional resources} satisfy the stated premise, even with different support sets of size at most $r$.
Enlarge each branch's support to size $r$ by including independent outside factors, grant the mixture label to the adversary, and average the uniform bound.
\end{remark}

\subsection{The support-dependent lower construction}\label{supp:support-lower-proof}
The following is the complete proof of Proposition~\ref{sup:lower}.
\begin{proof}
Choose a reconstruction set $T\supseteq S$ of size $t$ and let $Q\coloneqq T\setminus S$.
Its nonzero reconstruction coefficients satisfy $\sum_{i\in T}\lambda_iY_i=s$.
Choose $u\in\F_p^\times$ with $u(s'-s)=(p-1)/2$.
The support shares the state
\[
 \frac{|0\rangle^{\otimes r}+|1\rangle^{\otimes r}}{\sqrt2}.
\]
Device $i\in S$ applies the phase $\omega^{u\lambda_iY_i}$ to its $|1\rangle$ component and emits the resulting qubit.
Put $Z_j\coloneqq u\lambda_jY_j$ for $j\in Q$.
The support phase is exactly $us-\sum_{j\in Q}Z_j$.
Since $|Q|=t-r<t$, the variables $Z_j$ are mutually independent and uniform for each fixed secret.

For each $j\in Q$, use the arc labels of Lemma~\ref{lem:arc-labels}.
For $K=d^2$, the $K$ labels are the orthogonal Bell states obtained locally from an independent maximally entangled $d$-dimensional pair, as in Yan and Wang \cite[arXiv version, Section~II]{YW}.
For $K=d$, use $d$ orthogonal output states without a reference.
The lemma gives a physically allowed state map $\sigma_j$ with
\[
 H_j\coloneqq \E_z\omega^{-z}\sigma_j(z),\qquad\norm{H_j}\ge\kappa(K),
\]
for every permitted prime.
Device $j$ applies this map to $Z_j$.
Devices outside $T$ emit fixed states, which can be omitted from the distance calculation.

Write
\[
 G\coloneqq \bigotimes_{j\in Q}H_j,
 \qquad D\coloneqq \bigotimes_{j\in Q}\E_z\sigma_j(z).
\]
In the basis $|0\rangle^{\otimes r},|1\rangle^{\otimes r}$ of the support subspace, uniform averaging over $Z_Q$ gives
\[
 \rho_s=\frac12
 \begin{pmatrix}
 D&\omega^{-us}G^\dagger\\
 \omega^{us}G&D
 \end{pmatrix}.
\]
For $Q=\varnothing$, take $G=D=1$.
The diagonal blocks are independent of the secret, and
\[
 \norm{\begin{pmatrix}0&X^\dagger\\X&0\end{pmatrix}}=2\norm{X}.
\]
Consequently,
\[
 \Adv(\rho_s,\rho_{s'})
 =\frac{|\omega^{us}-\omega^{us'}|}{2}\norm{G}
 =\cos\!\left(\frac\pi{2p}\right)
       \prod_{j\in Q}\norm{H_j}
 \ge\cos\!\left(\frac\pi{2p}\right)\kappa(K)^{t-r}.
\]
The initial GHZ state and outside resources satisfy the corresponding resource premise, and every channel used is local.
\end{proof}

\subsection{A revelation consequence near one half}\label{supp:support-near-half}
\begin{corollary}[A revelation bound near rate one half]\label{sup:below-half}
Let $c_{\rm q}=1/2-10^{-9}$, and let $a,n_{\rm q}>0$ be the constants in Theorem~\ref{bh:main}.
Suppose the full resource satisfies Equation~\eqref{sup:no-outside-reference}, every device emits at most one qubit, and
\begin{equation}\label{sup:below-half-condition}
 N\coloneqq n-r\ge n_{\rm q},\qquad
 t-r-1\ge c_{\rm q}(n-r).
\end{equation}
Then every chosen-secret distance is at most $2e^{-a(n-r)}$.
The assertion also holds for share-independent classical mixtures of such whole resources, with a common residual length cutoff.
For fixed $R>c_{\rm q}$, it gives exponential security whenever $t/n\ge R$ and $r/n\le\alpha$ with
\[
 \alpha<\frac{R-c_{\rm q}}{1-c_{\rm q}}.
\]
In particular, at $R=1/2$ it permits every fixed
$\alpha<2\cdot10^{-9}/(1+2\cdot10^{-9})$ for sufficiently large $n$.
\end{corollary}
\begin{proof}
Reveal $Y_S$ and apply Lemma~\ref{lem:conditioning}.
The outside shares become affine images of an $(N,t-r)$ Shamir instance; its zero-secret code has dimension $t-r-1$.
Theorem~\ref{bh:main} bounds each secret output by $e^{-aN}$ from the product of the local average states.
That product is the same for both secrets, since uniform averaging is unchanged by the nonzero affine coordinate transformations.
The triangle inequality therefore gives $2e^{-aN}$.
The support and reference output can be simulated from $Y_S$ by common CPTP postprocessing, and the revealed-share distribution is secret independent.
Averaging the conditional bounds proves the claim, also for labeled whole-resource mixtures.
The finite condition is exactly
\[
 r\le\left\lfloor\frac{t-1-c_{\rm q}n}{1-c_{\rm q}}\right\rfloor,
\]
besides the residual length cutoff.
The fixed strict margin in the displayed fraction absorbs the $-1$ and keeps $N$ linear in $n$.
\end{proof}
This last corollary composes the revelation reduction of \cite[Theorem~17]{BKNY} with Theorem~\ref{bh:main}.
It gives a support guarantee at and slightly below rate one half under Equation~\eqref{sup:no-outside-reference}.
For fixed $R>1/2$, its support region is contained in that of Theorem~\ref{sup:upper}.

\section{Weighted block estimates and marginal transfer}\label{supp:blocks}

\subsection{The scalar estimate for bounded blocks}\label{supp:blocks-weighted}
Equation~\eqref{mb:slice-bound} gives a squared-mass bound that grows with the block length and a nonzero-frequency cap independent of that length.
To use these two estimates, we extend the weighted Fourier thresholding argument of Sun and Wootters \cite[Lemma~1.15 and Section~4.2]{SW} to weights depending jointly on bounded groups of coordinates.
The large weights in a group give bounded scalar lists after coordinate projection; the remaining weights are controlled by two disjoint collections of whole groups.

\begin{lemma}[Weighted sums over bounded blocks]\label{mb:weighted}
Fix an integer $\ell\ge1$, constants $c_*\in(0,1)$ and $K\ge1$, and $\epsilon\in(0,1/2)$.
There are constants $A,b,n_0>0$, depending only on these parameters, such that the following holds.
Let $n\ge n_0$, let $p\ge n$ be prime, and let $\mathcal G$ be a length-$n$ generalized Reed--Solomon code of dimension
\[
 1\le k\le(1/2-\epsilon)n,
 \qquad \Delta\coloneqq n-k+1.
\]
Partition its coordinates into nonempty sets $I_1,\ldots,I_M$ of size at most $\ell$.
Suppose $w_j:\F_p^{I_j}\to[0,1]$ satisfy
\[
 w_j(0)=1,\qquad
 w_j(v)\le c_*\quad(v\ne0),\qquad
 \sum_v w_j(v)^2\le K.
\]
Then
\begin{equation}\label{mb:weighted-bound}
 \sum_{a\in\mathcal G\setminus\{0\}}
       \prod_{j=1}^M w_j(a_{I_j})
 \le A n^b c_*^{\lceil\Delta/\ell\rceil}.
\end{equation}
The constants are uniform over the prime field, evaluation points, nonzero coordinate multipliers, partition, and weights.
\end{lemma}
\begin{proof}
Put $r\coloneqq \lfloor\epsilon n/\ell\rfloor$ and take $n$ sufficiently large that
\begin{equation}\label{mb:information-budget}
 r\ge\frac{\epsilon n}{2\ell},\qquad
 n-\ell r\ge(1-\epsilon)n\ge2k+\ell-1.
\end{equation}
For example, the displayed inequalities hold for $n\ge\max\{4,2\ell/\epsilon\}$.
Choose a fixed $\theta\in(0,c_*)$ so small that
\[
 2\sqrt K\,\theta^{\epsilon/(2\ell)}\le c_*.
\]
For each block define
\[
 H_j\coloneqq \{0\}\cup\{v\ne0:w_j(v)>\theta\}.
\]
Its size is at most $1+K/\theta^2$.
Projecting $H_j$ onto each coordinate of $I_j$ therefore gives scalar lists of size at most the fixed integer $L\coloneqq \lceil1+K/\theta^2\rceil$.
Call a block light for a codeword $a$ if $a_{I_j}\notin H_j$.

If fewer than $r$ blocks are light, the word agrees with these scalar lists on at least
\[
 n-\ell(r-1)\ge(1-\epsilon)n
 \ge k+\lceil n/4\rceil
\]
coordinates, after increasing the length cutoff if necessary.
The first part of Lemma~\ref{gv:quantitative-list}, with agreement slack $1/4$, bounds the number of such words by $A_0n^{b_0}$ for constants depending only on $L$.
The dimension condition $k\le3n/4$ holds here.
Every nonzero codeword has at least $\Delta$ nonzero coordinates and thus meets at least $\lceil\Delta/\ell\rceil$ nonzero blocks.
The words in this first class contribute at most
$A_0n^{b_0}c_*^{\lceil\Delta/\ell\rceil}$.

For the other words, mark any $r$ light blocks and sum over all possible marked sets $\mathcal M$; there are at most $2^M$ such sets.
Fix $\mathcal M$.
By Equation~\eqref{mb:information-budget}, the unmarked blocks cover at least $2k+\ell-1$ coordinates.
Greedily choose a collection $\mathcal U$ of whole unmarked blocks covering between $k$ and $k+\ell-1$ coordinates.
There remain at least $k$ coordinates, so choose a disjoint collection $\mathcal V$ of whole unmarked blocks covering at least $k$ coordinates.
Both coordinate projections are injective on $\mathcal G$, because any $k$ coordinates form an information set.
Extract $\theta^r$ from the marked blocks and discard all other factors outside $\mathcal U\cup\mathcal V$, which are at most one.
Cauchy--Schwarz gives
\begin{align*}
 \sum_{\substack{a\in\mathcal G\\\mathcal M\text{ light for }a}}
       \prod_jw_j(a_{I_j})
 &\le\theta^r
   \left(\sum_{a\in\mathcal G}\prod_{j\in\mathcal U}w_j(a_{I_j})^2\right)^{1/2}
   \left(\sum_{a\in\mathcal G}\prod_{j\in\mathcal V}w_j(a_{I_j})^2\right)^{1/2}\\
 &\le\theta^r K^{(|\mathcal U|+|\mathcal V|)/2}
 \le\theta^r K^{M/2}.
\end{align*}
For the second inequality, injectivity bounds each sum by the unrestricted sum over all tuples on its blocks; that sum factors into their squared masses.
Since $M\le n$, Equation~\eqref{mb:information-budget} also gives $r\ge\epsilon M/(2\ell)$.
The total contribution of the second class is therefore at most
\[
 2^M\theta^rK^{M/2}
 \le\bigl(2\sqrt K\,\theta^{\epsilon/(2\ell)}\bigr)^M
 \le c_*^M
 \le c_*^{\lceil\Delta/\ell\rceil}.
\]
The last inequality uses $M\ge\lceil n/\ell\rceil\ge\lceil\Delta/\ell\rceil$.
Combining the two classes and enlarging $A,b,n_0$ proves the claim.
All list and geometric cutoffs depend only on the fixed parameters in the statement.
\end{proof}

\subsection{The block marginal criterion}\label{supp:blocks-transfer}
To apply Proposition~\ref{mb:independent}, we compare a retained marginal with a product over blocks and reveal only the omitted shares.
Corollary~\ref{thm:domination-blocks} applies the reference filter of Theorem~\ref{thm:domination} to this product exactly as to a product over individual devices.
Proposition~\ref{mb:criterion} records both the residual threshold and the cost of the block correlations.

\begin{proof}[Proof of Proposition~\ref{mb:criterion}]
Write $m\coloneqq \ell M$, $u\coloneqq n-t$, and $\Lambda\coloneqq C_0\beta^{M-1}$.

Reveal the omitted shares.
The first strict inequality in Equation~\eqref{mb:strict-criterion} gives
\[
 \limsup_{n\to\infty}\frac um
 \le\frac{\delta(\ell+g)}\ell<\frac12.
\]
Thus $n-m<t$ eventually, and Lemma~\ref{lem:conditioning} gives a residual Shamir instance on $m$ devices with threshold $m-u$ above $m/2$ by a fixed margin.
Its coordinatewise affine substitutions preserve the separate local channels.
Corollary~\ref{thm:domination-blocks}, with Proposition~\ref{mb:independent} supplying the independent-block bound at length $m$ and threshold $m-u$, consequently gives
\begin{equation}\label{mb:transferred-bound}
 \Adv(\rho_s,\rho_{s'})
 \le C_0\beta^{M-1}A m^b
       \kappa(d^2)^{\lceil(m-u)/\ell\rceil}
       \exp\!\left(\frac{2\pi^2M}{3p^2}\right).
\end{equation}
Here the residual dual dimension is $u+1$, and the field remains the original $\F_p$.

Since $M=n/(\ell+g)+O(1)$ and
$\lceil(m-u)/\ell\rceil\ge M-u/\ell$, Equation~\eqref{mb:transferred-bound} implies
\[
 \log\Adv(\rho_s,\rho_{s'})
 \le-\left[\frac{-\log\bigl(\beta\kappa(d^2)\bigr)}{\ell+g}
       -\frac{\delta}{\ell}\log\frac1{\kappa(d^2)}\right]n
       +O(\log n).
\]
The second strict inequality in Equation~\eqref{mb:strict-criterion} makes the coefficient in brackets positive.
The finite-field term is uniform because $p>n$.
Every exponent in Equation~\eqref{mb:exponent} follows by absorbing the logarithmic term into the strict exponent margin.
\end{proof}

\section{Explicit sparse estimates and memory examples}\label{supp:memory}
\subsection{Explicit estimates from sparse marginals}\label{app:sparse-memory}
The block argument in Section~\ref{sec:mixing} proves security at every fixed rate above one half.
Here we give the singleton comparisons and their explicit finite estimates.
They give the prefactors and the singleton sufficient rates used in the memory example, whose block rates come from Proposition~\ref{mb:criterion}, and they apply the reset-channel and boundary-normalization comparisons of Lemma~\ref{mb:comparison} to these sparse selections.
Lemma~\ref{lem:mix-finite} supplies the finite estimates for Theorems~\ref{thm:mix} and~\ref{thm:memory} and Corollary~\ref{cor:mps}; Proposition~\ref{prop:sparse-rate} uses Theorem~\ref{thm:half}.
All resource definitions are those of Section~\ref{sec:mixing}.

\paragraph{Sparse sites in a mixing chain.}
The mixing hypothesis applies to intervals, whereas the reduction uses a sparse set of devices.
Tracing out the intervening sites gives the needed comparison.

\begin{lemma}\label{lem:mix-product}
Assume the upper operator mixing condition, and fix $\beta>1$ and a corresponding gap $g$.
Let $\alpha$ be the stationary one-site marginal, and choose sites
$I\coloneqq \{1,1+h,1+2h,\ldots\}\cap[n]$, where $h\coloneqq g+1$.
Then $m\coloneqq |I|=\lceil n/h\rceil$ and
\[
 \Omega_{A_I}\le\beta^{m-1}\alpha^{\otimes m}.
\]
\end{lemma}
\begin{proof}
This is the first item of Lemma~\ref{mb:comparison} with $\ell=1$ and $M=m$: apply the mixing inequality to the interval containing the previously selected sites and the next selected site, trace out the omitted sites in the first interval, and iterate $m-1$ times.
\end{proof}

An exponentially growing domination factor requires a quantitative version of the estimate for independent device--reference pairs.
The next bound retains a constant finite-field cost even when the selected-set size grows with $n$.

\begin{lemma}[Uniform finite-field cost]\label{lem:mix-finite}
Write $u\coloneqq n-t$ and $H\coloneqq 1+3\kappa(4)^{-2}$.
Let $I\subseteq[n]$ have size $m\ge2(u+1)$, let $C\ge1$ and $\beta>1$, and let $\alpha_1,\ldots,\alpha_m$ be density matrices on its device registers.
If $\Omega_{A_I}\le C\beta^{m-1}\alpha_1\otimes\cdots\otimes\alpha_m$, then every permitted one-qubit leakage strategy obeys
\[
 \Adv\le C e^{\pi^2/18}
 \beta^{m-1}\kappa(4)^{m-u}\bigl[H^{u+1}-1\bigr].
\]
This bound is uniform for every prime $p>n$.
\end{lemma}
\begin{proof}
The residual threshold is $m-u>0$ and $q\coloneqq m/(u+1)\ge2$.
Use the second, simpler bound in Equation~\eqref{eq:finite} on the residual sharing with $\lambda=1/\kappa(4)$ and apply marginal domination.
Expand its bracket into $u+1$ terms.
In term $j$, replacing $c_p$ by $\kappa(4)$ costs a relative factor at most
\[
 \left(1+\frac{2\pi^2}{3p^2}\right)^{j(q-2)}
 \le\exp\!\left(\frac{2\pi^2(m-2)}{3p^2}\right)
 \le e^{\pi^2/18}.
\]
Indeed, $0\le j(q-2)\le m-2(u+1)\le m-2$ and $p\ge m+1$.
The last inequality follows from $(m-2)/(m+1)^2\le1/12$, as in Theorem~\ref{thm:sharp}.
The remaining terms sum to $H^{u+1}-1$.
\end{proof}

Choose the spacing so that the domination factor per selected site is smaller than $1/\kappa(4)$.
The following theorem shows that a positive leakage exponent then survives a fixed threshold deficit.

\begin{theorem}[Security under upper operator mixing]\label{thm:mix}
Suppose the $n$-device marginal is the restriction of one fixed stationary state satisfying Definition~\ref{def:upper-mixing}.
Any external extension of that marginal is permitted.
There are constants $\delta,c>0$ and $n_0$ such that, for $n\ge n_0$ and
$t/n\ge1-\delta$, every permitted one-qubit leakage strategy has distance at most $2^{-cn}$, uniformly in all primes $p>n$ and evaluation points.
\end{theorem}
\begin{proof}
Choose $1<\beta<1/\kappa(4)$, its gap $g$, and $h\coloneqq g+1$.
Set
\[
 A\coloneqq -\log(\beta\kappa(4))>0,\qquad B\coloneqq \log(H/\kappa(4))>0,
 \qquad \delta\coloneqq \frac{A}{2hB}.
\]
Here $A<B$, so for sufficiently large $n$, $u\le\delta n$ implies
$m=\lceil n/h\rceil\ge2(u+1)$.
Lemmas~\ref{lem:mix-product} and \ref{lem:mix-finite} yield
\[
 \Adv\le\frac{e^{\pi^2/18}H}{\beta}
             e^{-Am+Bu}
 \le\frac{e^{\pi^2/18}H}{\beta}e^{-An/(2h)}.
\]
Absorb the fixed prefactor by increasing $n_0$.
\end{proof}

The stronger independent-pair estimate gives a second quantitative criterion for the same sparse marginal.
It separates the loss from revealing shares from the cost of the remaining correlations.

\begin{proposition}[A rate from sparse marginal domination]\label{prop:sparse-rate}
Fix an integer $h\ge1$, constants $C\ge1$ and $1<\beta<1/\kappa(4)$, and put
\[
 L_\kappa\coloneqq \log(1/\kappa(4)),\qquad
 \delta_*(h,\beta)\coloneqq \frac1h
   \min\left\{\frac12,1-\frac{\log\beta}{L_\kappa}\right\}.
\]
Suppose a selected marginal on $m\coloneqq \lceil n/h\rceil$ devices obeys
$\Omega_{A_I}\le C\beta^{m-1}\bigotimes_{i\in I}\alpha_i$.
For any fixed $0\le\delta<\delta_*(h,\beta)$, every permitted one-qubit strategy has exponentially small leakage whenever $t/n\ge1-\delta$.
More precisely, for every
\[
 0<c<\frac{-\log(\beta\kappa(4))/h-\delta L_\kappa}{\log2},
\]
the bound $\Adv\le2^{-cn}$ holds for all sufficiently large $n$, uniformly over primes $p>n$, evaluation points, and finite adversary references.
\end{proposition}
\begin{proof}
Put $u\coloneqq n-t$ and $t_I\coloneqq m-u$.
Since $u/m\le h\delta<1/2$, the residual threshold ratio $t_I/m$ exceeds $1/2$ by a fixed margin.
Apply Theorem~\ref{thm:half} with $d=2$ to the residual instance and transfer it by Theorem~\ref{thm:domination}.
The finite-field exponential factor is uniformly bounded, and all polynomial factors contribute $O(\log n)$ to the logarithm.
Thus
\[
 \log\Adv
 \le (m-1)\log\beta+(m-u-1)\log\kappa(4)+O(\log n)
 \le-\left[\frac{-\log(\beta\kappa(4))}h-\delta L_\kappa\right]n
       +O(\log n).
\]
The coefficient in square brackets is positive by the definition of $\delta_*$.
\end{proof}
For a mixing chain, Lemma~\ref{lem:mix-product} supplies $C=1$ and $h=g+1$.
Thus every fixed choice of gap and domination factor gives an explicit sufficient threshold through Proposition~\ref{prop:sparse-rate}.

\paragraph{Primitive memories and matrix product states.}
\begin{theorem}[Primitive-memory resources]\label{thm:memory}
Fix a finite-dimensional memory $R$, a finite-dimensional device register $A$, and a CPTP emission channel $\Psi:R\to A\otimes R$.
Suppose $T\coloneqq \Tr_A\circ\Psi$ is primitive.
There are $\delta,c>0$ and $n_0$, depending only on $\Psi$, such that for every extension of the joint marginal of the $n$ emitted registers, one-qubit local leakage has distance at most $2^{-cn}$ whenever $n\ge n_0$ and $t/n\ge1-\delta$.
The constants are uniform over the initial memory state, all primes $p>n$, evaluation points, and permitted leakage channels.
\end{theorem}
\begin{proof}
Let $\sigma$ be the full-rank stationary state of $T$.
For each $\epsilon>0$, the second item of Lemma~\ref{mb:comparison}, with $\beta=1+\epsilon$, gives a fixed $g$ with
\[
 T^g\le_{\rm CP}(1+\epsilon)\mathcal R_\sigma,
 \qquad \mathcal R_\sigma(X)\coloneqq \Tr(X)\sigma;
\]
its proof chooses $g$ so that the Choi operator-norm error is at most $\epsilon\lambda_{\min}(\sigma)$.

Start the sequential emission process from any state $\tau$ and put
$\alpha\coloneqq \Tr_R\Psi(\sigma)$.
Take $I\coloneqq \{1,1+h,1+2h,\ldots\}\cap[n]$, where $h\coloneqq g+1$ and $m\coloneqq \lceil n/h\rceil$.
Between successive selected sites there are exactly $g$ omitted emissions.
Put $\Lambda_0\coloneqq \|\sigma^{-1/2}\tau\sigma^{-1/2}\|_\infty\le1/\lambda_{\min}(\sigma)$, where $\|\cdot\|_\infty$ is the operator norm.
The third item of Lemma~\ref{mb:comparison}, with $\ell=1$, $M=m$ and $\beta=1+\epsilon$, replaces the $m-1$ gap channels by the reset bound and gives
\[
 \Omega_{A_I}\le\Lambda_0(1+\epsilon)^{m-1}\alpha^{\otimes m}.
\]
Tracing the final memory and any emissions after the last selected site preserves the comparison constant.
Choosing $1+\epsilon<1/\kappa(4)$ and applying Lemma~\ref{lem:mix-finite} proves the same near-full-threshold conclusion as Theorem~\ref{thm:mix}, uniformly in $\tau$.
The full initial resource may be any extension, including all preparation environments and the final memory.
The device register $A$ can have any fixed finite dimension; its emitted leakage is still one qubit.
\end{proof}
More quantitatively, any verified comparison
$T^g\le_{\rm CP}\beta\mathcal R_\sigma$ with $1<\beta<1/\kappa(4)$ gives the sufficient deficits
$\delta<\delta_*(g+1,\beta)$ of Proposition~\ref{prop:sparse-rate}.
The factor $\Lambda_0$ is uniformly bounded by $1/\lambda_{\min}(\sigma)$ and affects only the length threshold.
For a fixed finite-dimensional channel, the CP comparison can be checked from its Choi matrix.

The same estimate covers pure matrix product states when the final boundary conditioning has probability bounded below.
Primitivity supplies that lower bound uniformly in the two boundary vectors.

\begin{corollary}[Pure open-boundary matrix product states]\label{cor:mps}
Fix positive integers $D,v$ and $D\times D$ matrices $A_1,\ldots,A_v$ satisfying
$\sum_jA_j^\dagger A_j=I_D$, with primitive transfer channel
$T(X)\coloneqq \sum_jA_jXA_j^\dagger$.
For sufficiently large $n$, normalize the state whose amplitudes are
$\langle u_R|A_{j_n}\cdots A_{j_1}|u_L\rangle$, where $u_L,u_R\in\mathbb C^D$ are unit vectors and each $j_i$ ranges over $\{1,\ldots,v\}$.
There are $\delta,c>0$ and $n_0$, depending only on the matrices, for which every allowed one-qubit leakage strategy has distance at most $2^{-cn}$ when $n\ge n_0$ and $t/n\ge1-\delta$.
This holds uniformly over both boundary vectors, primes $p>n$, and evaluation points.
\end{corollary}
\begin{proof}
Let $\sigma$ be the full-rank stationary state of $T$.
Choose $1+\epsilon<1/\kappa(4)$, a corresponding $g$, and the same sparse set $I$ as in the proof of Theorem~\ref{thm:memory}.
The fourth item of Lemma~\ref{mb:comparison}, with $\ell=1$, $M=m$ and $\beta=1+\epsilon$, gives
\[
 \Omega_{A_I}\le\frac{2}{\lambda_{\min}(\sigma)^2}(1+\epsilon)^{m-1}\alpha^{\otimes m}
\]
for all sufficiently large $n$, uniformly over both boundary vectors, where $\alpha$ is the one-site emission state of $V\coloneqq \sum_j|j\rangle\otimes A_j$ started from $\sigma$.
Its proof postselects the final memory onto $|u_R\rangle$ with success probability $q_n\ge\tfrac12\lambda_{\min}(\sigma)$, so normalizing costs at most $2/\lambda_{\min}(\sigma)$ beyond the initial factor.
Lemma~\ref{lem:mix-finite} then gives the asserted leakage bound after increasing $n_0$ to absorb the fixed prefactor.
\end{proof}
The primitive transfer condition is the usual injectivity condition after blocking sufficiently many sites \cite[arXiv version, Section~V, Eq.~(10)]{Sanz}.
The boundary normalization in Lemma~\ref{mb:comparison} is part of the open-boundary formulation.
\subsection{An explicit mixing resource}\label{supp:memory-examples}
A two-dimensional memory already produces a resource whose two-point correlations remain nonzero at every distance.
It also gives explicit constants in the security estimate.
For $0<\mu<1$, let
\[
 A_0\coloneqq \sqrt{1-3\mu/4}\,I,\qquad
 A_j\coloneqq \sqrt{\mu/4}\,\sigma_j\quad(j=1,2,3).
\]
Here $\sigma_1,\sigma_2,\sigma_3$ are the Pauli matrices.
Use the emission isometry $V\coloneqq \sum_{j=0}^3|j\rangle\otimes A_j$ and stationary initial memory $I/2$.
The memory is two-dimensional and the emitted device register is four-dimensional.
The memory channel is $T=(1-\mu)\id+\mu\mathcal R_{I/2}$, and the exact CP comparison is
\[
 T^g\le_{\rm CP}\bigl[1+3(1-\mu)^g\bigr]\mathcal R_{I/2}.
\]
The Choi eigenvalue ratios to reset are $1+3(1-\mu)^g$ and $1-(1-\mu)^g$.
For stationary initial memory and
$O\coloneqq |0\rangle\langle3|+|3\rangle\langle0|$,
\[
 \langle O_i\rangle=0,\qquad
 \langle O_iO_j\rangle
 =4(1-3\mu/4)(\mu/4)(1-\mu)^{j-i-1}>0 \quad(i<j).
\]
To check this, the insertion map is
$\mathcal F_O(X)\coloneqq \sqrt{(1-3\mu/4)(\mu/4)}(X\sigma_3+\sigma_3X)$;
$\sigma_3$ is an eigenoperator of $T$ with eigenvalue $1-\mu$.
Thus correlations persist at every separation, while leakage remains exponentially small.
The correlation calculation concerns the unconditioned stationary resource.

For example, $\mu=1/2,g=8$ gives $\beta=259/256$ and
$\beta\kappa(4)\approx0.91086689799<1$.
At full threshold the bound is a fixed constant times
$2^{-cn}$, where $c\coloneqq-\log_2(\beta\kappa(4))/9\approx0.01496531575$.
On a fixed interval below full threshold, the choices of $A,B,h$ in the proof of Theorem~\ref{thm:mix} give the corresponding constants.

For the sharper singleton sufficient rate in this example, take $\mu\coloneqq 1/2$ and $g\coloneqq 6$.
Then $h=7$ and $\beta=67/64$, with
$(67/64)\sqrt{\kappa(4)}<1$.
Consequently $\delta_*(7,67/64)=1/14$.
Proposition~\ref{prop:sparse-rate} proves exponential security whenever
\[
 \frac tn\ge\frac{13}{14}+\eta
\]
for a fixed $\eta>0$ and sufficiently large $n$.
This bound is obtained from the displayed reset comparison for this particular memory channel.

The same reset comparison also illustrates the gain from retaining blocks.
For $\mu=1/2$, keep $g=6$ and $\beta=67/64$; the inequality $\beta\sqrt{\kappa(4)}<1$ makes the minimum in Equation~\eqref{mb:strict-criterion} equal to $1/2$ for qubit outputs.
Taking $\ell\coloneqq 24$ gives exponential security at every rate $t/n\ge3/5+\eta$ with a fixed $\eta>0$; taking $\ell\coloneqq 54$ similarly gives every rate $t/n\ge11/20+\eta$.
These are strict sufficient rate bounds for fixed block lengths.
The length cutoff depends on the block length and the list-recovery constants, while Theorem~\ref{mb:memory} reaches every fixed rate above one half by choosing a sufficiently long fixed block.

\paragraph{Nonprimitive memory.}
The usual sequential GHZ generator sends
$|b\rangle_R\mapsto|b\rangle_A|b\rangle_R$.
Starting from $|+\rangle_R$, the emitted registers together with the final memory form a GHZ state.
Retain that memory in the adversary reference; the share-dependent phase operations and a joint parity measurement give the known constant-advantage GHZ attack \cite{BKNY}.
The memory channel is dephasing and hence nonprimitive.

The next family tests the need for a fixed memory process or uniform control of its mixing as $n$ grows.

\begin{proposition}[A varying primitive memory]\label{prop:varying-memory}
For every $n\ge2$, there is an emission channel $\Psi_n:R\to A\otimes R$ with $\dim R=\dim A=2$ and primitive transfer channel whose $n$ emitted registers admit full-threshold one-qubit leakage of distance at least
\[
 (1-1/n)^n\cos\frac{\pi}{2p}
\]
for any prescribed pair of distinct secrets over every prime $p>n$.
The adversary reference includes preparation environments and the final memory.
Along every such sequence of primes the lower bound tends to $e^{-1}$.
\end{proposition}
\begin{proof}
Let $V|b\rangle\coloneqq |b\rangle_A|b\rangle_R$ for $b\in\{0,1\}$ and define, for $0<\epsilon<1$,
\[
 \Psi_\epsilon(X)\coloneqq (1-\epsilon)VXV^\dagger
   +\epsilon\Tr(X)|0\rangle\langle0|_A\otimes I_R/2.
\]
Retain a classical copy/reset branch flag from each emission in the adversary reference.
Tracing those flags gives exactly $\Psi_\epsilon$, so this is a permitted extension of the device marginal.
The transfer channel is
\[
 T_\epsilon(X)=(1-\epsilon)\mathcal D_Z(X)+\epsilon\Tr(X)I/2,
\]
where $\mathcal D_Z$ is dephasing in the computational basis.
It maps every nonzero positive operator to a positive definite operator and is therefore primitive in one step.
Its stationary state is $I/2$, and its eigenvalue on the Pauli operator $Z$ is $1-\epsilon$.

Start the memory in $|+\rangle$.
With probability $(1-\epsilon)^n$, no reset occurs and the emitted qubits together with the final memory form an $(n+1)$-qubit GHZ state.
This event is flagged and independent of all shares.
In the GHZ phase attack of \cite{BKNY}, device $i$ applies $\operatorname{diag}(1,\omega^{\lambda_iY_i})$, where $s=\sum_i\lambda_iY_i$.
Conditioned on no reset, the adversary receives
\[
 |G_s\rangle\coloneqq \frac{|0\rangle^{\otimes(n+1)}+\omega^s|1\rangle^{\otimes(n+1)}}{\sqrt2}.
\]
The trace distance between $|G_0\rangle$ and $|G_{(p-1)/2}\rangle$ is $\cos(\pi/(2p))$.
The affine substitution from Theorem~\ref{thm:sharp} transports this guarantee to any prescribed distinct pair.
The flags give orthogonal branches with secret-independent probabilities, so their contributions to trace distance add.
Keeping the no-reset contribution gives $(1-\epsilon)^n\cos(\pi/(2p))$.
Set $\epsilon\coloneqq 1/n$.
\end{proof}
For $\epsilon=1/n$, the transfer eigenvalue $1-\epsilon$ tends to one. For fixed $\epsilon>0$, the no-reset probability $(1-\epsilon)^n$ decays exponentially.
\section{Further structured resources}\label{sec:resources}
The marginal comparison also applies to resources specified by independent sources, stabilizer constraints, or a probability distribution on pure states.
These applications use different ways of locating a nearly product marginal.
The field and strategy uniformity follow from the same reduction.

\subsection{Independent source networks}\label{sec:edges}
Independent sources give a product comparison whenever no source touches two selected devices.
Let $G=(V,E_G)$ be a simple graph on $V=[n]$. Consider independent source states indexed by its edges
$\bigotimes_{e=\{i,j\}\in E_G}\Omega_{A_{i,e}A_{j,e}R_e}$, together with independent local states $\Omega_{A_{i,0}R_i}$.
Each vertex receives the incident edge registers, and each $R_e$ and $R_i$ belongs to the adversary.
For an independent vertex set $I$, give the attacker every initial register at $S\coloneqq V\setminus I$ and every share in $S$.
Each edge touching $I$ now belongs to a separate device--reference resource.
Edges internal to $S$ form an additional secret-independent reference.
The attacker can simulate the original $S$ channels later, using the revealed shares.
These channels commute with all $I$ channels because they act on disjoint input registers.

\begin{corollary}\label{cor:edge}
In the independent-source model above, suppose $G$ has an independent vertex set $I$ with $|I|\ge b_0 n$ for fixed $b_0\in(0,1]$.
Then there are $c>0,n_0$ such that every permitted one-qubit leakage strategy has distance at most $2^{-cn}$ for $n\ge n_0$ whenever
\[
 \frac tn\ge1-\frac{b_0}{2}+\eta
\]
for fixed $\eta>0$, uniformly over all primes $p>n$, evaluation points, permitted local maps, and finite source and reference dimensions.
The same rate holds for any fixed output dimension and for classical outputs.
\end{corollary}
\begin{proof}
Put $m\coloneqq |I|$ and $r\coloneqq n-m$.
Then $t-r-m/2\ge\eta n\ge\eta m$.
Apply Theorem~\ref{thm:half} to the conditional instance from Lemma~\ref{lem:conditioning}.
The revealed classical branches all have the same weight for both secrets, so the resulting bound is their probability-weighted average.
Finally apply the simulator and data processing.
\end{proof}
Bipartite source graphs admit $b_0=1/2$.
For $n\ge2$, a star admits the larger independent set of all $n-1$ leaves and the residual rate condition $(t-1)/(n-1)\ge1/2+\eta$.

\subsection{Graph-state resources}\label{sec:graph}
For graph states, a binary cut matrix determines the spectrum of the selected marginal.
We use the reduced-state and cut-rank formulas in \cite[arXiv version, Propositions~8 and~10]{Hein}.
Let $G=(V,E_G)$ be a simple graph on $V=[n]$, with adjacency matrix $\Gamma$ over $\F_2$. Its graph state $|G\rangle$ is the unique common $+1$ eigenvector of $K_v\coloneqq X_v\prod_{w\in N(v)}Z_w$, where $X,Z$ are Pauli operators and $N(v)$ is the set of neighbors of $v$.
For $I\subseteq V$, write $m\coloneqq |I|$, $S\coloneqq V\setminus I$, and
\[
 k_I\coloneqq \rank_{\F_2}\Gamma_{S,I}.
\]

\begin{proposition}\label{prop:cutrank}
The marginal on $I$ is a normalized projector of rank $2^{k_I}$.
In particular,
\[
 \rho_I\le2^{m-k_I}\frac{I_{2^m}}{2^m}.
\]
If $k_I=m$, the marginal is the product of $m$ maximally mixed qubits.
\end{proposition}
\begin{proof}
The cited formulas give nonzero eigenvalues $2^{-k_I}$ of multiplicity $2^{k_I}$ for $\rho_I$.
Thus $\rho_I\le2^{-k_I}I_{2^m}$, which gives the displayed bound and the full-rank case.
\end{proof}

For $n-m<t$, Theorem~\ref{thm:domination} therefore gives the finite graph-state bound
\[
 \Adv(\rho_s,\rho_{s'})
 \le 2^{m-k_I}\varepsilon_{\rm ref}(m,t-n+m,p).
\]
If every selected vertex is nonisolated in $G$, its one-qubit marginal is $I/2$.
Define
$\Dmax(\rho\|\sigma)\coloneqq \log_2\inf\{a:\rho\le a\sigma\}$,
$S(\rho)\coloneqq -\Tr(\rho\log_2\rho)$, and
$D(\rho\|\sigma)\coloneqq \Tr\rho(\log_2\rho-\log_2\sigma)$,
with relative entropies infinite when support inclusion fails.
Under this nonisolation hypothesis, the minimum product domination cost, over local density matrices $\tau_i$, is exactly $m-k_I$ bits, making $2^{m-k_I}$ the optimal product domination factor:
\[
 \inf_{\tau_i}\Dmax\!\left(\rho_I\middle\|\bigotimes_i\tau_i\right)=m-k_I.
\]
The upper bound is the preceding proposition.
For the reverse inequality, use $\Dmax\ge D$ and
\[
 D\!\left(\rho_I\middle\|\bigotimes_i\tau_i\right)
 =\sum_iS(\rho_i)-S(\rho_I)+\sum_iD(\rho_i\|\tau_i)
 \ge m-k_I.
\]
Here logarithms in entropic quantities are base two.

\paragraph{A deterministic dense family.}
Take two copies of $K_m$ and join corresponding vertices by a perfect matching.
The resulting $2m$-vertex graph is connected, has independence number at most two, and its matrix across the two cliques is $I_m$.
Thus one entire clique has a maximally mixed marginal.
Corollary~\ref{cor:product} gives exponential security at every rate above
$3/4$ by a fixed margin, for every fixed output dimension and for classical outputs.
Using the elementary scalar estimates instead gives the sufficient rates
$(1+\tau_{\rm ref})/2\approx0.9598606378$ for qubit outputs and
$(1+\tau_{\rm cl})/2\approx0.8586271094$ for classical-bit outputs.
Here the edges specify stabilizer generators; the resource is the associated pure graph state.

\paragraph{A random-graph family.}
For sufficiently large $n$, in $G(n,1/2)$ fix $I$ of size
$m\coloneqq \lfloor n/2\rfloor-\lceil\log_2n\rceil$.
The cross matrix has independent uniform binary entries.
For each nonzero column vector $v$, the probability that it lies in the kernel is $2^{-(n-m)}$.
A union bound gives
\[
 \Pr[k_I<m]\le(2^m-1)2^{-(n-m)}
 <2^{2m-n}\le n^{-2}.
\]
Thus the same asymptotic sufficient rate holds with probability tending to one.
The graph-state criterion is the rank of the cross matrix. For example, two leaves of a star are nonadjacent but have a correlated graph-state marginal.

\subsection{Typical pure resources}\label{sec:haar}
Typical entanglement estimates are standard; see \cite{HLW}.
A fixed subsystem smaller than half of the qubits is close to maximally mixed with high probability. The smooth marginal comparison makes that single event sufficient for all leakage strategies.

\begin{theorem}[Haar-typical security]\label{thm:haar}
For each fixed $\eta>0$ with $3/4+\eta\le1$, there are constants $c,b>0$ and $n_0$ such that, for every $n\ge n_0$, a Haar-random $n$-qubit pure resource has, except with probability at most $2^{-bn}$, chosen-secret leakage distance at most $2^{-cn}$ whenever
\[
 \frac tn\ge\frac34+\eta.
\]
The good-resource event works simultaneously for every prime $p>n$, every set of evaluation points, and every permitted one-qubit local strategy, including strategies chosen with knowledge of the resource.
The same conclusion holds for every fixed output dimension and for classical outputs.
\end{theorem}
\begin{proof}
Choose a predetermined set $I$ of $m\coloneqq \lfloor(1/2-\xi)n\rfloor$ device qubits, where
$\xi\coloneqq \min\{1/8,\eta/2\}$.
Let $A$ denote these qubits and $B$ their complement, and write
$\rho_A\coloneqq \Tr_B|\psi\rangle\langle\psi|$, $d_A\coloneqq 2^m$, and $d_B\coloneqq 2^{n-m}$.
Let $\mathsf S$ swap two copies of $A\otimes B$.
The Haar second moment is
\[
 \E(|\psi\rangle\langle\psi|^{\otimes2})
 =\frac{I+\mathsf S}{d_Ad_B(d_Ad_B+1)}.
\]
This follows from unitary invariance on the symmetric subspace and its dimension.
The swap trick gives
\[
 \E\Tr(\rho_A^2)=\frac{d_A+d_B}{d_Ad_B+1}.
\]
Recall that $P(\rho,\sigma)^2=1-F(\rho,\sigma)^2$.
If $\lambda_j$ are the eigenvalues of $\rho_A$, then
\begin{align*}
 P(\rho_A,I/d_A)^2
 &\le2\bigl(1-F(\rho_A,I/d_A)\bigr)\\
 &=\sum_j(\sqrt{\lambda_j}-1/\sqrt{d_A})^2\\
 &\le d_A\sum_j(\lambda_j-1/d_A)^2
 =d_A\Tr(\rho_A^2)-1.
\end{align*}
Thus
\[
 \E P(\rho_A,I/d_A)^2
 \le\frac{d_A^2-1}{d_Ad_B+1}\le\frac{d_A}{d_B}
 \le2^{-2\xi n}.
\]
Markov's inequality shows that
$P(\rho_A,I/d_A)\le2^{-\xi n/2}$ except with probability at most $2^{-\xi n}$.

On that event, apply \eqref{eq:smooth} with $\Lambda=1$.
Writing $t_I\coloneqq t-n+m$ for the residual threshold, the rate hypothesis and flooring give
\[
 t_I-\frac m2\ge
 (\eta-\xi/2)n-\frac12
 \ge(\eta/2)m
\]
for sufficiently large $n$.
Theorem~\ref{thm:half} bounds the residual leakage exponentially in $m$, and
the approximation adds at most $2^{1-\xi n/2}$.
Since $m$ is linear in $n$, their sum is exponentially small in $n$.
The event depended only on $\rho_A$, so it gives the conclusion simultaneously for all listed strategies and fields.
\end{proof}
This is a statement about typical resources; the preceding graph-state examples and the preparation model below supply explicit structural alternatives.
\section{Preparation depth}\label{sec:depth}
A preparation circuit constrains which devices can share correlations.
Disjoint backward lightcones identify a product marginal, so a shallow preparation imposes a security bound.
The gates in this model act on complete device register groups, making their arity a bound on causal communication between devices.

Fix an integer arity $a\ge2$, independently of $n$, and partition all input registers into $n$ complete device register groups.
Initially their marginal is a product.
A depth-$\ell$ preparation consists of layers of disjoint CPTP gates, chosen before the sharing, each acting on at most $a$ complete groups.
Every causal influence must pass through these gates; no free global classical communication, adaptive feedforward, or postselection is included.
Independent local ancillas and local channels may be included.
The disjointness is at the group level: a layer of gates that are disjoint on individual registers can still couple many devices when a device holds many registers.
Any adversary reference extending the resulting device marginal is allowed.

Each output backward lightcone contains at most $a^\ell$ input groups.
Each input group belongs to at most $a^\ell$ output backward lightcones.
Consequently each output lightcone intersects at most $a^{2\ell}-1$ others.
A greedy choice gives at least
\[
 M\coloneqq \left\lceil\frac{n}{a^{2\ell}}\right\rceil
\]
outputs with disjoint backward lightcones.
Their joint marginal is a product: pull back arbitrary product observables to disjoint input supports and factor expectations in the product input.

Recall the Fourier cap $c_p=c_{4,p}=\min\{1,\kappa(4)+2\pi^2\kappa(4)/(3p^2)\}$.

\begin{theorem}[Full-threshold depth barrier]\label{thm:depth}
In the complete-group preparation model above, with fixed arity $a\ge2$ and depth $\ell$, suppose $t=n$ and $M=\lceil n/a^{2\ell}\rceil\ge2$.
Every permitted leakage strategy satisfies
\[
 \Adv(\rho_s,\rho_{s'})\le\min\{1,3c_p^{M-2}\}.
\]
For any fixed $\delta\in(0,1)$, a family of attacks with advantage at least $\delta$ requires
\[
 \ell\ge\frac12\log_a n-O_\delta(1).
\]
Conversely, the known GHZ leakage attack admits resource preparation depth $O(\log_a n)$.
Thus logarithmic depth is necessary and sufficient in order of growth for constant-advantage attacks at full threshold in this preparation model.
\end{theorem}
\begin{proof}
Choose $M$ outputs with disjoint backward lightcones and reveal the other shares.
The residual threshold equals $M$, and Corollary~\ref{cor:product} transfers the same bound from independent device--reference pairs.
In the second bound of Equation~\eqref{eq:finite}, the Fourier exponent is $q=M$, so that bound simplifies, for every $\lambda\ge1$, to $3c_p^{M-2}$.

Put $c_\circ\coloneqq(1+\kappa(4))/2<1$.
For $n\ge\lceil2\pi\sqrt{\kappa(4)/(3(1-\kappa(4)))}\rceil=11$ and $p>n$,
\[
 c_p\le\kappa(4)+\frac{2\pi^2\kappa(4)}{3n^2}\le\frac{1+\kappa(4)}{2}=c_\circ.
\]
Advantage at least $\delta$ then forces
\[
 M\le M_\delta\coloneqq 2+\frac{\log(3/\delta)}{\log(1/c_\circ)}.
\]
Since $M\ge n/a^{2\ell}$, the lower bound follows.
The case $M<2$ already forces a lower bound of the same order.

For the upper bound, a tree of controlled-NOT operations prepares an $n$-qubit GHZ state in depth $\lceil\log_2n\rceil$, which is $O(\log_a n)$ for fixed $a\ge2$.
Given a reconstruction set, device $i$ applies the phase determined by its share and its Lagrange coefficient, and all devices measure in the $X$ basis.
The parity has expectation $\cos(2\pi s/p)$.
Choosing separated secrets gives advantage bounded away from zero.
This uses the GHZ attack from \cite{BKNY} and establishes the matching depth order.
\end{proof}

For $u\coloneqq n-t$, the residual threshold on a product-marginal set of size $m$ is $m-u$.
Recall the fixed cap $c_\circ<1$.
If $m\ge2(u+1)$ and $c_p\le c_\circ$, use $c_\circ$ as the pointwise Fourier cap and choose $\lambda\coloneqq 1/c_\circ$ in the simpler moment estimate from the proof of Theorem~\ref{thm:finite}.
This gives
\[
 \Adv(\rho_s,\rho_{s'})
 \le c_\circ^{m-u}\left[(1+3c_\circ^{-2})^{u+1}-1\right].
\]
In particular, $m$ larger than a fixed constant times
$u+\log(1/\delta)+1$ guarantees advantage at most $\delta$.
Combining this with the lightcone packing gives a quantitative depth obstruction near full threshold.

\section{The domination-two extension}\label{ppt:proof}
We prove Theorem~\ref{ppt:domination-two} by realizing its local state
families in real Hilbert unit balls and checking the estimates of
Section~\ref{bh:section}. The scalar inputs remain
Lemmas~\ref{bh:truncated}, \ref{bh:real-gap}, \ref{bh:product-mds}
and~\ref{bh:large-support-sum}.

\begin{proof}[Proof of Theorem~\ref{ppt:domination-two}]
All states are supported on $\supp(\gamma_i)$ by
\eqref{ppt:domination}; every inverse below is restricted to that support.
Suppressing $i$, define Hermitian matrices
\[
 G\coloneqq\gamma^{1/2},\qquad
 W(y)\coloneqq\gamma^{-1/4}\sigma(y)\gamma^{-1/4},\qquad
 v(y)\coloneqq W(y)-G.
\]
The Hermitian matrices form a real Hilbert space with inner product
$\langle U,V\rangle\coloneqq\Tr(UV)$.
We have $0\le W(y)\le2G$, $\Tr(W(y)G)=1$, and $\Tr(G^2)=1$.
The trace of a product of two positive operators is nonnegative, so
\[
 \Tr(W(y)^2)\le2\Tr(W(y)G)=2,
 \qquad \|v(y)\|_2^2=\Tr(W(y)^2)-1\le1.
\]
Set
\[
 b_i\coloneqq\E_yv_i(y),\qquad
 u_i(y)\coloneqq v_i(y)-b_i,\qquad r_i\coloneqq\|b_i\|_2.
\]
The Hilbert-space variance identity gives
\begin{equation}\label{ppt:variance}
 r_i\le1,\qquad \|u_i(y)\|_2\le1+r_i,\qquad
 \E_y\|u_i(y)\|_2^2\le1-r_i^2.
\end{equation}

For every supported operator $Z$, Schatten H\"older with exponents
$4,2,4$ gives
\begin{equation}\label{ppt:holder}
 \|Z\|_1
 \le\|\gamma^{1/4}\|_4^2
       \|\gamma^{-1/4}Z\gamma^{-1/4}\|_2
 =\|\gamma^{-1/4}Z\gamma^{-1/4}\|_2.
\end{equation}
This applies to non-Hermitian Fourier coefficients as well.
For $a\ne0$, put $w_i(a)\coloneqq\|\widehat v_i(a)\|_2$.
The constant $G_i$ has no nonzero Fourier coefficients, so
\eqref{ppt:holder}, Parseval and Hermitian conjugation at opposite
frequencies give
\begin{equation}\label{ppt:fourier}
 \|\widehat\sigma_i(a)\|_1\le w_i(a),\qquad
 \sum_{a\ne0}w_i(a)^2\le1-r_i^2\le1,
 \qquad w_i(a)=w_i(-a)\le2^{-1/2}.
\end{equation}
Here $p$ is odd once the absolute length cutoff is at least two.

Define the centered support operators and their Hilbert-space counterparts by
\[
 T_A(z)\coloneqq\E_{X\in z+\mathcal G}
       \bigotimes_{i\in A}(\sigma_i(X_i)-\bar\sigma_i),\qquad
 M_A(z)\coloneqq\E_{X\in z+\mathcal G}
       \bigotimes_{i\in A}u_i(X_i).
\]
The tensor version of \eqref{ppt:holder}, with the density operator
$\gamma_A\coloneqq\bigotimes_{i\in A}\gamma_i$, yields
\begin{equation}\label{ppt:tensor-bound}
 \|T_A(z)\|_1\le\|M_A(z)\|_2.
\end{equation}
As in \eqref{bh:support-expansion}--\eqref{bh:exact-support}, with
$\sigma_i$ in place of $\rho_i$, multilinearity gives the support
expansion of $R_z-R_*$, uniform MDS marginals imply $T_A(z)=0$ for
$0<|A|\le\nu$, and normalized Fourier inversion gives the exact-support
formula for larger supports.
Consequently Lemma~\ref{bh:truncated} applies with cap $2^{-1/2}$ and
squared-mass parameter one. For fixed $c_0,\eta>0$, if
$\nu\ge c_0n$ and $\nu\le m\le n$ with $2\nu-m\ge\eta n$, then
\begin{equation}\label{ppt:low}
 \sum_{0<|A|\le m}\|T_A(z)\|_1
 \le(A_0n^{b_0}+1)2^{-\nu/2},
\end{equation}
where $A_0,b_0$ and the length cutoff depend only on $c_0,\eta$.

For the larger supports, define
\[
 K_i(x,y)\coloneqq\Tr(u_i(x)u_i(y)),\qquad
 B_A\coloneqq\sum_{i\in A}r_i^2.
\]
These kernels are real, have mean zero on $\F_p^2$, and satisfy
$|K_i(x,y)|\le(1+r_i)^2$.
For independent uniform $X,Y\in z+\mathcal G$,
\begin{equation}\label{ppt:gram}
 \|M_A(z)\|_2^2=\E_{X,Y}\prod_{i\in A}K_i(X_i,Y_i).
\end{equation}
Assume for now $n/3<\nu\le n/2$ and $|A|\ge\nu+1$.
With $j_0\coloneqq\lfloor(3\nu-n)/2\rfloor$ and $\gamma_*=(3/4)^{1/8}$,
the derivation of \eqref{bh:small-bias}, using
Lemmas~\ref{bh:real-gap} and~\ref{bh:product-mds} for $K_i/(1+r_i)^2$
on $\F_p^2$ and equations \eqref{ppt:tensor-bound}--\eqref{ppt:gram}, gives
\begin{equation}\label{ppt:small-bias}
 \|T_A(z)\|_1\le\gamma_*^{j_0/2}\prod_{i\in A}(1+r_i)
 \le\exp\!\left(\frac{j_0}{2}\log\gamma_*+\sqrt{nB_A}\right).
\end{equation}
For the complementary bias estimate, the calculation of Supplementary
Section~\ref{supp:bh-bias} applies with
$\|\sigma_i(x)-\bar\sigma_i\|_1\le\|u_i(x)\|_2\le1+r_i\le2$ from
\eqref{ppt:holder} and \eqref{ppt:variance}, uniformity within each
$\nu$-coordinate MDS marginal, and the variance bound in
\eqref{ppt:variance}. It gives
\begin{equation}\label{ppt:large-bias}
 \|T_A(z)\|_1
 \le\exp\!\left((n-2\nu)(\log2+\tfrac12)-\tfrac12B_A\right).
\end{equation}
If $r_i=1$, \eqref{ppt:variance} makes $u_i$ identically zero, so the
inequalities include that endpoint.

For $c_{\rm q}\le\nu/n\le1/2$, choose
$m\coloneqq\lfloor(1-10^{-8})n\rfloor$ and split the large supports at
$B_A=2^{-18}n$.
Then $2\nu-m\ge8\cdot10^{-9}n$.
Equations \eqref{ppt:small-bias} and~\eqref{ppt:large-bias} are exactly
the two scalar bounds used in Supplementary Section~\ref{supp:bh-bias},
so \eqref{bh:combined-bias} holds for these operators, and
Lemma~\ref{bh:large-support-sum}, whose proof uses only that bound, gives
\[
 \sum_{|A|>m}\|T_A(z)\|_1
 \le(n+1)e^{-10^{-6}n+\lambda_*/2},
 \qquad \lambda_*\coloneqq-\log\gamma_*.
\]
Combining this with \eqref{ppt:low} and the support expansion gives
\[
 \Adv(R_z,R_*)\le\frac12\left[
 (A_0n^{b_0}+1)2^{-\nu/2}+(n+1)e^{-10^{-6}n+\lambda_*/2}\right].
\]
Every constant and scalar input is the same as in the proof of
Theorem~\ref{bh:main}, so the same choices of $a,n_0$ give $e^{-an}$.
For larger $\nu$, restrict to the GRS subcode of dimension
$\nu_0\coloneqq\lceil c_{\rm q}n\rceil$.
For sufficiently large $n$ it lies in the direct range. Every original
coset is a uniform mixture of subcode cosets with the same $R_*$;
convexity completes the proof.
\end{proof}

\end{document}